\pdfoutput=1
\documentclass[a4paper,11pt]{article} 
\usepackage[utf8x]{inputenc} 
\usepackage[english]{babel}
\usepackage{amsmath,amsthm,amsfonts,amssymb,mathrsfs, bbm}
\usepackage{latexsym}
\usepackage{authblk}
\usepackage{lineno}
\usepackage{eucal}
\usepackage[dvipsnames]{xcolor}
\usepackage{hyperref}
\usepackage{afterpage}
\usepackage{lmodern}
\usepackage[a4paper,left=2.9cm,right=2.9cm]{geometry}
\usepackage{graphicx}

\definecolor{BeauBlue}{rgb}{0, 0.2, .9}
\definecolor{BeauOrange}{rgb}{.8, .1, 0}
\usepackage{hyperref}
\hypersetup{
	colorlinks = true,
	linkcolor = BeauBlue,
	urlcolor = cyan,
	citecolor = BeauOrange,
}

\numberwithin{equation}{section}

\newtheorem{theorem}{Theorem}[section] 
\newtheorem{proposition}[theorem]{Proposition} 
\newtheorem{corollary}[theorem]{Corollary}
\newtheorem{lemma}[theorem]{Lemma}

\newtheorem{remark}[theorem]{Remark}
\newtheorem{assumption}[theorem]{Assumption}

\DeclareMathOperator{\Tr}{Tr} 

\newcommand{\numberset}{\mathbb}
\newcommand{\N}{\numberset{N}} 
\newcommand{\Z}{\numberset{Z}} 
 
\newcommand{\R}{\numberset{R}}

\newcommand{\ul}{\underline}

\title{Large Scale Response of Gapless $1d$ and Quasi-$1d$ Systems}
\author[1]{Marcello Porta}
\author[1]{Harman Preet Singh}
\affil[1]{Mathematics Area, SISSA, Via Bonomea 265, 34136 Trieste, Italy}

\date{\today}

\begin{document}

\maketitle

\begin{abstract}
We consider the transport properties of non-interacting, gapless one-dimensional quantum systems and of the edge modes of two-dimensional topological insulators, in the presence of time-dependent perturbations. We prove the validity of Kubo formula, in the zero temperature and infinite volume limit, for a class of perturbations that are weak and slowly varying in space and in time, in an Euler-like scaling. The proof relies on the representation of the real time Duhamel series in imaginary time, which allows to prove its convergence uniformly in the scaling parameter and in the size of system, at low temperatures. Furthermore, it allows to exploit a suitable cancellation for the scaling limit of the model, related to the emergent anomalous chiral gauge symmetry of relativistic one-dimensional fermions. The cancellation implies that, as the temperature and the scaling parameter are sent to zero, the linear response is the only contribution to the full response of the system. The explicit form of the leading contribution to the response function is determined by lattice conservation laws. In particular, the method allows to prove the quantization of the edge conductance of $2d$ quantum Hall systems from quantum dynamics.
\end{abstract}

\tableofcontents

\section{Introduction}\label{sec:intro}

The linear response approximation is a widely used tool in condensed matter physics, allowing to determine the response of systems initially at equilibrium after introducing a weak and time-dependent perturbation. Let us informally describe the setting. We consider lattice fermions, on a lattice $\Gamma_{L} \subset \mathbb{Z}^{d}$ with side length $L$, in a grand-canonical picture. Let $\mathcal{H}$ be the Fock space Hamiltonian of the system, and let $\rho_{\beta,\mu, L}$ be the associated equilibrium Gibbs state, at inverse temperature $\beta$ and chemical potential $\mu$:
\begin{equation}
\rho_{\beta,\mu,L} = \frac{e^{-\beta (\mathcal{H} - \mu \mathcal{N})}}{\Tr e^{-\beta (\mathcal{H} - \mu \mathcal{N})}}\;,
\end{equation}
with $\mathcal{N}$ the number operator. We will always suppose that the Hamiltonian $\mathcal{H}$ is local, {\it i.e.} it is given by the sum of monomials in the fermionic creation and annihilation operators that only involve subsets $X$ of $\Gamma_{L}$ with diameter bounded uniformly in $L$. Suppose that, far in the past, the system is exposed to an external, local perturbation $\mathcal{P}$, that is turned on slowly in time; we describe this setting by introducing the time-dependent Hamiltonian:
\begin{equation}
\mathcal{H}(\eta t) = \mathcal{H} + \varepsilon e^{\eta t} \mathcal{P}\;,\qquad t\leq 0\;.
\end{equation}
The parameter $\eta > 0$ is called the adiabatic parameter, and it describes the rate of variation of the perturbation in time, while the parameter $\varepsilon$ describes the strength of the perturbation. In what follows, we will be interested in the regime in which $\beta$ and $L$ are large, and eventually will be taken to infinity. We will always suppose that $\varepsilon$ and $\eta$ are small, uniformly in $\beta , L$. 

The evolution of the state is defined by the Schr\"odinger-von Neumann equation:
\begin{equation}\label{eq:SvN}
i\partial_{t} \rho(t) = [ \mathcal{H}(\eta t), \rho(t) ]\;,\qquad \rho(-\infty) = \rho_{\beta,\mu, L}\;,
\end{equation}
and we will be interested in the variation of physical observables computed on the solution of (\ref{eq:SvN}):
\begin{equation}\label{eq:variation}
\Tr \mathcal{O}_{X} \rho(t) - \Tr \mathcal{O}_{X} \rho(-\infty)\;,
\end{equation}
with $\mathcal{O}_{X}$ a lattice observable spatially supported on $X\subset \Gamma_{L}$. In order to understand the effect of the perturbation on (\ref{eq:variation}), it is natural to try to set up a perturbative expansion, for $\varepsilon$ small. By Duhamel iteration one has, neglecting higher order terms in $\varepsilon$:
\begin{equation}\label{eq:kubointro}
\Tr \mathcal{O}_{X} \rho(0) - \Tr \mathcal{O}_{X} \rho(-\infty) = -i \varepsilon \int_{-\infty}^{0} dt\, e^{\eta t} \Tr \Big[ \mathcal{O}_{X}, e^{i\mathcal{H}t} \mathcal{P} e^{-i\mathcal{H}t} \Big] \rho_{\beta, \mu, L} + \text{h.o.t..}
\end{equation}
The explicit term in the right-hand side of (\ref{eq:kubointro}) is called Kubo formula \cite{Kubo}, and it describes the linear response of the system at time $t=0$.

The first natural problem to address is about the evaluation of Kubo formula, for specific models, and for interesting perturbations $\mathcal{P}$. This is the starting point of the theory of the integer quantum Hall effect (IQHE), where Kubo formula is used to describe the variation of the current after introducing a weak external electric field, for gapped fermionic models at zero temperature. Already proving that Kubo formula gives a finite result as $\eta \to 0$ is nontrivial: in the case of the IQHE, this is ultimately due to the presence of a spectral gap, which makes the time integral convergent uniformly in $\eta$, thanks to oscillations. One can actually prove much more than just finiteness of Kubo formula: the transverse conductivity of the system, at zero temperature and for $L\to \infty$, is quantized in multiples of $1/2\pi$ (in units such that $e^{2} = \hbar = 1$, with $e$ the electric charge and $\hbar$ the reduced Planck's constant), provided the chemical potential $\mu$ is chosen in a spectral gap of $\mathcal{H}$. The mathematically rigorous literature on the subject is by now very large, and covers a large class of models, see {\it e.g.} \cite{Graf, AW} for mathematical reviews. Recently, the proof of quantization of the Hall conductivity has been extended to gapped many-body systems, \cite{HM, GMPhall, BBDRF}, see \cite{GMPrev} for a recent review.

More generally, Kubo formula is widely used to predict the behavior of semi-metallic and metallic systems as well, where no spectral gap is available. This is the case of graphene, for instance, where the longitudinal conductivity defined from Kubo formula is equal to $1/4$. This remarkable phenomenon has been observed in \cite{Na}. A first explanation starting from a non-interacting model has been given in \cite{SPG}; more recently, the universality of the longitudinal conductivity of graphene has been proved in \cite{GMPcond}, starting from weakly interacting fermions on the honeycomb lattice. In recent times, rigorous results concerning the Kubo formula for a wide class of interacting gapless systems have been obtained; specifically, one-dimensional metals \cite{BM1, BM2, MPdrude, CMP}, two-dimensional Hall systems at criticality \cite{GJMP, GMPhald} and edge modes of $2d$ systems \cite{AMP, MPspin, MPmulti}; three dimensional Weyl semimetals \cite{GMPweyl}. See \cite{GMPrev} for an overview.

The next natural question is about the validity of Kubo formula, in the sense of rigorously proving that higher order terms in (\ref{eq:kubointro}) are indeed subleading, uniformly in $\beta$ and in $L$. For gapped systems at zero temperature, the validity of Kubo formula can be explained using the adiabatic theorem. This has been done in \cite{ASY} for noninteracting systems on the lattice, in \cite{ES} for the Landau Hamiltonian and more recently in \cite{BDRF,MT,Teu} for interacting lattice models. See \cite{HT} for a review and further references. In the case of the quantum Hall effect, one can actually show that linear response is exact, for non-interacting \cite{KS, MM} and also for interacting system \cite{BDRFL}: all power-law corrections after the linear response vanish. Recently, the derivation of Kubo formula for the quantum Hall effect has been extended to the case of $2d$ disordered, non-interacting systems, where the spectral gap is replaced by a mobility gap \cite{DREF}. A different proof of the validity of Kubo formula, which also applies to small positive temperatures, has been given in \cite{GLMP}, for weakly interacting fermions, using cluster expansion methods. There, the higher order terms in the Duhamel expansion are controlled using an extension of the argument behind the stability of KMS states \cite{BR2}; thanks to a suitable complex deformation argument, all terms in the Duhamel expansion can be represented in terms of imaginary time correlations, whose good decay properties for gapped systems can be proven using cluster expansion techniques. This method can also be used to prove the validity of the many-body adiabatic theorem \cite{GLMP}.

It is an open problem to understand the validity of Kubo formula in a metallic phase starting from a microscopic model. This is of obvious physical relevance, in view of the ubiquitous applications of linear response in cases in which no spectral gap is present. In the case of a finite system coupled to external leads at thermal equilibrium, the related problem of proving convergence to a non-equilibrium steady state (and proving the validity of Kubo formula) has been widely studied in the literature. Rigorous results have been proved under suitable spectral assumptions, for non-interacting \cite{AH, AP, Ne} and also interacting \cite{FMU, JOP, MCP, CMP2} systems.

In this paper we shall address the validity of Kubo formula for closed, extensive systems, in the presence of time-dependent perturbations. In order to prove the validity of Kubo formula for metallic systems, one has to face two main technical problems: the first is about the control of the time integration, that for gapped systems was ultimately possible thanks to the oscillations introduced by the spectral gap; the second is about the slow space decay of correlations, needed in order to control expressions of the type, setting $\tau_{t}(\mathcal P):=e^{i\mathcal H t}\mathcal P e^{-i\mathcal H t}$:
\begin{equation}
\sum_{\{X_{i}\}}\Tr \Big[ \Big[ \cdots \Big[ \Big[ \mathcal{O}_{X}, \tau_{t_{1}}(\mathcal{P}_{X_{1}}) \Big], \tau_{t_{2}}(\mathcal{P}_{X_{2}})\Big], \cdots \Big], \tau_{t_{n}}(\mathcal{P}_{X_{n}})\Big] \rho_{\beta, \mu, L}\;,
\end{equation}
which appear naturally in the Duhamel expansion of (\ref{eq:kubointro}), after representing the perturbation $\mathcal{P}$ as a sum of local term, $\mathcal{P} = \sum_{Y} \mathcal{P}_{Y}$. Due to the absence of a spectral gap, one cannot expect the correlation functions of the system to be integrable uniformly in $t$, $\beta$, $L$.

In this paper we focus on the simplest, yet nontrivial, gapless cases relevant for condensed matter applications. We shall consider non-interacting Fermi gases in one dimension, and edge modes for two-dimensional topological insulators, that can be viewed as quasi-$1d$ metallic systems. For the sake of the introduction, let us focus on one-dimensional systems. We shall denote by $\Gamma_{L}$ a one-dimensional lattice with length $L$ and periodic boundary conditions. The Hamiltonian $\mathcal{H}$ has the form:
\begin{equation}\label{eq:17}
\mathcal{H} = \sum_{x,y \in \Gamma_{L}}\sum_{\rho,\rho' = 1}^{M} a^{*}_{x,\rho} H_{\rho\rho'}(x;y) a_{y,\rho'}\;,
\end{equation}
where $\rho = 1,\ldots, M$ labels internal degrees of freedom, $a^{*}_{x,\rho}, a_{x,\rho}$ are the fermionic creation and annihilation operators, and $H$ is a single-particle Hamiltonian, which we assume to be translation-invariant: $H_{\rho\rho'}(x;y) \equiv H_{\rho\rho'}(x-y)$. This allows to introduce the Bloch Hamiltonian $\hat H(k)$, $k\in \mathbb{T}$, an $M\times M$ Hermitian matrix, whose eigenvalues as a function of $k$ describe the energy bands. We will suppose that the energy bands cross the Fermi level $\mu$ at $N$ distinct points, called the Fermi points, $k_{F}^{\omega}$, $\omega = 1,\ldots, N$. We will assume that the intersections are transversal, which is generically true; this allows to approximate the dispersion relation in proximity of the Fermi level $\mu$ by an asymptotically linear law, with slope $v_{\omega}$, also called the Fermi velocity. 

Let us now describe the class of perturbations we will consider. Let $\mu(x)$ be a smooth and fast-decaying function. We shall consider the following type of perturbations:
\begin{equation}\label{eq:pert}
\varepsilon e^{\eta t}\mathcal{P} = \theta e^{\eta t}\sum_{x\in \Gamma_{L}} \sum_{\rho = 1}^{M} \mu(\theta x) a^{*}_{x,\rho} a_{x,\rho}
\end{equation}
where $0<\theta \ll 1$ and $0<\eta \ll 1$, independent of $\beta, L$. From (\ref{eq:pert}), we see that the perturbation lives in a spatial region with diameter of order $1/\theta$, and its strength is $\theta$. We will be interested in the following order of limits: first $L\to \infty$; then $\beta \to \infty$; and then $\eta, \theta \to 0$. More precisely, we will consider an Euler-like scaling, in which $\theta = a \eta$, with $a\leq \left|\log \eta\right|$. This leaves room to have $\eta \gg \theta$ or $\eta \ll \theta$; the two regimes will give rise to different outcomes, as we will see. The Euler scaling is the natural regime in which one expects that the long-time behavior of large class of integrable systems is effectively described by hydrodynamic models, see \cite{Doyon, Ess} for recent reviews. Proving that this is indeed the case is a hard mathematical problem, which partially motivates the present work.

In order to clarify the picture, let us give some heuristics on the large-scale properties of the model. Let us consider the Euclidean two-point function, defined as:
\begin{equation}\label{eq:2ptintro}
\langle {\bf T}\, \gamma_{x_{0}} (a_{x,\rho}) \gamma_{y_{0}}(a^{*}_{y,\rho'}) \rangle_{\beta, L}\;,
\end{equation}
with $0\leq x_{0}, y_{0} < \beta$, $\gamma_{t}(A) = e^{t(\mathcal{H} - \mu\mathcal{N})} A  e^{-t(\mathcal{H} - \mu\mathcal{N})}$ the imaginary time evolution, and ${\bf T}$ the fermionic time-ordering (defined below in (\ref{eq:time})). Let $\ul x = (x_{0}, x)$ and $\ul y = (y_{0}, y)$. Then, for $|\ul x - \ul y|  = O(1/\theta)$ and $\theta \ll 1$, the zero temperature and infinite volume two-point function is given by, up to faster decaying terms:
\begin{equation}
\langle {\bf T}\, \gamma_{x_{0}} (a_{x,\rho}) \gamma_{y_{0}}(a^{*}_{y,\rho'}) \rangle_{\infty} \simeq \sum_{\omega = 1}^{N} e^{ik_{F}^{\omega}(x - y)} P_{\rho\rho'}^{\omega} \int \frac{d{\ul q}}{(2\pi)^{2}}\, \frac{e^{i\ul q\cdot (\ul x - \ul y)}}{iq_{0} + v_{\omega} q_{1}} \chi(|\ul q| / \theta)\;,
\end{equation}
where $P^{\omega}$ is the Bloch projector associated with the energy band crossing the Fermi level at $k_{F}^{\omega}$, and $\chi(\cdot)$ is a smooth cutoff function, equal to $1$ if the argument is less than $1$ and equal to zero if the argument is larger than $2$. That is, the two-point function is a superpositions of ultraviolet-regularized two-point functions of $1+1$ dimensional chiral fermions, with propagators given by $(\partial_{0} - i v_{\omega} \partial_{1})^{-1}$.

Thus, it makes sense to define $\underline{X} = \theta \underline{x}$ and $\underline{Y} = \theta \underline{y}$. In terms of these rescaled variables,
\begin{equation}\label{eq:gresc}
\langle {\bf T}\, \gamma_{X_{0} / \theta} (a_{X / \theta,\rho}) \gamma_{Y_{0} / \theta}(a^{*}_{Y / \theta,\rho'}) \rangle_{\infty} \simeq \theta \sum_{\omega = 1}^{N} e^{ik_{F}^{\omega}(X - Y) / \theta} P_{\rho\rho'}^{\omega} \int \frac{d{\ul q}}{(2\pi)^{2}}\, \frac{e^{i\ul q\cdot (\ul X - \ul Y)}}{iq_{0} + v_{\omega} q_{1}} \chi(|\ul q|)\;;
\end{equation}
that is, up to a fast oscillating phase, the two-point function decays as $\theta \times | \ul X - \ul Y |^{-1}$. Concerning the perturbation (\ref{eq:pert}), it can be rewritten as, defining $T = \theta t$, and recalling that $\eta = (1/a)\theta$:
\begin{equation}
\varepsilon e^{\eta t}\mathcal{P} = \theta e^{(1/a) T}\sum_{X\in \theta \Gamma_{L}}\sum_{\rho=1}^{M} \mu(X) a^{*}_{X / \theta,\rho} a_{X / \theta,\rho}\;;
\end{equation}
hence, defining $\widetilde\rho(T) := \rho(T / \theta)$ with $\rho(\cdot)$ the solution of (\ref{eq:SvN}), we get:
\begin{equation}\label{eq:tilderho}
i\partial_{T} \widetilde\rho(T) =  \Big[ \widetilde{\mathcal{H}}((1/a) T), \widetilde \rho(T) \Big]\;,
\end{equation}
where:
\begin{equation}\label{eq:rescH}
\widetilde{\mathcal{H}}((1/a)T) = \sum_{X,Y}\sum_{\rho,\rho'} a^{*}_{X/\theta,\rho} \widetilde{H}_{\rho\rho'}(X-Y) a_{Y/\theta,\rho'} + e^{(1/a) T}\sum_{X\in \theta \Gamma_{L}}\sum_{\rho} \mu(X) a^{*}_{X / \theta,\rho} a_{X / \theta,\rho}
\end{equation}
and $\widetilde{H}(X-Y) = (1/\theta) H((X-Y)/\theta)$. The leading behavior of the two-point function of the rescaled, unperturbed Hamiltonian in (\ref{eq:rescH}), computed at imaginary times $X_{0}$, $Y_{0}$, is given by (\ref{eq:gresc}). Thus, in order to understand the contribution of the different chiralities to the solution of (\ref{eq:tilderho}), as $\theta \to 0$, one is tempted to consider the dynamics of the following time-dependent, continuum model:
\begin{equation}\label{eq:cont1}
\mathcal{H}_{\omega}((1/a)T) = \int_{\mathbb{R}} dX\,\Psi^{*}_{X} v_{\omega} i\partial_{X} \Psi_{X} + e^{(1/a) T} \int_{\mathbb{R}} dX\, \mu(X) \Psi^{*}_{X} \Psi_{X}\;,
\end{equation}
where $\Psi_{X}$ is now a fermionic field, understood as an operator-valued distribution. It is reasonable to expect that, as $\theta \to 0$, the leading contribution to the transport properties associated with (\ref{eq:tilderho}) can be represented as a superposition of terms that evolve according to the dynamics generated by (\ref{eq:cont1}). Eq. (\ref{eq:cont1}) can be viewed as describing an adiabatic perturbation of relativistic $1+1$ dimensional chiral fermions. Of course, the expression (\ref{eq:cont1}) is formal, in view of the unboundedness from above and from below of the dispersion relation. Also, observe that, since we are making no assumption on the size of $\mu(\cdot)$, after rescaling time no small parameter is left in front of the perturbation.

Putting aside the issue of mathematical rigor, the advantage of this representation is that the model (\ref{eq:cont1}) can be studied via bosonization: the density operator $\Psi^{*}_{X} \Psi_{X}$ is known to behave like a free bosonic field, see {\it e.g.} \cite{ML, FGM, Gia}. Since the Duhamel series associated with (\ref{eq:cont1}) is expressed in terms of correlation function of density operators, this observation is expected to introduce a drastic simplification in the perturbation theory for the lattice model: for quasi-free bosonic systems, all connected correlation functions of order higher than two are zero. This property is expected to hold also in presence of weak density-density interactions, and it is an instance of the integrability of the emergent QFT. In the last years, there has been enormous interest in the application of methods from the theory of integrable systems to study non-equilibrium quantum dynamics, starting from \cite{OCDY, BCNF}. This led to the derivation of generalized hydrodynamic models for the long-time dynamics of quantum systems; see {\it e.g.} \cite{Doyon, Ess} for reviews.

In this paper, we will rigorously study the transport properties of the family of lattice models described after (\ref{eq:17}), for perturbations in the Euler-like scaling regime, recall (\ref{eq:pert}). Concerning the observable $\mathcal{O}_{X}$ in (\ref{eq:variation}), we will restrict the attention to the density operator and to the current density operator. They are related by a lattice continuity equation, and this will play a crucial role in our analysis. Due to the absence of a spectral gap, the validity of Kubo formula cannot be inferred from the adiabatic theorem, since the latter is not available in a quantitative form suitable for our purposes. Instead, our method relies on the application of \cite{GLMP} to the case of gapless systems. This approach, based on a complex deformation argument in time, allows us to obtain an exact representation for the real time Duhamel series in terms of Euclidean correlation functions, where the real time evolution is replaced by an imaginary time dynamics, as in (\ref{eq:2ptintro}). This allows to translate the oscillations introduced by the real time dynamics into (slow) decay properties of Euclidean correlations, which can be studied in a much more efficient way. 

Furthermore, this representation is particularly convenient for comparing the perturbation theory of the original lattice model with the perturbation theory of the effective continuum model (\ref{eq:cont1}), in the presence of a suitable ultraviolet cutoff. This allows in particular to detect a cancellation in the perturbation theory for the lattice model, at every order, which is related to the bosonic behavior of the density operators of the effective model. The proof of this cancellation however does not rely on bosonization, but rather on a careful analysis of the regularized loop integrals arising in the perturbation theory of the lattice model. This type of cancellation is well-known in the QFT literature, see \cite{FGM2}. It can also be viewed as a consequence of the emergent, anomalous gauge symmetry of the $1+1$ dimensional relativistic chiral fermions. In our framework, this cancellation is used up to a high, $\theta$-dependent order; to control the higher orders, we show that the Duhamel series is actually convergent, and this allows to show that very high orders do not matter as $\theta \to 0$. All together, we show that the linear response term dominates, up to terms that are subleading as $\theta \to 0$: that is, in the scaling limit the linear response is exact. Concerning the linear response term, it can be computed explicitly in terms of the Fermi velocities, following the strategy already used in \cite{BM1, BM2, MPdrude, AMP, MPmulti}, based on lattice conservation laws. 

Then, we adapt the strategy to study edge transport for $2d$ topological insulators. The complementary question of understanding bulk transport in presence of a bulk spectral gap has been studied in \cite{HT2}, where the validity of a many-body (super-)adiabatic theorem has been established. More recently, the validity of Kubo formula for many-body system under a local gap assumption has been proved in \cite{HW}. Here, we focus on the edge transport properties, which are effectively described by a massless quasi-$1d$ system, and we prove the validity of Kubo formula for the edge current and for the edge density. This expression is the starting point of the analyses of \cite{AMP, MPmulti}, which also considered interacting systems. In particular, the explicit computation of the edge linear response, combined with the subleading estimate for the sum of all the higher order terms, allows us to prove the quantization of the edge conductance from quantum dynamics.

In perspective, we believe that the strategy introduced in the present paper could be extended to consider non-integrable, weakly interacting gapless systems on the lattice. To do so, the analysis of Euclidean correlation functions performed in the present paper has to be carried out using rigorous renormalization group methods. In the last years, these methods allowed to put on rigorous grounds several predictions based on formal bosonization arguments for $1d$ and quasi-$1d$ systems; see {\it e.g.} \cite{BMchiral, BFM1, BFM2, Mabook, AMP, MPmulti}. We plan to come back to this interesting extension in the near future.

The paper is organized as follows. In Section \ref{sec:1d} we introduce the setting for one-dimensional fermions; we introduce the class of models we consider, we define the current and density operators, and we recall the rigorous Wick rotation of \cite{GLMP}. In Section \ref{sec:main1d} we state our main result for one-dimensional fermions, Theorem \ref{thm:main1d}, that gives an explicit expression for the full response of the density and of the current operator on large scales. In Section \ref{sec:proof1d} we discuss the proof of Theorem \ref{thm:main1d}. Finally, in Section \ref{sec:2d} we discuss the extension of Theorem \ref{thm:main1d} to the case of edge transport in two-dimensional topological insulators. Our main result here is Theorem \ref{thm:main2d}, that gives the explicit expression of the large-scale response of the system, and in particular proves the quantization of the edge conductance.

\paragraph{Acknowledgements.} M. P. and H. P. S. acknowledge support by the European Research Council through the ERC-StG MaMBoQ, n. 802901. M. P. acknowledges support from the MUR, PRIN 2022 project MaIQuFi cod. 20223J85K3. This work has been carried out under the auspices of the GNFM of INdAM. We thank the anonymous referees for valuable comments.

\paragraph{Data availability statement.} This manuscript has no associated data.

\paragraph{Conflict of interest.} The authors have no conflict of interest to declare that are relevant to the content of this article.

\section{Lattice fermions}\label{sec:1d}

\subsection{Hamiltonian and Gibbs state}

Let $L \in 2\mathbb{N}+1$. We consider fermions on a one-dimensional lattice 
\[
\Gamma_{L} = \left[-\left\lfloor\frac{L}{2}\right\rfloor,\left\lfloor\frac{L}{2}\right\rfloor\right]\cap \Z
\]
endowed with periodic boundary conditions. We shall allow for internal degrees of freedom, {\it e.g.} the spin, and we shall denote by $S_{M} = \{ 1,\ldots, M \}$ the set of their labels. We shall use the notation $\Lambda_{L} = \Gamma_{L} \times S_{M}$, which we can view as a decorated $1d$ lattice. Given $x\in \Gamma_{L}$ and $\rho \in S_{M}$, we shall denote by ${\bf x} = (x,\rho)$ the corresponding point in $\Lambda_{L}$. We will use the following distance on $\Gamma_{L}$:
\begin{equation}
| x- y |_{L}^{2} = \min_{n\in \mathbb{Z}} | x - y + n L |^{2}\;.
\end{equation}
Let $H$ be a single-particle Hamiltonian on $\Lambda_{L}$. That is, $H$ is a self-adjoint matrix, with entries given by $H({\bf x}; {\bf y}) \equiv H_{\rho\rho'}(x;y)$. We shall suppose that the Hamiltonian is finite-ranged. Also, we shall assume that the Hamiltonian $H$ is the periodization of an Hamiltonian $H^{\infty}$ defined over $\ell^{2}(\mathbb{Z} \times S_{M})$:
\begin{equation}
H_{\rho\rho'}(x; y) = \sum_{n \in \mathbb{Z}} H^{\infty}_{\rho\rho'}( x + nL; y )\;.
\end{equation}
Furthermore, we assume that $H^{\infty}$, and hence $H$, are translation-invariant:
\begin{equation}\label{eq:periodiz}
H_{\rho\rho'}^{\infty}( x; y ) = H_{\rho\rho'}^{\infty}( x + z; y + z)\qquad \text{for all $z\in \mathbb{Z}$}.
\end{equation}
A natural example is the lattice Laplacian, for spinless fermions ($M=1$):
\begin{equation}
(H^{\infty} \psi)(x) = t(\psi_{x-1} + \psi_{x+1} - 2\psi_{x})\;,
\end{equation}
with $t \in \mathbb{R}$ the hopping parameter.

Translation-invariance allows to coveniently express the Hamiltonian in momentum space. To this end, let us introduce the set of allowed quasi-momenta:
\begin{equation}\label{eq:BLdef}
B_{L} = \Big\{ k = \frac{2\pi}{L} n \, \Big| \, 0\leq n \leq L-1 \Big\}\;.
\end{equation}
The set $B_{L}$ is the (discretized) Brillouin zone. Given a function $f$ on $\Gamma_{L}$, we define its Bloch transform as:
\begin{equation}\label{eq:fk}
\hat f(k) = \sum_{x\in \Gamma_{L}} e^{-ikx} f(x)\qquad \text{for all $k\in B_{L}$.}
\end{equation}
This relation can be inverted, as:
\begin{equation}
f(x) = \frac{1}{L} \sum_{k\in B_{L}} e^{ikx} \hat f(k)\;.
\end{equation}
We will always consider functions that satisfy periodic boundary conditions on $B_{L}$. Thus, it is natural to introduce the following norm in momentum space:
\begin{equation}
| k |_{\mathbb{T}} = \min_{n\in \mathbb{Z}} | k + 2\pi n |\qquad \text{for all $k\in B_{L}$.}
\end{equation}
Next, we define the Bloch Hamiltonian $\hat H(k)$ as the matrix with entries:
\begin{equation}\label{eq:bloch1d}
\hat H_{\rho\rho'}(k) = \sum_{x\in \Gamma_{L}} e^{-ikx} H_{\rho\rho'}(x; 0)\;.
\end{equation}
That is, $\hat H(k)$ is an $M\times M$ Hermitian matrix. By (\ref{eq:periodiz}), it is given by the restriction of $\hat H^{\infty}(k)$, defined on the circle $\mathbb{T}$ of length $2\pi$, to $B_{L}$. We shall make the following assumptions on the spectrum of the Bloch Hamiltonian and on the chemical potential $\mu \in \mathbb{R}$. 

\begin{assumption}[Low energy spectrum.]\label{ass:A} There exists $\Delta > 0$ such that the following is true.
\begin{itemize}
\item[(i)] There exists $N \in \mathbb{N}$, disjoint sets $I_{\omega} \subset \mathbb T$ labelled by $\omega = 1, \ldots, N$, and strictly monotone, smooth functions $e_{\omega}: I_{\omega} \to \mathbb{R}$ such that:
\begin{equation}
\sigma(H^{\infty}) \cap (\mu - \Delta, \mu + \Delta) = \bigcup_{\omega = 1}^{N} \mathrm{Ran}(e_{\omega})\;.
\label{eq:spectrum}
\end{equation}
\item[(ii)] We introduce the $\omega$-Fermi point $k_{F}^{\omega} \in I_{\omega}$ and the $\omega$-Fermi velocity $v_{\omega}$ as
\begin{equation}
e_{\omega}(k_{F}^{\omega}) = \mu\;,\qquad v_{\omega} = \partial_{k} e_{\omega}(k_{F}^{\omega})\;. 
\end{equation}
Notice that $v_{\omega} \neq 0$, by the strict monotonicity of $e_{\omega}$.
\item[(iii)] For any $\omega=1,\dots, N$, and $k\in I_{\omega}$, $e_{\omega}(k)$ is a non-degenerate eigenvalue of $\hat H^{\infty}(k)$.
\end{itemize}
\end{assumption}
\begin{remark}
\begin{itemize}
\item[(i)] The simplest example of Hamiltonian that fits the setting is the Laplacian on $\mathbb{Z}$, whose energy band is:
\begin{equation}
\varepsilon(k) = 2t(\cos(k) - 1)\;,
\end{equation}
and where $N=2$.
\item[(ii)] In general, it is a well-known fact that, for short-ranged one-dimensional lattice models, the number of Fermi points $N$ has to be even, with net chirality equal to zero: $\sum_{\omega =1 }^{N} v_{\omega} / |v_{\omega}| = 0$.
\item[(iii)] The above setting is generic, and it covers a large class of translation-invariant, $1d$ quantum systems. We are, however, ruling out Fermi points with zero velocity, and degenerate Fermi points. We believe that the latter restriction could be avoided, by a technical but straightforward extension of the analysis carried out in the present paper. Instead, we think that the presence of Fermi points with zero velocity would have a dramatic effect in our analysis, and we do not know whether the our results would extend to that case.
\end{itemize}
\end{remark}
We will describe the many-particle system in a grand-canonical setting, in the fermionic Fock space. The fermionic Fock space associated with the model on a finite lattice is:
\begin{equation}
\mathcal{F} = \mathbb{C} \oplus \bigoplus_{n\geq 1} \ell^{2}_{\text{a}}(\Lambda_{L}^{n})\;,
\end{equation}
with $\ell_{\text{a}}^{2}(\Lambda_{L}^{n})$ the set of antisymmetric, square summable functions on $\Lambda_{L}$. We shall use the notation ${\bf x} = (x, \rho)$ for points in $\Lambda_{L}$, with $x\in \Gamma_{L}$ and $\rho \in S_{M}$, and we shall denote by $a^{*}_{{\bf x}}, a_{{\bf x}}$ the usual fermionic creation and annihilation operators, satisfying the canonical anticommutation relations:
\begin{equation}
\{ a_{{\bf x}}, a^{*}_{{\bf y}} \} = \delta_{{\bf x}, {\bf y}}\;,\qquad \{ a_{{\bf x}}, a_{{\bf y}} \} = \{ a^{*}_{{\bf x}}, a^{*}_{{\bf y}} \} = 0\;.
\end{equation}
In terms of these objects, we shall lift the Hamiltonian $H$ to the Fock space as:
\begin{equation}\label{eq:H1d}
\mathcal{H} = \sum_{{\bf x}, {\bf y} \in \Lambda_{L}} a^{*}_{{\bf x}} H({\bf x}; {\bf y}) a_{{\bf y}}\;.
\end{equation}
In the following, it will also be convenient to rewrite the Hamiltonian in momentum space. For any $k \in B_{L}$, we define the Bloch transform of the fermionic operators as:
\begin{equation}\label{eq:foua}
\hat a_{(k, \rho)} = \sum_{x\in \Gamma_{L}} a_{(x, \rho)} e^{-ikx}\;,\qquad \hat a^{*}_{(k, \rho)} = \sum_{x\in \Gamma_{L}} a^{*}_{(x,\rho)} e^{ikx}\;,
\end{equation}
and we shall also set $\hat a^{\sharp}_{{\bf k}} \equiv \hat a^{\sharp}_{(k,\rho)}$ with ${\bf k} = (k, \rho)$. Eqs. (\ref{eq:foua}) can be inverted as, for $x\in \Gamma_{L}$:
\begin{equation}\label{eq:foua2}
a_{(x,\rho)} = \frac{1}{L} \sum_{k \in B_{L}} e^{ikx} \hat a_{(k, \rho)}\;,\qquad a^{*}_{(x,\rho)} = \frac{1}{L} \sum_{k \in B_{L}} e^{-ikx} \hat a^{*}_{(k, \rho)}\;.
\end{equation}
In terms of the momentum-space operators, the Hamiltonian can be rewritten as:
\begin{equation}
\mathcal{H} = \frac{1}{L}\sum_{k \in B_{L}} \sum_{\rho,\rho'\in S_{M}} \hat a^{*}_{(k, \rho)} \hat H_{\rho\rho'}(k) \hat a_{(k, \rho')}\;.
\end{equation}
Finally, the Gibbs state of the system at inverse temperature $\beta>0$ is:
\begin{equation}
\langle \mathcal{O} \rangle_{\beta,L} = \Tr \mathcal{O} \rho_{\beta, \mu, L}\;,\qquad \rho_{\beta, \mu, L} = \frac{e^{-\beta (\mathcal{H} - \mu \mathcal{N})}}{\Tr e^{-\beta (\mathcal{H} - \mu \mathcal{N})}}\;,
\end{equation}
with $\mu \in \mathbb{R}$ the chemical potential and $\mathcal{N}$ the number operator. We shall assume that Assumption \ref{ass:A} holds, for our choice of the chemical potential $\mu$.

\subsection{Dynamics and linear response}



\paragraph{Perturbing the system.} We will be interested in the response properties of the system, after exposing it to a time-dependent and slowly varying perturbation. We shall consider time-dependent Hamiltonians of the form:
\begin{equation}
\mathcal{H}(\eta t) = \mathcal{H} + e^{\eta t} \mathcal{P}\;,
\end{equation}
for $\eta > 0$ and $t\leq 0$. We shall consider finite-ranged perturbations:
\begin{equation}
\mathcal{P} = \sum_{X\subseteq \Lambda_{L}} \Phi_{X}
\end{equation}
with $\Phi_{X} = 0$ if $|X| > R$, with $R>0$ independent of $L$. Later, we will make a more specific choice of the perturbation. The time variation of the perturbation is the usual adiabatic switching, with the explicit choice of the exponential switch function, widely used in applications. 

Let $t,s\leq 0$, and let $\mathcal{U}(t;s)$ be the two-parameter unitary group generated by $\mathcal{H}(\eta t)$:
\begin{equation}
i\partial_{t} \mathcal{U}(t;s) = \mathcal{H}(\eta t) \mathcal{U}(t;s),\qquad \mathcal{U}(s;s) = \mathbbm{1}\;.
\end{equation}
Let us consider the evolution of the Gibbs state $\rho_{\beta, \mu, L}$,
\begin{equation}
\rho(t) = \lim_{T\to +\infty} \mathcal{U}(t;-T) \rho_{\beta, \mu, L} \mathcal{U}(t;-T)^{*}.
\end{equation}
We will be interested in the variation of physical observables,
\begin{equation}\label{eq:diffO}
\Tr \mathcal{O} \rho(t) - \Tr \mathcal{O} \rho_{\beta, \mu, L}\;,
\end{equation}
and on the dependence on the external perturbation. We will be interested in the following order of limits: first $L\to \infty$, then $\beta \to \infty$ and then $\eta \to 0^{+}$. For small perturbations, the dynamics of the system can be studied via the Duhamel series around the initial equilibrium state.

\begin{proposition}[Duhamel series]\label{prop:duha} Let $\tau_{t}(\mathcal{A})$ be the Heisenberg evolution of the observable $\mathcal{A}$:
\begin{equation}
\tau_{t}(\mathcal{A}) = e^{i\mathcal{H} t} \mathcal{A} e^{-i\mathcal{H} t}\;.
\end{equation}
The following identity holds true:
\begin{equation}\label{eq:duhamel}
\begin{split}
\Tr \mathcal{O} \rho(t) -  \langle \mathcal{O} \rangle_{\beta,L} &= \sum_{n=1}^{\infty} (-i)^{n} \int_{-\infty \leq s_{n} \leq \ldots \leq s_{1} \leq t} d \underline{s}\, e^{\eta(s_{1} + \ldots +s_{n})} \\&\quad \cdot \langle [ \cdots [[ \tau_{t}(\mathcal{O}), \tau_{s_{1}}(\mathcal{P})], \tau_{s_{2}}(\mathcal{P})] \cdots \tau_{s_{n}}(\mathcal{P}) ] \rangle_{\beta,L}\;.
\end{split}
\end{equation}
For $\eta > 0$, the sum in the right-hand side is absolutely convergent.
\end{proposition}
\begin{proof} The proof is standard. We refer the reader to {\it e.g.} \cite{GLMP}.
\end{proof}
The linear response approximation amounts to the truncation to first order of the Duhamel series for the non-autonomous dynamics generated by $\mathcal{H}(\eta t)$. We have:
\begin{equation}\label{eq:diffO2}
\Tr \mathcal{O} \rho(t) - \Tr \mathcal{O} \rho_{\beta, \mu, L} = -i \int_{-\infty}^{t} ds\, e^{\eta s} \Tr \big[ \tau_{t}(\mathcal{O}), \tau_{s}(\mathcal{P}) \big] \rho_{\beta, \mu, L} + \text{h.o.t.}.
\end{equation}
In many applications it is tacitly assumed that the linear response approximation is valid; this is very often motivated by the quantitative agreement with experimental observations. From a mathematically rigorous viewpoint, however, proving that indeed linear response dominates the expansion is a nontrivial task, particularly so in gapless cases as the one considered here. 

We shall consider perturbations of the following type:
\begin{equation}\label{eq:Pdef}
\mathcal{P} =  \theta \sum_{{\bf x} \in \Lambda_{L}} \mu(\theta x) a^{*}_{{\bf x}} a_{{\bf x}}\;;
\end{equation}
that is, we shall consider the second-quantization of a time-dependent local potential (see Figure \ref{fig:perturb1d}), slowly varying in space, on a length scale $1/\theta$, for $0<\theta \ll 1$. Concerning the function $\mu(\theta x)$, we assume it is the periodization of a smooth, compactly supported function $\mu_{\infty}(\theta x)$ on $\mathbb{R}$:
\begin{equation}\label{eq:period}
\mu(\theta x) = \sum_{a\in \mathbb{Z}} \mu_{\infty}(\theta (x + a L))\;.
\end{equation}
In particular, we can rewrite the function $\mu(\theta x)$ in terms of the Fourier transform of $\mu_{\infty}$ as:
\begin{equation}\label{eq:muFou}
\mu(\theta x) = \frac{1}{L} \sum_{p \in B_{L}} e^{ipx} \hat \mu_{\theta}(p)\;,\qquad \hat \mu_{\theta}(p) = \sum_{n \in \mathbb{Z}} \hat \mu_{\infty}((p + 2\pi n)/\theta)\;.
\end{equation}
\begin{figure}
    \centering
    \includegraphics[scale=0.75]{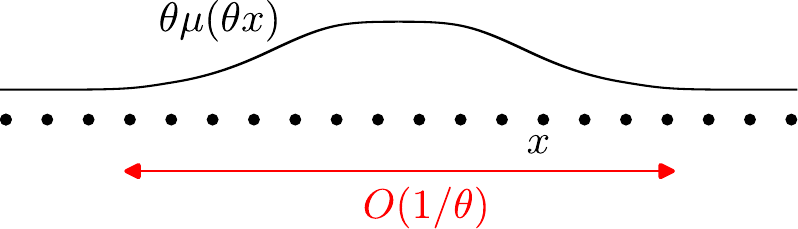}
    \caption{Qualitative representation of the rescaled local potential entering the definition of $\mathcal P$, Eq. (\ref{eq:Pdef}).}
    \label{fig:perturb1d}
\end{figure}
\paragraph{Density and current operators.} Let:
\begin{equation}
n_{{\bf x}} = a^{*}_{{\bf x}} a_{{\bf x}}\;,\qquad n_{x} = \sum_{\rho \in S_{M}} n_{(x, \rho)}\;.
\end{equation}
In order to introduce the current operator, let us consider the time-variation of the density operator:
\begin{equation}\label{eq:conteq}
\begin{split}
\partial_{t} \tau_{t}( n_{{\bf x}} ) &= i \tau_{t} \big( [ \mathcal{H}, n_{{\bf x}} ]\big) \\
&= \sum_{{\bf y} \in \Lambda_{L}} \tau_{t}( i a^{*}_{{\bf y}} H({\bf y}; {\bf x}) a_{{\bf x}} - i a^{*}_{{\bf x}} H({\bf x}; {\bf y}) a_{{\bf y}}  ) \\
&\equiv \sum_{{\bf y} \in \Lambda_{L}} \tau_{t}( j_{{\bf y}, {\bf x}})\;,
\end{split}
\end{equation}
where $j_{{\bf y}, {\bf x}} = i a^{*}_{{\bf y}} H({\bf y}; {\bf x}) a_{{\bf x}} - i a^{*}_{{\bf x}} H({\bf x}; {\bf y}) a_{{\bf y}}$ is the bond current between ${\bf x}$ and ${\bf y}$. Let us derive the conservation law for $n_{x}$. To this end, in (\ref{eq:conteq}) we sum over the the internal label $\rho$ in $\mathbf x=(x,\rho)$, and we introduce the discrete derivative $\text{d}_{x}$ as $(\text{d}_{x}f)(x) = f(x) - f(x-1)$ to obtain the continuity equation
\begin{equation}\label{eq:cons1d}
\partial_{t} \tau_{t}( n_{x} ) = -\text{d}_{x} \tau_{t}(j_{x})\;,
\end{equation}
where we introduced the lattice current $j_{x}$, determined by
\begin{equation}\label{eq:defj1d}
\mathrm d_{x}j_{x}:= \sum_{y\in \Gamma_{L}}\sum_{\rho,\rho'\in S_{M}} j_{(x,\rho),(y,\rho')}\;.
\end{equation}
The simple identity (\ref{eq:cons1d}) will have far reaching consequences in our analysis. Let $j_{\nu,x}$, $\nu = 0, 1$, be defined as $j_{0,x}= n_{x}$, $j_{1, x} = j_{x}$, then we may characterize them in momentum-space as:
\begin{equation}\label{eq:currFou}
j_{\nu,x}=\frac{1}{L}\sum_{p\in B_{L}} e^{ipx} \hat \jmath_{\nu, p}\qquad \hat\jmath_{\nu,p}=\frac{1}{L}\sum_{k\in B_{L}} (a^{*}_{k-p}, \hat J_{\nu}(k,p) a_{k})\;,
\end{equation}
where we used the shorthand notation $(a^{*}, J a)=\sum_{\rho,\rho'} a^{*}_{\rho}J_{\rho\rho'} a_{\rho'}$. We have
\begin{equation}\label{eq:Jnudef}
\hat J_{0}(k,p) = 1\;, \qquad \hat J_{1}(k,p)=i\frac{\hat H(k)-\hat H(k-p)}{1-e^{-ip}}\;.
\end{equation}
\begin{remark}
If the Hamiltonian $H$ has range 1, then $j_{1,x}$ takes a particularly simple form, given by the bond current between sites $x$ and $x+1$:
\begin{equation}
j_{1,x}\equiv j_{x}=\sum_{\rho,\rho'\in S_{M}} j_{(x,\rho),(x+1,\rho')}\;.
\end{equation}
\end{remark}
Finally, we shall define a smeared version of the density and of the current operators as
\begin{equation}\label{eq:jmuphi}
j_{\nu}(\mu_{\theta}) = \sum_{x\in \Gamma_{L}} j_{\nu,x} \theta \mu(\theta x)\;.
\end{equation}
\paragraph{Imaginary time evolution.}  For $t\in \mathbb{R}$, we define the imaginary time evolution of an observable $\mathcal{O}$ as:
\begin{equation}\label{eq:gammatdef}
\gamma_{t}(\mathcal{O}) := e^{t ( \mathcal{H} - \mu \mathcal{N} )} \mathcal{O} e^{-t(\mathcal{H} - \mu \mathcal{N})}\;.
\end{equation}
Observe that if $[\mathcal{O}, \mathcal{N}] = 0$, the following identity holds:
\begin{equation}
\gamma_{t}(\mathcal{O}) = \tau_{-it}(\mathcal{O})\;.
\end{equation}
An important consequence of the definition of the imaginary time dynamics (\ref{eq:gammatdef}) is the Kubo-Martin-Schwinger (KMS) identity, for all $t,s\in \mathbb{R}$:
\begin{equation}
\langle \gamma_{t}(\mathcal{A}) \gamma_{s}(\mathcal{B})  \rangle_{\beta,L} = \langle \gamma_{s+\beta}(\mathcal{B}) \gamma_{t}(\mathcal{A})  \rangle_{\beta,L}\;.
\end{equation}
In our finite-dimensional setting, the KMS identity is a simple consequence of the definitions, and of the cyclicity of the trace. Nevertheless, it will have far-reaching consequences. In order to discuss them, we need to introduce the notions of time-ordering and of Euclidean correlation functions. 

Let $0\leq t_{i} < \beta$, for $i=1,2,\ldots, n$, such that $t_{i} \neq t_{j}$ for $i\neq j$. We define the time-ordering of $\gamma_{t_{1}}(a^{\sharp_{1}}_{{\bf x}_{1}}),\ldots, \gamma_{t_{n}}(a^{\sharp_{n}}_{{\bf x}_{n}})$ as:
\begin{equation}\label{eq:time}
{\bf T} \gamma_{t_{1}}(a^{\sharp_{1}}_{{\bf x}_{1}}) \cdots \gamma_{t_{n}}(a^{\sharp_{n}}_{{\bf x}_{n}}) =  (-1)^{\pi}  \gamma_{t_{\pi(1)}}(a^{\sharp_{\pi(1)}}_{{\bf x}_{\pi(1)}}) \cdots \gamma_{t_{\pi(n)}}(a^{\sharp_{\pi(n)}}_{{\bf x}_{\pi(n)}})\;,
\end{equation}
where $\pi$ is the permutation needed in order to bring the times in a decreasing order, from the left, with sign $(-1)^{\pi}$. In case two or more times are equal, the definition (\ref{eq:time}) is extended via normal ordering. Other extensions are of course possible; it is worth anticipating that in our applications this arbitrariness will play no role, ultimately because it involves a zero measure set of times, and because the algebra of fermionic operators on a finite lattice consists of bounded operators. The above definition extends to all operators on the fermionic Fock space by linearity. In particular, for $\mathcal{O}_{1}, \ldots, \mathcal{O}_{n}$ even in the number of creation and annihilation operators, we have:
\begin{equation}\label{eq:Tord1}
{\bf T} \gamma_{t_{1}}(\mathcal{O}_{1}) \cdots \gamma_{t_{n}}(\mathcal{O}_{n}) = \gamma_{t_{\pi(1)}}(\mathcal{O}_{\pi(1)}) \cdots \gamma_{t_{\pi(n)}}(\mathcal{O}_{\pi(n)})\;.
\end{equation}
\paragraph{Euclidean correlation functions and Wick rotation.} Let $t_{i} \in [0,\beta)$, for $i=1,\ldots, n$. Given operators $\mathcal{O}_{1}, \ldots, \mathcal{O}_{n}$, we define the time-ordered Euclidean correlation function as:
\begin{equation}\label{eq:Tord2}
\langle {\bf T} \gamma_{t_{1}}(\mathcal{O}_{1}) \cdots \gamma_{t_{n}}(\mathcal{O}_{n}) \rangle_{\beta,L}\;.
\end{equation}
From the definition of fermionic time-ordering, and from the KMS identity, it is not difficult to check that:
\begin{equation}\label{eq:extension}
\begin{split}
\langle {\bf T} \gamma_{t_{1}}(\mathcal{O}_{1}) \cdots \gamma_{\beta}(\mathcal{O}_{k}) &\cdots \gamma_{t_{n}}(\mathcal{O}_{n}) \rangle_{\beta,L} \\ &= (\pm 1) \langle {\bf T} \gamma_{t_{1}}(\mathcal{O}_{1}) \cdots \gamma_{0}(\mathcal{O}_{k}) \cdots \gamma_{t_{n}}(\mathcal{O}_{n}) \rangle_{\beta,L}\;;
\end{split}
\end{equation}
in the special case in which the operators involve an even number of creation and annihilation operators, which will be particularly relevant for our analysis, the overall sign is $+1$. The property (\ref{eq:extension}) allows to extend in a periodic (sign $+1$) or antiperiodic (sign $-1$) way the correlation functions to all times $t_{i} \in \mathbb{R}$. From now on, when discussing time-ordered correlations we shall always assume that this extension has been taken, unless otherwise specified.

An important example of Euclidean correlation function is the two-point function,
\begin{equation}
\langle {\bf T} \gamma_{t}(a_{{\bf x}}) a^{*}_{{\bf y}} \rangle_{\beta,L}\;.
\end{equation}
By the properties of the fermionic time-ordering, the two point function extends to a $\beta$-antiperiodic function in the imaginary time variables. Let $k_{0} = (2\pi/\beta)( n + 1/2 )$ be the frequencies relevant for $\beta$-antiperiodic functions, also called fermionic Matsubara frequencies, and let $k\in B_{L}$ be a point in the Brillouin zone. We define the Fourier transform of the two-point function as:
\begin{equation}
\begin{split}
(g(\underline{k}))_{\rho\rho'} &:= \int_{0}^{\beta} d t\, e^{-ik_{0} t} \sum_{x\in \Lambda_{L}} e^{-ikx} \langle {\bf T} \gamma_{t}(a_{(x,\rho)}) a^{*}_{(0,\rho')} \rangle_{\beta,L}\;.
\end{split}
\end{equation}
For a non-interacting model, the two-point function can be computed explicitly in terms of the Bloch Hamiltonian. We have:
\begin{equation}\label{eq:gfou}
g(\underline{k}) = \frac{1}{ik_{0} + \hat H(k) - \mu}\;.
\end{equation}
The two-point function completely characterizes the Gibbs state of system: being the Hamiltonian quadratic in the fermionic creation and annihilation operators, all Euclidean correlation function can be computed via the fermionic Wick's rule.

Next, we define the connected time-ordered Euclidean correlation functions, or time-ordered Euclidean cumulants, as:
\begin{equation}\label{eq:Tcumul}
\begin{split}
&\langle {\bf T} \gamma_{t_{1}}(\mathcal{O}_{1}); \cdots; \gamma_{t_{n}}(\mathcal{O}_{n}) \rangle_{\beta,L} \\
&\qquad := \frac{\partial^{n}}{\partial \lambda_{1} \cdots \partial \lambda_{n}} \log \Big\{ 1 + \sum_{I \subseteq \{1, 2,\ldots, n\}} \lambda(I) \langle {\bf T} \mathcal{O}(I)  \rangle_{\beta,L} \Big\}\Big|_{\lambda_{i} = 0}
\end{split}
\end{equation}
where $I$ is a non-empty ordered subset of $\{1, 2,\ldots, n\}$, $\lambda(I) = \prod_{i\in I} \lambda_{i}$ and $\mathcal{O}(I) = \prod_{i\in I} \gamma_{t_{i}}(\mathcal{O}_{i})$. For $n=1$, this definition reduces to $\langle {\bf T} \gamma_{t_{1}}(\mathcal{O}_{1}) \rangle \equiv \langle \gamma_{t_{1}}(\mathcal{O}_{1}) \rangle = \langle \mathcal{O}_{1} \rangle$, while for $n=2$ one gets $\langle {\bf T} \gamma_{t_{1}}(\mathcal{O}_{1}); \gamma_{t_{2}}(\mathcal{O}_{2}) \rangle = \langle {\bf T} \gamma_{t_{1}}(\mathcal{O}_{1}) \gamma_{t_{2}}(\mathcal{O}_{2}) \rangle - \langle {\bf T} \gamma_{t_{1}}(\mathcal{O}_{1}) \rangle \langle {\bf T} \gamma_{t_{2}}(\mathcal{O}_{2}) \rangle$. The following relation between correlation functions and connected correlation function holds true:
\[
\langle {\bf T} \gamma_{t_{1}}(\mathcal{O}_{1}) \cdots \gamma_{t_{n}}(\mathcal{O}_{n}) \rangle_{\beta,L} = \sum_{P} \prod_{J\in P} \langle {\bf T} \gamma_{t_{j_{1}}}(\mathcal{O}_{j_{1}}); \cdots; \gamma_{t_{j_{|J|}}}(\mathcal{O}_{j_{|J|}}) \rangle_{\beta,L}\;,
\]
where $P$ is the set of all partitions of $\{1, 2, \ldots, n\}$ into ordered subsets, and $J$ is an element of the partition $P$, $J = \{ j_{1}, \ldots, j_{|J|} \}$. 

The next proposition establishes the connection between the real time perturbation theory generated by the Duhamel series of Proposition \ref{prop:duha} and the Euclidean correlation functions, for special choices of the adiabatic parameter $\eta$.
\begin{proposition}[Wick rotation]\label{prop:wick} Let $\eta_{\beta} \in \frac{2\pi}{\beta} \mathbb{N}$. Let $\mathcal{O}$ be such that $[\mathcal{O}, \mathcal{N}] = 0$. Then, the following identity holds true:
\begin{equation}\label{eq:wick}
\begin{split}
&\int_{-\infty \leq s_{n} \leq \ldots \leq s_{1} \leq t} d \underline{s}\, e^{\eta_{\beta}(s_{1} + \ldots +s_{n})} \langle [ \cdots [[ \tau_{t}(\mathcal{O}), \tau_{s_{1}}(\mathcal{P})], \tau_{s_{2}}(\mathcal{P})] \cdots \tau_{s_{n}}(\mathcal{P}) ] \rangle_{\beta,L} \\
&= \frac{(-i)^{n} e^{n \eta_{\beta} t}}{n!} \int_{[0,\beta)^{n}} d\underline{s}\, e^{-i\eta_{\beta}(s_{1} + \ldots +s_{n})}\langle {\bf T} \gamma_{s_{1}}(\mathcal{P}); \gamma_{s_{2}}(\mathcal{P}); \cdots; \gamma_{s_{n}}(\mathcal{P});\mathcal{O}\rangle_{\beta,L}\;.
\end{split}
\end{equation}
\end{proposition}
We refer the reader to \cite{GLMP} for the proof of this identity, which holds for general lattice models in $d$ dimensions. Eq. (\ref{eq:wick}) introduces a very useful representation of the Duhamel series. In particular, for non-interacting models, all contributions to the Duhamel series can be computed starting from the two-point function (\ref{eq:gfou}) via Wick's rule. Following \cite{GLMP}, for general choices of $\eta>0$ we will use an approximation argument to replace $\eta$ with its best approximant $\eta_{\beta} \geq \eta$, based on Lieb-Robinson bounds for non-autonomous dynamics.

\section{Large scale response of $1d$ systems}\label{sec:main1d}

Let us define the full response as:
\begin{equation}\label{eq:resp1d}
\chi^{\beta, L}_{\nu}(x;\eta,\theta) := \frac{1}{\theta} \Big( \Tr j_{\nu,x} \rho(0) - \Tr j_{\nu,x} \rho_{\beta, \mu, L}\Big)\qquad \nu = 0, 1.
\end{equation}
%
%
%
The quantity $\chi^{\beta,L}_{0}$ is called the susceptibility, while the quantity $\chi^{\beta,L}_{1}$ is called the conductance. The next theorem gives precise information about these objects.
\begin{theorem}[Response of one-dimensional systems]\label{thm:main1d} Let $\theta = a\eta$. There exist constants $w,\gamma,\eta_{0} > 0$ such that, for all $0\leq\eta<\eta_{0}$ and $a \leq w\left|\log\eta\right|$, and for $\beta, L$ large enough:
\begin{equation}\label{eq:main1d}
\begin{split}
\chi^{\beta,L}_{\nu}(x;\eta,\theta) =-\sum_{\omega=1}^{N}\frac{v_{\omega}^{\nu}}{2\pi|v_{\omega}|}\int_\R \hat\mu_{\infty}(q)e^{iq\theta x}\frac{v_\omega q}{-i/a+v_\omega q}\,\frac{dq}{2\pi} + O(\eta^{\gamma}) + O\Big( \frac{1}{\eta^{3} \beta} \Big)
\end{split}
\end{equation}
where $v_{\omega}^{\nu}=\delta_{\nu,0}+ v_{\omega}\delta_{\nu,1}$. 
\end{theorem}
\begin{corollary}\label{cor:1d} In particular, we can isolate three regimes.
\begin{enumerate}
\item Suppose that $a \to 0$ as $\eta\to0^{+}$. Then, 
\begin{equation}
\chi_{\nu}(x,\eta,\theta) = O(a) + O(\eta^{\gamma}) + O\Big( \frac{1}{\eta^{3} \beta} \Big)\;.
\end{equation}
\item Suppose that $a$ is constant in $\eta$. Then:
\begin{equation}
\begin{split}
&\chi^{\beta, L}_{\nu}(x;\eta,\theta) \\
&= - \sum_{\omega=1}^{N}\frac{v_{\omega}^{\nu}}{2\pi|v_{\omega}|}\int_\R \mu_{\infty}(\theta x-y)\Big[\delta(y)-\frac{1}{a|v_{\omega}|} e^{-y/av_{\omega}}\Theta(y/v_{\omega})\Big]\,dy+ O(\eta^{\gamma}) + O\Big( \frac{1}{\eta^{3} \beta} \Big)\;,
\end{split}
\end{equation}
where $\Theta$ is the Heaviside step function.
\item Finally, suppose that $a \to \infty$ as $\eta \to 0^{+}$. Then:
\begin{equation}
\chi^{\beta, L}_{\nu}(x;\eta,\theta)= - \mu_{\infty}(\theta x)\sum_{\omega=1}^{N}\frac{v_{\omega}^{\nu}}{2\pi|v_{\omega}|} + O(\eta^{\gamma}) + O\Big( \frac{1}{\eta^{3} \beta} \Big)\;.
\end{equation}
In particular,
\begin{equation}
\lim_{\eta\to0^{+}} \lim_{\beta \to \infty} \lim_{L\to \infty}\chi^{\beta, L}_{1}(x;\eta,\theta)=-\mu_{\infty}(0)\sum_{\omega=1}^{N}\frac{\mathrm{sgn}(v_{\omega})}{2\pi} = 0
\end{equation}
since the number of left chiral fermionic modes equals the number of right chiral modes, for short-ranged one-dimensional systems.
\end{enumerate}
\end{corollary}

\begin{remark}
\begin{itemize}
\item[(i)] The above theorem justifies the validity of linear response for non-interacting one dimensional systems in the absence of a spectral gap, in the zero temperature limit and in the thermodynamic limit. Remarkably, the leading contribution to transport only depends on the Fermi velocities $v_{\omega}$; all the higher order contributions in the perturbation are taken into account by the $O(\eta^{\gamma})$ corrections.

\item[(ii)] The parameter $a = \theta / \eta$ defines the relative magnitude between the variations in time and in space of the perturbation. The regime $a\ll 1$ corresponds to the situation in which the system is effectively exposed to the perturbation for a shorter time, while $a\gg 1$ corresponds to a longer time. The choice $a=1$ defines a natural transition between the two regimes. It coincides with the Euler scaling, which, as mentioned in the introduction, is the starting point for hydrodynamic effective models. This framework can also be used to derive a linear response theory, see {\it e.g.} \cite{DVD} and references therein.

\item[(iii)] As the proof will show, due to the absence of a spectral gap all terms in the Duhamel expansion of the full response (\ref{eq:resp1d}) are apparently of the same order in $\theta$, and in particular of the same order of the linear response. If so, this would not allow to prove that the linear term gives the dominant contribution to (\ref{eq:resp1d}), since we are not making any assumption on the size of $\mu(\theta x)$. The validity of linear response arises thanks to a crucial cancellation taking place at higher orders in the Duhamel expansion. This cancellation allows to prove that all terms after the linear response, up to a large $\theta$-dependent order $n(\theta)$, are suppressed by an extra factor $\eta^{\gamma}$ with $\gamma>0$, and hence they vanish as $\eta\to 0$; this is ultimately due to a cancellation taking place in the scaling limit of the model, often called ``loop cancellation'' in quantum field theory, related to bosonization \cite{FGM2}. Concerning the orders $n$ larger that $n(\theta)$, we will prove that they give a small contribution, thanks to improved estimated on the Duhamel series, that allow to prove its summability.
\end{itemize}
\end{remark}

\section{Proof of Theorem \ref{thm:main1d}}\label{sec:proof1d}

\subsection{Auxiliary dynamics} Let $\eta_{\beta}$ be the smallest number in $(2\pi / \beta) \mathbb{N}$, such that $\eta_{\beta} \geq \eta$. We define:
\begin{equation}
\mathcal{H}_{\beta,\eta}(t) = \mathcal{H} + \theta e^{\eta_{\beta} t} \sum_{{\bf x} \in \Lambda_{L}} \mu_{\alpha}(\theta x) a^{*}_{{\bf x}} a_{{\bf x}}\;,
\end{equation}
where $\mu_{\alpha}(\theta x)$ is obtained from $\mu(\theta x)$ after cutting off large momenta, in a smooth way:
\begin{equation}\label{eq:mualpha}
\mu_{\alpha}(\theta x) = \frac{1}{L} \sum_{p \in B_{L}} e^{ipx} \hat\mu_{\alpha,\theta}(p)\;,\qquad \hat\mu_{\alpha,\theta}(p) = \hat \mu_{\theta}(p) \chi( \theta^{\alpha-1} |p|_{\mathbb{T}})\;,
\end{equation}
with $\alpha\in(0,1)$ and $\chi(\cdot) : \mathbb{R}^{+}\to [0;1]$ a smooth cutoff function such that:
\begin{equation}\label{eq:cutoff}
\chi(t)=
\begin{cases}
1 & t\leq1\\
0 & t\geq2.
\end{cases}
\end{equation}
With respect to the original perturbation, in the definition of the function $\mu_{\alpha}(\theta x)$ we are cutting off all momenta $p$ of norm greater than $2\theta^{1-\alpha}$. By the smoothness of the function $\mu(\cdot)$, this regularization will have a minor effect.

Let us denote by $\widetilde{\mathcal{U}}(t;s)$ the two-parameter unitary group generated by $\mathcal{H}_{\beta,\eta}(t)$,
\begin{equation}
i\partial_{t} \widetilde{\mathcal{U}}(t;s) = \mathcal{H}_{\beta,\eta}(t) \widetilde{\mathcal{U}}(t;s)\;,\qquad \widetilde{\mathcal{U}}(s;s) = \mathbbm{1}\;.
\end{equation}
The next result allows to control the error introduced by replacing the dynamics generated by $\mathcal{H}(\eta t)$ with the dynamics generated by $\mathcal{H}_{\beta,\eta}(t)$. 
\begin{proposition}[Approximation by the auxiliary dynamics]\label{prp:LR} Under the same assumptions of Theorem \ref{thm:main1d}, it follows that, for any $m\in \mathbb{N}$:
\begin{equation}\label{eq:LRprop}
\Big\| \widetilde{\mathcal{U}}(t;-\infty)^{*} \mathcal{O}_{X} \widetilde{\mathcal{U}}(t;-\infty) - \mathcal{U}(t;-\infty)^{*} \mathcal{O}_{X} \mathcal{U}(t;-\infty)   \Big\| \leq C_{m} \frac{\theta^{1 + \alpha (m -1)}}{\eta^{2}} + \frac{C\theta}{\eta^{3} \beta}\;.
\end{equation}
\end{proposition}
\begin{remark} Let us briefly comment on the reason for introducing the auxiliary dynamics. Let $\mathcal{A}$ be such that $[\mathcal{A}, \mathcal{N}] = 0$, and let $\tilde \tau_{t}(\mathcal{A}) = e^{\eta t} \tau_{t}(\mathcal{A})$, which we can view as a (non-unitary) regularization of the original dynamics, damped in the past. Then, if $\eta \in \frac{2\pi}{\beta} \mathbb{N}$, we observe that $\tilde \tau_{t}(\cdot)$ satisfies the KMS identity:
\begin{equation}\label{eq:KMSreg}
\langle \tilde \tau_{t}(\mathcal{A}) \mathcal{B}  \rangle_{\beta,L} = \langle \mathcal{B} \tilde \tau_{t + i\beta}(\mathcal{A}) \rangle_{\beta,L}\;,
\end{equation}
where we used the trivial but crucial identity $e^{\eta t} = e^{\eta (t + i\beta)}$. The KMS identity for the regularized dynamics (\ref{eq:KMSreg}) plays a key role in the mapping of the real time Duhamel series into an imaginary time expansion, which can be efficiently studied \cite{GLMP}. See Proposition \ref{prop:wick}. Due to the fact that the difference between any $\eta > 0$ and its best approximation $\eta_{\beta} \in \frac{2\pi}{\beta} \mathbb{N}$ vanishes as $\beta \to \infty$, this approach is suitable to study zero or low temperature systems (lower than some $\eta$-dependent value).

Concerning the momentum-space regularization of the perturbation, this will be used to rule out the scattering of different Fermi points, at orders in the expansion that are less that some $\theta$-dependent (high) order.
\end{remark}
\begin{proof}[Proof of Proposition \ref{prp:LR}] The proof of this proposition is a standard argument, based on Lieb-Robinson bounds. It is a slight generalization of an analogous result in \cite{GLMP}. We reproduce it for completeness. We start by writing:
\begin{equation}\label{eq:diffLR}
\begin{split}
&\Big\| \widetilde{\mathcal{U}}(t;-\infty)^{*} \mathcal{O}_{X} \widetilde{\mathcal{U}}(t;-\infty) - \mathcal{U}(t;-\infty)^{*} \mathcal{O}_{X} \mathcal{U}(t;-\infty)  \Big\| \\
&\qquad =  \Big\| \mathcal{O}_{X} - \mathcal{U}_{\text{I}}(t;-\infty)^{*} \mathcal{O}_{X} \mathcal{U}_{\text{I}}(t;-\infty) \Big\|\;,
\end{split}
\end{equation}
where $\mathcal{U}_{\text{I}}(t;s) = \mathcal{U}(t;s) \widetilde{\mathcal{U}}(t;s)^{*}$. Then, since:
\begin{equation}
i\partial_{s} \mathcal{U}_{\text{I}}(t;s) = -\mathcal{U}_{\text{I}}(t;s) \widetilde{\mathcal{U}}(t;s) \big( \mathcal{H}(\eta s) - \mathcal{H}_{\eta, \beta}(s) \big) \widetilde{\mathcal{U}}(t;s)^{*}
\end{equation}
we have:
\begin{equation}
\begin{split}
&\mathcal{O}_{X} - \mathcal{U}_{\text{I}}(t;-\infty)^{*} \mathcal{O}_{X} \mathcal{U}_{\text{I}}(t;-\infty) \\
&\quad = \int_{-\infty}^{t} ds\, \frac{d}{ds} \mathcal{U}_{\text{I}}(t;s)^{*} \mathcal{O}_{X} \mathcal{U}_{\text{I}}(t;s) \\
& \quad = i \int_{-\infty}^{t} ds\, \mathcal{U}_{\text{I}}(t;s)^{*} \Big[  \mathcal{O}_{X}, \widetilde{\mathcal{U}}(t;s) \big( \mathcal{H}(\eta s) - \mathcal{H}_{\eta, \beta}(s) \big) \widetilde{\mathcal{U}}(t;s)^{*}\Big]  \mathcal{U}_{\text{I}}(t;s)\;,
\end{split}
\end{equation}
which implies:
\begin{equation}\label{eq:LR3}
\begin{split}
&\Big\| \mathcal{O}_{X} - \mathcal{U}_{\text{I}}(t;-\infty)^{*} \mathcal{O}_{X} \mathcal{U}_{\text{I}}(t;-\infty) \Big\| \\
&\qquad \leq  \int_{-\infty}^{t} ds\,\Big\| \Big[ \widetilde{\mathcal{U}}(t;s)^{*} \mathcal{O}_{X} \widetilde{\mathcal{U}}(t;s),  \big( \mathcal{H}(\eta s) - \mathcal{H}_{\eta, \beta}(s) \big)\Big] \Big\| \\
&\qquad \leq \int_{-\infty}^{t} ds\,\Big( \Big\| \Big[ \widetilde{\mathcal{U}}(t;s)^{*} \mathcal{O}_{X} \widetilde{\mathcal{U}}(t;s), \mathcal{Q}_{1}(s)\Big] \Big\| + \Big\| \Big[ \widetilde{\mathcal{U}}(t;s)^{*} \mathcal{O}_{X} \widetilde{\mathcal{U}}(t;s), \mathcal{Q}_{2}(s)\Big] \Big\|\Big)\;,
 \end{split}
\end{equation}
where:
\begin{equation}
\begin{split}
\mathcal{Q}_{1}(s) &= \theta(e^{\eta s} - e^{\eta_{\beta} s}) \sum_{{\bf x} \in \Lambda_{L}} \mu(\theta x) a^{*}_{{\bf x}} a_{{\bf x}} \\
\mathcal{Q}_{2}(s) &= \theta e^{\eta_{\beta} s} \sum_{{\bf x} \in \Lambda_{L}} (\mu(\theta x) -\mu_{\alpha}(\theta x)) a^{*}_{{\bf x}} a_{{\bf x}}\;.
\end{split}
\end{equation}
Consider the second term. We have:
\begin{equation}\label{eq:Q2}
\begin{split}
&\Big\| \Big[ \widetilde{\mathcal{U}}(t;s)^{*} \mathcal{O}_{X} \widetilde{\mathcal{U}}(t;s), \mathcal{Q}_{2}(s)\Big] \Big\| \\&\qquad \leq \theta e^{\eta_{\beta} s} \sum_{{\bf x} \in \Lambda_{L}} \big| \mu(\theta x) - \mu_{\alpha}(\theta x) \big| \Big\| \Big[ \widetilde{\mathcal{U}}(t;s)^{*} \mathcal{O}_{X} \widetilde{\mathcal{U}}(t;s), n_{{\bf x}} \Big] \Big\| \\
&\qquad \leq \theta e^{\eta_{\beta} s} \| \mu - \mu_{\alpha} \|_{\infty} \sum_{{\bf x}\in \Lambda_{L}} \Big\| \Big[ \widetilde{\mathcal{U}}(t;s)^{*} \mathcal{O}_{X} \widetilde{\mathcal{U}}(t;s), n_{{\bf x}} \Big] \Big\|\;.
\end{split}
\end{equation}
By the Lieb-Robinson bound, we can further estimate:
\begin{equation}\label{eq:Q2b}
\Big\| \Big[ \widetilde{\mathcal{U}}(t;s)^{*} \mathcal{O}_{X} \widetilde{\mathcal{U}}(t;s), \mathcal{Q}_{2}(s)\Big] \Big\| \leq C_{X} \theta e^{\eta_{\beta} s} (1 + |t-s|) \| \mu - \mu_{\alpha} \|_{\infty}\;.
\end{equation}
The norm can be bounded as:
\begin{equation}
\| \mu - \mu_{\alpha} \|_{\infty} \leq \frac{1}{L} \sum_{p \in B_{L}} \big| \mu_{\theta}(p)\big| \big[1 - \chi( \theta^{\alpha-1} |p|_{\mathbb{T}})\big]\;.
\end{equation}
Next, by the smoothness of $\mu(x)$, we have $(1 + |q|_{\mathbb{T}}^{m}) |\hat \mu(q)| \leq C_{m}$, which easily gives:
\begin{equation}
\frac{1}{L} \sum_{p \in B_{L}} \big| \mu_{\theta}(p)\big| \big[1 - \chi( \theta^{\alpha-1} |p|_{\mathbb{T}})\big] \leq K_{m} \theta^{\alpha(m-1)}\;,\qquad\quad \forall m\in \mathbb{N}\;.
\end{equation}
Plugging this bound in (\ref{eq:Q2b}), we obtain:
\begin{equation}
\Big\| \Big[ \widetilde{\mathcal{U}}(t;s)^{*} \mathcal{O}_{X} \widetilde{\mathcal{U}}(t;s), \mathcal{Q}_{2}(s)\Big] \Big\| \leq C_{X} K_{m}  \theta e^{\eta_{\beta} s} (1 + |t-s|) \theta^{\alpha(m-1)}\;,
\end{equation}
which allows to estimate the corresponding contribution to (\ref{eq:LR3}) as:
\begin{equation}\label{eq:estQ1}
\int_{-\infty}^{t} ds\, \Big\| \Big[ \widetilde{\mathcal{U}}(t;s)^{*} \mathcal{O}_{X} \widetilde{\mathcal{U}}(t;s), \mathcal{Q}_{2}(s)\Big] \Big\| \leq C_{m} \frac{\theta^{1 + \alpha (m -1)}}{\eta^{2}}\;.
\end{equation}
Consider now the contribution due to $\mathcal{Q}_{1}(s)$ in the right-hand side of (\ref{eq:LR3}). This term can be estimated as in \cite[Proposition 4.1]{GLMP}, we omit the details. The result is:
\begin{equation}\label{eq:estQ2}
\int_{-\infty}^{t} ds\, \Big\| \Big[ \widetilde{\mathcal{U}}(t;s)^{*} \mathcal{O}_{X} \widetilde{\mathcal{U}}(t;s), \mathcal{Q}_{1}(s)\Big] \Big\| \leq  \frac{C\theta }{\eta^{3} \beta}\;.
\end{equation}
Combining (\ref{eq:diffLR}), (\ref{eq:LR3}), (\ref{eq:estQ1}), (\ref{eq:estQ2}), the final result (\ref{eq:LRprop}) follows.
\end{proof}
Therefore, by Proposition \ref{prp:LR} we have:
\begin{equation}\label{eq:approxdyn}
\Tr j_{\nu,x} \rho(t) = \Tr j_{\nu,x} \widetilde\rho(t) + \mathcal{E}^{\beta, L}_{\nu} (x,t;\eta, \theta)\;,
\end{equation}
where $\tilde \rho(t) = \widetilde{\mathcal{U}}(t;-\infty) \rho_{\beta, \mu, L} \widetilde{\mathcal{U}}(t;-\infty)^{*}$ and, by Proposition \ref{prp:LR}:
\begin{equation}\label{eq:estLR}
| \mathcal{E}^{\beta, L}_{\nu} (x,t;\eta, \theta) | \leq C_{m} \frac{\theta^{1 + \alpha (m -1)}}{\eta^{2}} + \frac{C\theta}{\eta^{3} \beta}\;.
\end{equation}
We require this error term is vanishing once divided by $\theta = a\eta$, recall (\ref{eq:resp1d}). To this end, it is sufficient to consider:
\begin{equation}\label{eq:errcond}
a=o\left(\eta^{\frac{2}{\alpha(m-1)}-1}\right).
\end{equation}
In order to leave room for $a \to \infty$ as $\eta \to 0^{+}$, we will choose $\alpha\in(0,1)$ and we will pick $m\in\N$ so that $\alpha(m-1)>2$. Later, we will have to introduce stronger constraints on $a$.

\subsection{Duhamel series of the auxiliary dynamics}

Consider the main term in the right-hand side of (\ref{eq:approxdyn}). By Proposition \ref{prop:wick}:
\begin{equation}\label{eq:duha1}
\begin{split}
& \Tr j_{\nu,x} \widetilde\rho(t) -  \Tr j_{\nu,x} \rho_{\beta, \mu, L} \\&\quad = \sum_{n\geq 1}\frac{(-1)^{n} e^{n \eta_{\beta} t}}{n!} \int_{[0,\beta)^{n}} d\underline{s}\, e^{-i\eta_{\beta}(s_{1} + \ldots +s_{n})}\langle {\bf T} \gamma_{s_{1}}(\widetilde{\mathcal{P}}); \cdots; \gamma_{s_{n}}(\widetilde{\mathcal{P}}); j_{\nu,x}\rangle_{\beta,L}\;.
 \end{split}
\end{equation}
It is convenient to rewrite the $n$-th order contribution to the expansion in Fourier space. Recalling the definition (\ref{eq:jmuphi}) and the form of the perturbation (\ref{eq:Pdef}),
\begin{equation}\label{eq:fou}
\begin{split}
&\int_{[0,\beta)^{n}} d\underline{s}\, e^{-i\eta_{\beta}(s_{1} + \ldots +s_{n})}\langle {\bf T} \gamma_{s_{1}}(\widetilde{\mathcal{P}}); \cdots; \gamma_{s_{n}}(\widetilde{\mathcal{P}}); j_{\nu,x}\rangle_{\beta,L}\\
&\quad = \int_{[0,\beta)^{n}} d\underline{s}\, e^{-i\eta_{\beta}(s_{1} + \ldots +s_{n})}\langle {\bf T} \gamma_{s_{1}}(j_{0}(\mu_{\alpha,\theta}));  \cdots; \gamma_{s_{n}}(j_{0}(\mu_{\alpha,\theta})); j_{\nu,x}\rangle_{\beta,L}\;.
\end{split}
\end{equation}
Next, using that:
\begin{equation}\label{eq:currfourier}
j_{0}(\mu_{\alpha,\theta}) = \frac{\theta}{L} \sum_{p \in B_{L}} \hat \mu_{\alpha,\theta}(-p) \hat\jmath_{\nu,p}\;,\qquad \hat\jmath_{\nu,p} = \sum_{{\bf x} \in \Lambda_{L}} e^{-ipx} \hat\jmath_{\nu,{\bf x}}\;,
\end{equation}
and writing, for $\underline{p} = (p_{0}, p_{1})$ with $p_{1} = p$ and $p_{0} = \eta_{\beta}$,
\begin{equation}\label{eq:densitytimefourier}
\hat n_{\underline{p}} := \int_{0}^{\beta} ds\, e^{-i\eta_{\beta} s} \gamma_{s}(\hat n_{p})\;,
\end{equation}
we can rewrite the right-hand side of (\ref{eq:fou}) as:
\begin{equation}\label{eq:fou2}
\begin{split}
&\int_{[0,\beta)^{n}} d\underline{s}\, e^{-i\eta_{\beta}(s_{1} + \ldots +s_{n})}\langle {\bf T} \gamma_{s_{1}}(j_{0}(\mu_{\alpha,\theta}));  \cdots; \gamma_{s_{n}}(j_{0}(\mu_{\alpha,\theta})); j_{\nu,x}\rangle_{\beta,L} \\
&\qquad = \frac{\theta^{n}}{L^{n}} \sum_{\{p_{i}\} \in B_{L}^{n}} \Big[\prod_{i=1}^{n} \hat\mu_{\alpha,\theta}(-p_{i})\Big] e^{ip_{n+1}x} \frac{1}{\beta L}\langle {\bf T}\, \hat n_{\underline{p}_{1}}\;; \cdots \;; \hat n_{\underline{p}_{n}}\;; \hat\jmath_{\nu,\underline{p}_{n+1}} \rangle_{\beta,L}\;,
\end{split}
\end{equation}
where $\underline{p}_{n+1} = -\underline{p}_{1} - \ldots - \underline{p}_{n}$. In order to derive (\ref{eq:fou2}), we used the space-time translation invariance of the Gibbs state. Next, by Wick's rule:
\begin{equation}\label{eq:corrwick}
\begin{split}
&\frac{1}{\beta L}\langle {\bf T}\, \hat n_{\underline{p}_{1}}\;; \cdots \;; \hat n_{\underline{p}_{n}}\;; \hat\jmath_{\nu,\underline{p}_{n+1}}  \rangle_{\beta,L} \\
&\quad = -\frac{1}{\beta L} \sum_{\pi\in S_{n}}\sum_{\underline{k} \in \mathbb{M}_{\beta} \times B_{L}} \Tr \Big[ \hat J_{\nu}(k, p_{n+1})\prod_{i = 1}^{n+1} g\Big(\underline{k} + \sum_{j <i} \underline{p}_{\pi(j)}\Big) \Big]
\end{split}
\end{equation}
where $S_{n}$ is the set of permutations $\pi$ of $\{1, \ldots, n\}$; $\mathbb{M}_{\beta} = (2\pi / \beta) (\mathbb{Z} + 1/2)$ is the set of the fermionic Matsubara frequencies, and $g(\underline{k})$ is given by (\ref{eq:gfou}).

In the following, it will be convenient to decompose the fermionic propagator in a singular plus a regular part. Let $\chi(\cdot)$ be a smooth cutoff function as in (\ref{eq:cutoff}). We write:
\begin{equation}\label{eq:gsplit}
\begin{split}
g(\underline{k}) &= g(\underline{k})\chi(|\hat H(k) - \mu|/\Delta) +  g(\underline{k})(1-\chi(|\hat H(k) - \mu|/\Delta)) \\
&= g_{\text{a}}(\underline{k}) + g_{\text{b}}(\underline{k})\;;
\end{split}
\end{equation}
the propagator $g_{\text{b}}$ satisfies the estimate:
\begin{equation}
\| \text{d}_{k_{0}}^{n_{0}} \text{d}_{k_{1}}^{n_{1}} g_{\text{b}}(\underline{k}) \| \leq \frac{C_{n_{0}, n_{1}}}{1+|k_{0}|^{n_{0}+1}}\qquad \text{uniformly in $k$.}
\end{equation}
Consider now the first term in (\ref{eq:gsplit}). It is:
\begin{equation}
g_{\text{a}}(\underline{k}) = \sum_{\omega=1}^{N} \frac{\chi(|e_{\omega}(k) - \mu|/\Delta)}{ik_{0}+e_{\omega}(k)-\mu} P_{\omega}(k)\;,
\end{equation}
with $P_{\omega}(k)$ a rank-one projector,
\begin{equation}
P_{\omega}(k):= |\xi_{\omega}(k) \rangle \langle \xi_{\omega}(k)|\qquad\qquad\hat H(k)\xi_{\omega}(k)=e_{\omega}(k)\xi_{\omega}(k)\;.
\end{equation}
We can further rewrite $g_{\text{a}}(\underline{k})$, as follows. By assumption, the Fermi point $k_{F}^{\omega}$ is not a critical point for $e_{\omega}$, that is $v_{\omega} =e'_{\omega}(k_{F}^{\omega})\neq 0$. Thus, setting $\underline{k} := (k_{0}, k_{1})$ and $\underline{k}_{F}^{\omega} := (0, k_{F}^{\omega})$,
\begin{equation}\label{eq:omegadelta}
\sum_{\omega=1}^{N} \frac{\chi(|e_{\omega}(k) - \mu|/\Delta)}{ik_{0}+e_{\omega}(k)-\mu} P_{\omega}(k) =\sum_{\omega=1}^{N} \frac{\chi^{\omega}_{\delta}(\underline k)}{D_{\omega}(\underline k)} P_{\omega}(k) + \sum_{\omega = 1}^{N} \Big( \frac{\chi(|e_{\omega}(k) - \mu|/\Delta)}{ik_{0}+e_{\omega}(k)-\mu}  -  \frac{\chi^{\omega}_{\delta}(\underline k)}{D_{\omega}(\underline k)} \Big) P_{\omega}(k)
\end{equation}
where: $\chi^{\omega}_{\delta}(\underline{k}):=\chi(\left\|\underline{k} - \underline{k}_{F}^{\omega}\right\|_{\omega}/\delta)$, $\|\underline q\|_{\omega}^{2}:=q_{0}^{2}+v_{\omega}^{2}|q|_{\mathbb{T}}^{2}$ and $D_{\omega}(\underline{k}_{F}^{\omega} + \underline q):=iq_{0}+v_{\omega}q$. The parameter $\delta$ is chosen so that $\delta \ll \inf_{\omega\neq\omega'}|k_{F}^{\omega}-k_{F}^{\omega'}|_{\mathbb{T}}$ and
\begin{equation}\label{eq:omegaomega'}
|e_{\omega}(k)-\mu|\leq \Delta\qquad \forall k\in\mathcal B^{\omega}_{2\delta}(k_{F}^{\omega}):=\{k\in [0,2\pi] : | k- k_{F}^{\omega} |_{\mathbb{T}} < 2\delta / |v_{\omega}|\}
\end{equation}
for all $\omega$. Therefore, we can rewrite:
\begin{equation}
\begin{split}
g(\underline{k}) &= g_{\text{s}}(\underline{k}) + g_{\text{r}}(\underline{k})\qquad \text{with}\\
 g_{\text{s}}(\underline{k}) &:= \sum_{\omega=1}^{N} \frac{\chi^{\omega}_{\delta}(\underline k)}{D_{\omega}(\underline k)} P_{\omega}(k)\\
 g_{\text{r}}(\underline{k}) &:= g_{\text{b}}(\underline{k}) + \sum_{\omega = 1}^{N} \Big( \frac{\chi(|e_{\omega}(k) - \mu|/\Delta)}{ik_{0}+e_{\omega}(k)-\mu}  -  \frac{\chi^{\omega}_{\delta}(\underline k)}{D_{\omega}(\underline k)} \Big) P_{\omega}(k)\;.
\end{split}
\end{equation}
It is not difficult to see that:
\begin{equation}\label{eq:estgr}
\| g_{\text{r}}(\underline{k})\| \leq \frac{C}{1+|k_{0}|}\;,\qquad \big\| \text{d}_{k_{\alpha}} g_{\mathrm{r}}(\underline{k}) \big\| \leq \sum_{\omega = 1}^{N} \frac{C}{1+|k_{0}|} \frac{1}{\|\underline{k} - \underline{k}_{F}^{\omega}\|}\;.
\end{equation}
Using the above splitting we can decompose the left-hand side of (\ref{eq:corrwick}) as:
\begin{equation}\label{eq:split1d}
\frac{1}{\beta L}\langle {\bf T}\, \hat n_{\underline{p}_{1}}; \cdots ; \hat n_{\underline{p}_{n}}; \hat\jmath_{\nu,\underline{p}_{n+1}}  \rangle_{\beta,L} = \widetilde S^{\beta,L}_{n;\nu}(\underline{p}_{1}, \ldots, \underline{p}_{n}) + \widetilde R^{\beta,L}_{n;\nu}(\underline{p}_{1}, \ldots, \underline{p}_{n})
\end{equation}
where $\widetilde S^{\beta, L}_{n;\nu}(\underline{p}_{1}, \ldots, \underline{p}_{n})$ collects the contribution due only to the singular propagators $g_{\mathrm{s}}$,
\begin{equation}\label{eq:tildes}
\widetilde S^{\beta, L}_{n;\nu}(\underline{p}_{1}, \ldots, \underline{p}_{n}) := -\frac{1}{\beta L} \sum_{\underline{k} \in \mathbb{M}_{\beta} \times B_{L}} \sum_{\pi\in S_{n}}\Tr \Big[ \hat J_{\nu}(k, p_{n+1})\prod_{i = 1}^{n+1} g_{\mathrm{s}}\Big(\underline{k} + \sum_{j < i} \underline{p}_{\pi(j)}\Big) \Big]\;,
\end{equation}
while $\widetilde R^{\beta, L}_{n;\nu}(\underline{p}_{1}, \ldots, \underline{p}_{n})$ contains at least one bounded propagator $g_{\mathrm{r}}$. 
%
%
%
%
%
The next lemma allows to rewrite $\widetilde S^{\beta,L}_{n;\nu}$ in a more explicit way, up to subleading terms. In the following, we shall use the notation:
\begin{equation}
\int_{\beta, L} \frac{d\underline{k}}{(2\pi)^{2}}\, [\cdots] := \frac{1}{\beta L}  \sum_{\underline{k} \in \mathbb{M}_{\beta} \times B_{L}} [\cdots]\;;
\end{equation}
in fact, as $\beta,L\to \infty$, the sum converges to an integral over $\underline{k} \in \mathbb{R} \times \mathbb{T}$.
\begin{lemma}[Structure of the singular part]\label{lemma:singular}
Let $v_{\omega}^{0} = 1$ and $v_{\omega}^{1} = v_{\omega}$. Then for any $p_{1},\dots,p_{n}$ in the support of $\hat\mu_{\alpha,\theta}$,
\begin{equation}\label{eq:splitting}
\widetilde S^{\beta,L}_{n;\nu}(\underline{p}_{1}, \ldots, \underline{p}_{n}) = S^{\beta,L}_{n;\nu}(\underline{p}_{1}, \ldots, \underline{p}_{n}) + T^{\beta,L}_{n;\nu}(\underline{p}_{1}, \ldots, \underline{p}_{n})
\end{equation}
where, for $g_{\omega}(\underline{k}) = \chi_{\delta}^{\omega}(\underline{k}) / D_{\omega}(\underline{k})$,
\begin{equation}\label{eq:singular}
\begin{split}
S^{\beta,L}_{n;\nu}(\underline{p}_{1}, \ldots, \underline{p}_{n}) &= - \sum_{\omega = 1}^{N} v_{\omega}^{\nu} \int_{\beta,L} \frac{d\underline{k}}{(2\pi)^{2}} \sum_{\pi\in S_{n}} \prod_{i=1}^{n+1} g_{\omega}\Big(\ul k+\sum_{j< i} \ul p_{\pi(j)}\Big)\\
T^{\beta,L}_{n;\nu}(\underline{p}_{1}, \ldots, \underline{p}_{n}) &= - \sum_{\omega = 1}^{N} \int_{\beta,L} \frac{d\underline{k}}{(2\pi)^{2}} \sum_{\pi\in S_{n}}f_{n,\omega}^{\nu}(k;p_{\pi(1)},\dots,p_{\pi(n)}) \prod_{i=1}^{n+1} g_{\omega}\Big(\ul k+\sum_{j< i} \ul p_{\pi(j)}\Big)\;,
\end{split}
\end{equation}
where:
\begin{equation}\label{eq:rembound}
|f_{n,\omega}^{\nu}(k;p_{1},\dots,p_{n})|\leq C_{\delta}\Big( \delta_{\nu,1}|k-k_{F}^{\omega}|_{\mathbb{T}}+ n\sum_{i=1}^{n} |p_{i}|_{\mathbb{T}}\Big)\;.
\end{equation}
\end{lemma}
\begin{proof} We start with the following remark. Since all external momenta $\{p_{j}\}$ are in the support of $\hat\mu_{\alpha,\theta}$, we have that $|p_{j}|_{\mathbb{T}} \leq \theta^{1-\alpha}$. Furthermore, $g_{\omega}$ and $g_{\omega'}$ have disjoint support for $\omega \neq \omega'$, recall the discussion after (\ref{eq:omegaomega'}). Thus, for $\theta$ small enough, and for all $i=1,\dots, n$:
\begin{equation}
g_{\omega}(\ul k+\textstyle\sum_{j<i}\ul p_{j})g_{\omega'}(\ul k+\sum_{j<i+1}\ul p_{j})=0\quad\mathrm{if\ \omega'\neq\omega}.
\end{equation}
Hence, $\tilde S^{\beta, L}_{n;\nu}$ reduces to
\begin{equation}
\widetilde S^{\beta,L}_{n;\nu}(\underline{p}_{1}, \ldots, \underline{p}_{n})=-\sum_{\omega=1}^{N}\sum_{\pi\in S_{n}} \int_{\beta,L} \frac{d\underline{k}}{(2\pi)^{2}}\, \Big[ \prod_{i=1}^{n+1} g_{\omega} \Big(\ul k+\sum_{j< i} \ul p_{\pi(j)}\Big)\Big] t_{\omega}^{\nu}(k;p_{\pi(1)},\dots,p_{\pi(n)})
\end{equation}
with
\begin{equation}
t_{\omega}^{\nu}(k;p_{\pi(1)},\dots,p_{\pi(n)}):= \Tr\Big[\hat J_{\nu}(k, p_{n+1})\prod_{i=1}^{n+1} P_{\omega}\Big(k+\sum_{j< i} p_{\pi(j)}\Big)\Big]\;.
\end{equation}
Noticing that $t_{\omega}^{0}(k;0,\dots,0)\equiv1$ and $t_{\omega}^{1}(k;0,\dots,0)=e'_{\omega}(k)$, we Taylor expand $t_{\omega}^{\nu}(k;\cdot)$ around $(0,\dots,0)$, and $e'_{\omega}$ around $k_{F}^{\omega}$, obtaining
\begin{equation}
t_{\omega}^{\nu}(k;p_{1},\dots,p_{n})= v_{\omega}^{\nu} + f^{\nu}_{n,\omega}(k;p_{1},\dots,p_{n})\;;
\end{equation}
the bound (\ref{eq:rembound}) easily follows from the smoothness of $\hat J_{\nu}(k,p)$ and of $P_{\omega}(k)$.
\end{proof}
\begin{remark} This decomposition allows to isolate, at order $n$, the most singular contribution from the $n$-th order term in the Duhamel expansion. Indeed, for $\theta\propto\eta$, a simple rescaling shows that each summand appearing in $\theta^{n-1} S^{\beta,L}_{n;\nu}$ is $O(1)$, whereas the summands in $\theta^{n-1} T^{\beta,L}_{n;\nu}$ and $\theta^{n-1} \widetilde R^{\beta,L}_{n;\nu}$ are $o(1)$. We will see that the smallness of $\theta^{n-1} S^{\beta,L}_{n;\nu}$ will be guaranteed by a crucial cancellation.
\end{remark}
Let $R^{\beta,L}_{n;\nu} := \widetilde R^{\beta,L}_{n;\nu} + T^{\beta,L}_{n;\nu}$. We obtained the following representation for the Duhamel expansion of the response functions, recall (\ref{eq:duha1})-(\ref{eq:fou2}),
\begin{equation}\label{eq:fullresp}
\begin{split}
\chi^{\beta,L}_{\nu}(x;\eta,\theta) &=
-\sum_{n=1}^\infty \frac{(-\theta)^{n-1}}{n!} \frac{1}{L^{n}} \sum_{\{p_{i}\} \in B_{L}^{n}} \Big[\prod_{j=1}^{n} \hat\mu_{\alpha,\theta}(-p_{j})\Big] e^{ip_{n+1}x}\\&\quad \cdot [S^{\beta,L}_{n;\nu}(\underline{p}_{1},\dots, \underline{p}_{n})+R^{\beta,L}_{n;\nu}(\underline{p}_{1},\dots, \underline{p}_{n})] + E^{\beta,L}_{\nu}(x;\eta,\theta)\;,
\end{split}
\end{equation}
where: $\underline{p}_{i} := (\eta_{\beta}, p_{i})$ and $\ul p_{n+1} := -\sum_{i=1}^{n} \ul p_{i}$; the term $S^{\beta,L}_{n;\nu}$ is given by the first of (\ref{eq:singular}); the term $R^{\beta,L}_{n;\nu}(\underline{p}_{1},\dots, \underline{p}_{n})$ will be proven to be subleading as $\theta \to 0$; and, recall (\ref{eq:estLR}):
\begin{equation}\label{eq:errnu}
| E^{\beta,L}_{\nu}(x;\eta,\theta)|\leq  \frac{C\theta^{\alpha(m-1)}}{\eta^{2}} + \frac{C}{\eta^{3} \beta}\;.
\end{equation}
Our next task will be to evaluate explicitly the linear response, and to bound the higher order terms, starting from the identity (\ref{eq:fullresp}).
\subsection{Evaluation of the linear response}
In this subsection we will evaluate the main term in the expression for the response function, Eq. (\ref{eq:main1d}), as $\beta,L\to \infty$. The analysis is based on lattice conservation laws, regularity of correlations, and explicit computation of the scaling limit contribution. It has been already used to determine the linear response of $1d$ and quasi-$1d$ systems, \cite{BM1, BM2, MPdrude, AMP, MPmulti}, in a setting that also allow to include many-body interactions. For completeness, we shall reproduce all the steps here, for general $1d$ non-interacting systems.

Let $S_{n;\nu} = \lim_{\beta,L\to \infty}S^{\beta,L}_{n;\nu}$ and $R_{n;\nu} = \lim_{\beta,L\to \infty}R^{\beta,L}_{n;\nu}$. In this limit, the $n=1$ term in (\ref{eq:fullresp}) reads, setting $\ul p=\theta\ul q$, with $\ul q=(a^{-1},q)$,
\begin{equation}\label{eq:linresp}
\chi_{\nu}^{\mathrm{lin}}(x;\eta,\theta) = - \int_{\mathbb{T}_{\theta^{-1}}} \frac{dq}{2\pi}\, \hat\mu_\alpha(-q)e^{-i\theta q x} \big[S_{1;\nu}(\theta\ul q) + R_{1;\nu}(\theta\ul q)\big]\;,
\end{equation}
with $\hat \mu_{\alpha}(q) = \hat \mu(q) \chi( \theta^{\alpha} |q|)$, and $\mathbb{T}_{\theta^{-1}}$ is the torus $\mathbb{T}$ rescaled by a factor $1/\theta$. Let us start by discussing the $R_{1;\nu}$ term.
\begin{lemma}[Continuity of the remainder]\label{lemma:rem1}
For $\alpha\in(0,1)$, there exists $C_{\alpha} > 0$ such that:
\begin{equation}\label{eq:Rest}
|R_{1;\nu}(\underline{0})| \leq C,\qquad |R_{1;\nu}(\ul p)-R_{1;\nu}(\ul 0)|\leq C_{\alpha} \|\ul p\|^{\alpha}\;.
\end{equation}
\end{lemma}
\begin{proof} We have:
\begin{equation}
\widetilde R_{1;\nu}(\ul p)= \int_{\R\times \mathbb{T}} \frac{d\ul k}{(2\pi)^2} \Tr\Big(\hat J_{\nu}(k, -p)(g_{\mathrm{r}}(\ul k) g_{\mathrm{s}}(\ul k+\ul p)+g_{\mathrm{s}}(\ul k)g_{\mathrm{r}}(\ul k+\ul p)+g_{\mathrm{r}}(\ul k)g_{\mathrm{r}}(\ul k+\ul p))\Big)\;.
\end{equation}
By the estimate (\ref{eq:estgr}), $\widetilde R_{1;\nu}(\ul 0)$ is finite. Furthermore, by the smoothness of $\hat J_{\nu}(k, p)$ and using the bound, with $0<\alpha<1$,
\begin{equation}
\big\| g_{\mathrm{r}}(\underline{k} + \underline{p}) - g_{\mathrm{r}}(\underline{k}) \big\| \leq C\sum_{\omega} \frac{1}{1+|k_{0}|^{\alpha}}\frac{\|\underline{p}\|^{\alpha}}{\|\underline{k} - \underline{k}_{F}^{\omega}\|^{\alpha}}\,
\end{equation}
which follows from the second of (\ref{eq:estgr}), we easily get:
\begin{equation}
\Big| \widetilde R_{1;\nu}(\ul p) - \widetilde R_{1;\nu}(\ul 0) \Big| \leq C_{\alpha} \|\underline{p}\|^{\alpha}\;.
\end{equation}
Consider now the term $T_{1;\nu}$. From the estimate (\ref{eq:rembound}), and recalling that $g_{\omega}(\underline{k}) = \chi_{\delta}^{\omega}(\underline{k}) / D_{\omega}(\underline{k})$, we have:
\begin{equation}
\begin{split}
\Big| T_{1;\nu}(\underline{p}) \Big| &\leq C\sum_{\omega = 1}^{N} \int d\underline{k}\, \big| f_{1,\omega}^{\nu}(k;p)\big| | g_{\omega}(\underline{k} + \underline{p}) | | g_{\omega}(\underline{k})| \\
&\leq K \sum_{\omega = 1}^{N} \int d\underline{k}\, (|k - k_{F}^{\omega}|_{\mathbb{T}} + |p|_{\mathbb{T}}) | g_{\omega}(\underline{k} + \underline{p})| | g_{\omega}(\underline{k})| \\
&\leq \widetilde{K}(1 + \|\underline{p}\| \big|\log \|\underline{p}\|\big|)\;.
\end{split}
\end{equation}
Similarly,
\begin{equation}
\begin{split}
\Big| T_{1;\nu}(\underline{p}) - T_{1;\nu}(\underline{0}) \Big| &\leq K \sum_{\omega = 1}^{N} \int d\underline{k}\, |k - k_{F}^{\omega}|_{\mathbb{T}} | \big( g_{\omega}(\underline{k} + \underline{p}) - g_{\omega}(\underline{k})\big) | | g_{\omega}(\underline{k})| \\
&\quad + K \sum_{\omega = 1}^{N} \int d\underline{k}\, |p|_{\mathbb{T}} | g_{\omega}(\underline{k} + \underline{p}) | |g_{\omega}(\underline{k})|\;; \\
\end{split}
\end{equation}
the second term is bounded proportionally to $\|\underline{p}\| \left| \log \|\underline{p}\|\right|$. Consider the first term. Using that, for $0<\alpha<1$:
\begin{equation}
| g_{\omega}(\underline{k} + \underline{p}) - g_{\omega}(\underline{k}) | \leq C\Big(\frac{\chi_{\omega}(\underline{k})}{\| \underline{k} - \underline{k}_{F}^{\omega}\|^{1-\alpha}} + \frac{\chi_{\omega}(\underline{k} + \underline{p})}{\|\underline{k} - \underline{k}_{F}^{\omega} + \underline{p}\|^{1-\alpha}}\Big) \frac{\|\underline{p}\|^{\alpha}}{\|\underline{k} - \underline{k}_{F}^{\omega} + \underline{p}\|^{\alpha} \|\underline{k} - \underline{k}_{F}^{\omega} \|^{\alpha}} 
\end{equation}
we obtain:
\begin{equation}
\begin{split}
&\int d\underline{k}\, |k - k_{F}^{\omega}| | \big( g_{\omega}(\underline{k} + \underline{p}) - g_{\omega}(\underline{k})\big) | |g_{\omega}(\underline{k})|  \\
&\qquad \leq C\int d\underline{k}\,\Big(\frac{\chi_{\omega}(\underline{k})}{\| \underline{k} - \underline{k}_{F}^{\omega}\|^{1-\alpha}} + \frac{\chi_{\omega}(\underline{k} + \underline{p})}{\|\underline{k} - \underline{k}_{F}^{\omega} + \underline{p}\|^{1-\alpha}}\Big) \frac{\|\underline{p}\|^{\alpha}}{\|\underline{k} - \underline{k}_{F}^{\omega} + \underline{p}\|^{\alpha} \|\underline{k} - \underline{k}_{F}^{\omega}\|^{\alpha}} \\
&\qquad \leq \widetilde{C}_{\alpha} \|p\|^{\alpha}.
\end{split}
\end{equation}
Thus, we found:
\begin{equation}
\Big| T_{1;\nu}(\underline{p}) - T_{1;\nu}(\underline{0}) \Big| \leq C_{\alpha} \|\underline{p}\|^{\alpha}\;.
\end{equation}
This concludes the proof of (\ref{eq:Rest}).
\end{proof}
Next, we shall consider the singular term $S_{1;\nu}$ in (\ref{eq:linresp}), whose explicit form is
\begin{equation}\label{eq:S2bubble}
\begin{split}
S_{1;\nu}(\theta\ul q) &= -\sum_{\omega=1}^{N} v_{\omega}^{\nu} \int_{\R\times \mathbb{T}} g_{\omega}(\ul k)g_{\omega}(\ul k+\theta\ul q) \frac{d\ul k}{(2\pi)^2} \\
&\equiv \sum_{\omega=1}^{N} v_{\omega}^{\nu} \mathfrak B_{\delta}^{\omega}(\theta\ul q)\;,
\end{split}
\end{equation}
where $ \mathfrak B_{\delta}^{\omega}$ is the relativistic bubble diagram:
\begin{equation}\label{eq:bubble}
\mathfrak B_{\delta}^{\omega}(\theta\ul q):= -\int_{\R^{2}} \frac{\chi_{\delta}^{\omega}(\ul k)\chi_{\delta}^{\omega}(\ul k+\theta\ul q)}{D_{\omega}(\ul k)D_{\omega}(\ul k+\theta\ul q)} \frac{d\ul k}{(2\pi)^2}\;.
\end{equation}
The next proposition allows to compute it. The result is well-known, and we reproduce it here for completeness.
\begin{proposition}[The relativistic bubble diagram]\label{lemma:bubble} Let $\alpha \in (0,1)$. For any $\ul q=(q_{0},q)\neq\ul 0$ such that $\|\ul q\|\leq \theta^{-\alpha}$, we have:
\begin{equation}\label{eq:bubblediag}
\mathfrak B_{\delta}^{\omega}(\theta\ul q) = \frac{1}{4\pi |v_{\omega}|}\frac{-iq_0+v_{\omega}q}{iq_0+v_{\omega}q} + O(\theta^{1-\alpha})\;.
\end{equation}
\end{proposition}
\begin{proof} First of all, by performing the change of coordinates $\ul k\to \theta\ul k$ we obtain
\begin{equation}\label{eq:bubblescaling}
\mathfrak B_{\delta}^{\omega}(\theta\ul q) = \mathfrak B_{\delta/\theta}^{\omega}(\ul q)\;.
\end{equation}
Then, we rewrite:
\begin{equation}
\begin{split}
\mathfrak B_{\delta/\theta}^{\omega}(\ul q) &= -\int_{\R^{2}} \frac{\chi_{\delta/\theta}^{\omega}(\ul k)\chi_{\delta/\theta}^{\omega}(\ul k+\ul q)}{D_{\omega}(\ul k)D_{\omega}(\ul k+\ul q)} \frac{d\ul k}{(2\pi)^2} \\
&= - \frac{1}{D_{\omega}(\underline{q})} \int_{\R^{2}} \chi_{\delta/\theta}^{\omega}(\ul k)\chi_{\delta/\theta}^{\omega}(\ul k+\ul q)\left( \frac{1}{D_{\omega}(\ul k)} - \frac{1}{D_{\omega}(\underline{k} + \underline{q})} \right) \frac{d\ul k}{(2\pi)^2} \\
&= - \frac{1}{D_{\omega}(\underline{q})} \int_{\R^{2}} \frac{\chi_{\delta/\theta}^{\omega}(\ul k)\chi_{\delta/\theta}^{\omega}(\ul k+\ul q) - \chi_{\delta/\theta}^{\omega}(\ul k - \underline{q})\chi_{\delta/\theta}^{\omega}(\ul k)}{D_{\omega}(\underline{k})}\frac{d\ul k}{(2\pi)^2}\;.
\end{split}
\end{equation}
Observing that:
\begin{equation}
\chi_{\delta/\theta}^{\omega}(\ul k+\ul q) - \chi_{\delta/\theta}^{\omega}(\ul k - \underline{q}) = 2\underline{q} \cdot \nabla_{\underline{k}} \chi_{\delta/\theta}^{\omega}(\ul k) + r_{\omega}(\underline{k}, \underline{q})
\end{equation}
with $| r_{\omega}(\underline{k}, \underline{q}) | \leq C_{\delta} \theta^{2} \|\underline{q}\|^{2}$, we see that, for $\|\underline{q}\| \leq \theta^{-\alpha}$, using that $r_{\omega}(\underline{k}, \underline{q})$ is supported for $\|\underline{k}\| \sim (\delta / \theta)$:
\begin{equation}
\left| \frac{1}{D_{\omega}(\underline{q})} \int_{\R^{2}} \frac{\chi_{\delta/\theta}^{\omega}(\ul k)r_{\omega}(\underline{k}, \underline{q})}{D_{\omega}(\underline{k})} \frac{d\ul k}{(2\pi)^2} \right| \leq C_{\delta}\theta^{-\alpha + 1}
\end{equation}
which vanishes for $\theta \to 0$. Consider now the main term,
\begin{equation}
\begin{split}
&- \frac{2}{D_{\omega}(\underline{q})} \int_{\R^{2}} \frac{\chi_{\delta/\theta}^{\omega}(\ul k) \underline{q}\cdot \nabla_{\underline{k}} \chi_{\delta/\theta}^{\omega}(\ul k)}{D_{\omega}(\underline{k})} \frac{d\ul k}{(2\pi)^2} \\
&\qquad = - \frac{2q_{0}}{D_{\omega}(\underline{q})} \int_{\R^{2}} \frac{\chi_{\delta/\theta}^{\omega}(\ul k) \partial_{0} \chi_{\delta/\theta}^{\omega}(\ul k)}{D_{\omega}(\underline{k})} \frac{d\ul k}{(2\pi)^2}  - \frac{2q_{1}}{D_{\omega}(\underline{q})} \int_{\R^{2}} \frac{\chi_{\delta/\theta}^{\omega}(\ul k) \partial_{1} \chi_{\delta/\theta}^{\omega}(\ul k)}{D_{\omega}(\underline{k})} \frac{d\ul k}{(2\pi)^2}\;.
\end{split}
\end{equation}
After a rescaling and a change of variable in the $k_{1}$ variable, we get;
\begin{equation}
\begin{split}
&- \frac{2}{D_{\omega}(\underline{q})} \int_{\R^{2}} \frac{\chi_{\delta/\theta}^{\omega}(\ul k) \underline{q}\cdot \nabla_{\underline{k}} \chi_{\delta/\theta}^{\omega}(\ul k)}{D_{\omega}(\underline{k})} \frac{d\ul k}{(2\pi)^2} \\
&\quad =  - \frac{2q_{0}}{D_{\omega}(\underline{q})} \frac{1}{|v_{\omega}|}\int_{\R^{2}} \frac{\chi(\ul k) \partial_{0} \chi(\ul k)}{ik_{0} + k_{1}} \frac{d\ul k}{(2\pi)^2}  - \frac{2q_{1}}{D_{\omega}(\underline{q})} \frac{v_{\omega}}{|v_{\omega}|} \int_{\R^{2}} \frac{\chi(\ul k) \partial_{1} \chi(\ul k)}{ik_{0} + k_{1}} \frac{d\ul k}{(2\pi)^2}\;,
\end{split}
\end{equation}
where $\chi(\underline{k}) \equiv \chi(\|\underline{k}\|)$, see Eq. (\ref{eq:cutoff}). Then, one observes that this expression can be further rewritten as:
\begin{equation}
\begin{split}
&- \frac{2}{D_{\omega}(\underline{q})} \int_{\R^{2}} \frac{\chi_{\delta/\theta}^{\omega}(\ul k) \underline{q}\cdot \nabla_{\underline{k}} \chi_{\delta/\theta}^{\omega}(\ul k)}{D_{\omega}(\underline{k})} \frac{d\ul k}{(2\pi)^2} \\
&\qquad = - \frac{2(q_{0} + i v_{\omega}q_{1})}{D_{\omega}(\underline{q}) |v_{\omega}|} \int_{\R^{2}} \frac{\chi(\ul k) \partial_{0} \chi(\ul k)}{ik_{0} + k_{1}} \frac{d\ul k}{(2\pi)^2} \\
&\qquad = - \frac{2(q_{0} + i v_{\omega}q_{1})}{D_{\omega}(\underline{q}) |v_{\omega}|} \int_{\R^{2}} \frac{k_{0}}{\|\underline{k}\|}\frac{\chi(\|\ul k\|) \chi'(\|\ul k\|)}{ik_{0} + k_{1}} \frac{d\ul k}{(2\pi)^2} \\
&\qquad = - \frac{q_{0} + i v_{\omega}q_{1}}{D_{\omega}(\underline{q}) |v_{\omega}|} \int_{\R^{2}} \frac{k_{0}}{\|\underline{k}\|}\frac{(\chi^{2}(\|\ul k\|))'}{ik_{0} + k_{1}} \frac{d\ul k}{(2\pi)^2}\;. 
\end{split}
\end{equation}
The last integral can be computed switching to polar coordinates, and one finds:
\begin{equation}
- \frac{2}{D_{\omega}(\underline{q})} \int_{\R^{2}} \frac{\chi_{\delta/\theta}^{\omega}(\ul k) \underline{q}\cdot \nabla_{\underline{k}} \chi_{\delta/\theta}^{\omega}(\ul k)}{D_{\omega}(\underline{k})} \frac{d\ul k}{(2\pi)^2} = \frac{q_{0} + i v_{\omega}q_{1}}{D_{\omega}(\underline{q}) |v_{\omega}|} \frac{-i}{4\pi}
\end{equation}
which reproduces the main term in (\ref{eq:bubblediag}).
\end{proof}
The evaluation of the relativistic bubble diagram, combined with lattice conservation laws, allow to compute (\ref{eq:linresp}). 
\begin{proposition}[Evaluation of the linear response]\label{prop:linear} We have:
\begin{equation}\label{eq:linear}
\chi_{\nu}^{\mathrm{lin}}(x;\eta,\theta) =-\sum_{\omega=1}^{N}\chi_{\nu,\omega}\int_\R \hat\mu_{\infty}(q)e^{iq\theta x}\frac{v_\omega q}{-i/a+v_\omega q}\,\frac{dq}{2\pi} + O(\theta^{1-\alpha})\;,
\end{equation}
with $\chi_{\nu,\omega}=v_{\omega}^{\nu}/2\pi |v_\omega|$. Let $\theta = a \eta$. We can distinguish three regimes.
\begin{enumerate}
\item[(i)] Suppose that $a\to 0$ as $\eta\to0^{+}$. Then $\chi_{\nu}^{\mathrm{lin}}(x;\eta,\theta) = O(a)$.
\item[(ii)] Suppose that $a$ is constant in $\eta$. Then:
\begin{equation}\label{eq:realformula}
\chi_{\nu}^{\mathrm{lin}}(x;\eta,\theta)=-\sum_{\omega=1}^{N}\chi_{\nu,\omega}\int_\R \mu_{\infty}(\theta x-y)\Big[\delta(y)-\frac{1}{a|v_{\omega}|} e^{-y/av_{\omega}}\Theta(y/v_{\omega})\Big]\,dy+ O(\eta^{1-\alpha})\;.
\end{equation}
\item[(iii)] Suppose that $a\to \infty$ as $\eta \to 0^{+}$, so that $a\eta \to 0$. Then:
\begin{equation}\label{eq:ainfty}
\chi_{\nu}^{\mathrm{lin}}(x;\eta,\theta)= - \mu_{\infty}(\theta x)\sum_{\omega=1}^{N}\chi_{\nu,\omega} + O(\theta^{1-\alpha})\;.
\end{equation}
\end{enumerate}
\end{proposition}
\begin{proof}
By (\ref{eq:S2bubble}), we rewrite the linear response (\ref{eq:linresp}) as:
\begin{equation}
\begin{split}
\chi_{\nu}^{\mathrm{lin}}(x,\eta,\theta) &= - \int_{\mathbb{T}_{\theta^{-1}}} \frac{dq}{2\pi}\, \hat\mu_\alpha(-q)e^{-i\theta q x} \big[S_{1;\nu}(\theta\ul q) + R_{1;\nu}(\theta\ul q)\big] \\
&= -\int_{\R} \frac{dq}{2\pi}\, \hat\mu_{\infty}(-q)e^{-i\theta q x}\Big[  \sum_{\omega=1}^{N} v_{\omega}^{\nu} \mathfrak B_{\delta}^{\omega}(\theta\ul q) + R_{1;\nu}(\ul 0)\Big] + E_{\nu;1}(x;\eta,\theta)\;,
\end{split}
\end{equation}
where, by (\ref{eq:Rest}), and by the regularity properties of the function $\mu_{\infty}(x)$:
\begin{equation}
|E_{\nu;1}(x;\eta,\theta) | \leq K \theta\;. 
\end{equation}
Next, by (\ref{eq:bubblediag}), (\ref{eq:bubblescaling}), $\mathfrak B_{\delta}^{\omega}(\theta\ul q) = \mathfrak B_{\infty}^{\omega}(\ul q) + O(\theta^{1-\alpha})$, with:
\begin{equation}
\mathfrak B_{\infty}^{\omega}(\ul q)=\frac{1}{4\pi |v_{\omega}|}\frac{-iq_{0}+v_{\omega}q}{iq_{0}+ v_{\omega}q}\;;
\end{equation}
thus, recalling that $\underline{q} =  (a^{-1},q)$ and performing the change of variables $q\to-q$,
\begin{equation}\label{eq:linresp2}
\begin{split}
&\chi_{\nu}^{\mathrm{lin}}(x;\eta,\theta) \\&\quad = -\int_{\R} \frac{dq}{2\pi}\, \hat\mu_{\infty}(q)e^{i\theta q x}\Big[  \sum_{\omega=1}^{N} v_{\omega}^{\nu} \frac{1}{4\pi |v_{\omega}|}\frac{i/a+ v_{\omega}q}{-i/a+ v_{\omega}q} + R_{1;\nu}(\ul 0)\Big] + \sum_{i=1,2} E_{\nu;i}(x;\eta,\theta) 
\end{split}
\end{equation}
where:
\begin{equation}
| E_{\nu;2}(x;\eta,\theta) | \leq C\theta^{1-\alpha}\;.
\end{equation}
Let us now determine $R_{1;\nu}(\ul 0)$. We shall compute it by exploiting lattice conservation laws. Thanks to the continuity equation (\ref{eq:cons1d}), the following equality holds:
\begin{equation}
i\partial_{x_{0}}\langle\mathbf T n_{\ul x};j_{\nu,\ul y}\rangle_{\beta,L} + \mathrm d_{x} \langle\mathbf T j_{\ul x};j_{\nu,\ul y}\rangle_{\beta,L} = i\delta(x_{0}-y_{0}) \langle [n_{\ul x},j_{\nu,\ul y}]\rangle_{\beta,L}\;,
\end{equation}
where $\mathcal O_{\ul x}:=\gamma_{x_{0}}(\mathcal O_{x_{1}})$. The contact term on the right-hand side arises due to the definition of time-ordering. Taking the Fourier transform of left-hand side and right-hand side, we obtain:
\begin{equation}
p_{0} \frac{1}{\beta L}\langle\mathbf T \hat n_{\ul p};\hat\jmath_{\nu,-\ul p}\rangle_{\beta,L} + (1-e^{-ip}) \frac{1}{\beta L}\langle\mathbf T \hat\jmath_{\ul p};\hat\jmath_{\nu,-\ul p}\rangle_{\beta,L} = i\sum_{x\in\Z} e^{-ipx}\langle[n_{x},j_{\nu,0}]\rangle_{\beta,L}\;.
\end{equation}
Let $\underline{p} = (p_{0}, 0)$. Using that $[\sum_{x\in\Z} n_{x},j_{\nu,0}]=0$, we obtain $\langle\mathbf T \hat n_{\ul p};\hat\jmath_{\nu,-\ul p}\rangle_{\beta,L} = 0$. Hence, recalling that $(1/\beta L)\langle {\bf T} \hat n_{\ul p};\hat \jmath_{\nu,-\ul p}\rangle_{\beta,L} = S^{\beta,L}_{1;\nu}(\ul p)+R^{\beta,L}_{1;\nu}(\ul p)$, we obtain, in the $\beta, L \to \infty$ limit, and by the continuity of $R_{1;\nu}(\ul p)$ at $\underline{p} = \underline{0}$:
\begin{equation}
S_{1;\nu}((p_{0},0)) + R_{1;\nu}((p_{0},0))=0\qquad\Longrightarrow\qquad R_{1;\nu}(\ul 0)=-\lim_{p_{0}\to 0} S_{1;\nu}((p_{0},0))=\sum_{\omega=1}^{N} \frac{v_{\omega}^{\nu}}{4\pi|v_{\omega}|}\;.
\end{equation}
Plugging this identity in (\ref{eq:linresp2}), we obtain:
\begin{equation}\label{eq:lincomp}
\chi_{\nu}^{\mathrm{lin}}(x;\eta,\theta) =-\sum_{\omega=1}^{N}\frac{v_{\omega}^{\nu}}{2\pi |v_\omega|} \int_\R \frac{dq}{2\pi}\, \hat\mu_{\infty}(q)e^{iq\theta x}\frac{v_\omega q}{-i/a+v_\omega q} + \sum_{i=1,2} E_{\nu;i}(x;\eta,\theta)
\end{equation}
which proves (\ref{eq:linear}). Let us further analyze the integral. From (\ref{eq:lincomp}), it is clear that if $a \to 0$, the integral is vanishing. More generally, we can rewrite (\ref{eq:lincomp}) as:
\begin{equation}\label{eq:linx}
\chi_{\nu}^{\mathrm{lin}}(x;\eta,\theta) =-\sum_{\omega=1}^{N}\chi_{\nu,\omega}\Big[\mu_{\infty}(\theta x)+\frac{i}{a} \int_\R \frac{dq}{2\pi}\, \frac{\hat\mu_{\infty}(q)e^{iq\theta x}}{-i/a+v_\omega q}\Big] + \sum_{i=1,2} E_{\nu;i}(x;\eta,\theta)\;.
\end{equation}
To compute the integral, we employ the good decay properties of $\hat \mu_{\infty}(q)$ to rewrite it as:
\begin{equation}\label{eq:limQ}
\begin{split}
\int_\R \frac{\hat\mu_{\infty}(q)e^{iq\theta x}}{-i/a+v_\omega q}\,\frac{dq}{2\pi} &= \lim_{Q\to \infty} \int_{|q|\leq Q} \frac{\hat\mu_{\infty}(q)e^{iq\theta x}}{-i/a+v_\omega q}\,\frac{dq}{2\pi} \\
&= \lim_{Q\to \infty}  \int dy\, \mu_{\infty}(y) \int_{|q|\leq Q} \frac{e^{-iqy}e^{iq\theta x}}{-i/a+v_\omega q}\,\frac{dq}{2\pi} \\
&= \lim_{Q\to \infty}  \int dz\, \mu_{\infty}(z + \theta x) \int_{|q|\leq Q} \frac{e^{-iqz}}{-i/a+v_\omega q}\,\frac{dq}{2\pi}\;.
\end{split}
\end{equation}
Let $z>0$. By Cauchy theorem for holomorphic functions:
\begin{equation}
\begin{split}
\int_{|q|\leq Q} \frac{e^{-iqz}}{-i/a+v_\omega q}\,\frac{dq}{2\pi} &= \frac{1}{2\pi v_{\omega}} \int_{|q|\leq Q} \frac{e^{-iqz}}{-i/(av_{\omega}) + q}\, dq \\
&= -\frac{1}{2\pi v_{\omega}} (2\pi i) e^{z / (av_{\omega})} \mathbbm{1}(v_{\omega} < 0) + g^{+}_{Q}(z) 
\end{split}
\end{equation}
where $g^{+}_{Q}(z)$ is bounded uniformly in $z$ and $\lim_{Q\to \infty}g^{+}_{Q}(z) = 0$. Similarly, if $z<0$:
\begin{equation}
\int_{|q|\leq Q} \frac{e^{-iqz}}{-i/a+v_\omega q}\,\frac{dq}{2\pi} = \frac{1}{2\pi v_{\omega}} (2\pi i) e^{z / (av_{\omega})} \mathbbm{1}(v_{\omega} > 0) + g^{-}_{Q}(z) 
\end{equation}
with $g^{-}_{Q}(z)$ is bounded uniformly in $z$ and $\lim_{Q\to \infty}g^{-}_{Q}(z) = 0$. All in all, for $z\neq 0$:
\begin{equation}
\lim_{Q\to \infty} \int_{|q|\leq Q} \frac{e^{-iqz}}{-i/a+v_\omega q}\,\frac{dq}{2\pi} = \frac{i}{|v_{\omega}|} e^{z / (av_{\omega})} \mathbbm{1}((z/v_{\omega})<0).
\end{equation}
Plugging this in (\ref{eq:limQ}), we obtain, performing a change of variable $z\to -z$:
\begin{equation}
\int_\R \frac{\hat\mu_{\infty}(q)e^{iq\theta x}}{-i/a+v_\omega q}\,\frac{dq}{2\pi} =  \frac{i}{|v_{\omega}|} \int dz\, \mu_{\infty}(\theta x - z)  e^{-z / (av_{\omega})} \Theta(z/v_{\omega})
\end{equation}
with $\Theta(\cdot)$ the Heaviside step function, equal to $1$ for positive argument and zero otherwise. Inserting this formula in (\ref{eq:linx}), we get:
\begin{equation}
\begin{split}
\chi_{\nu}^{\mathrm{lin}}(x;\eta,\theta) &= -\sum_{\omega=1}^{N}\chi_{\nu,\omega}\Big[\mu_{\infty}(\theta x) - \frac{1}{a|v_{\omega}|} \int dz\, \mu_{\infty}(\theta x - z)  e^{-z / (av_{\omega})} \Theta(z/v_{\omega})\Big]\\&\quad + \sum_{i=1,2} E_{\nu;i}(x;\eta,\theta)
\end{split}
\end{equation}
which proves (\ref{eq:realformula}). Finally, the last claim (\ref{eq:ainfty}) follows after taking the $a\to \infty$ limit, and using that $\mu$ is integrable. This concludes the proof of Proposition \ref{prop:linear}.
\end{proof}

\subsection{Estimates for the higher order corrections}
Let:
\begin{equation}\label{eq:genbubble}
\mathfrak B_{n+1}^{\omega}(\ul p_{1},\dots,\ul p_{n}) := -\int_{\R^{2}} \prod_{i=1}^{n+1} g_{\omega}\Big(\ul k+\sum_{j< i} \ul p_{j}\Big) \frac{d\ul k}{(2\pi)^2}\;,
\end{equation}
with the understading that $\underline{p}_{n+1} = -\sum_{j=1}^{n} \underline{p}_{j}$. Thus, from (\ref{eq:singular}), in the $\beta, L\to \infty$ limit:
\begin{equation}\label{eq:bubbles}
S_{n;\nu}(\ul p_{1}, \ldots, \underline{p}_{n}) = \sum_{\omega=1}^{N} v^{\nu}_{\omega}\sum_{\pi\in S_{n}} \mathfrak B_{n+1}^{\omega}(\ul p_{\pi(1)},\dots,\ul p_{\pi(n)})\;.
\end{equation}
Our main goal will be to show that the are non-trivial cancellations in the sum over permutations.  These cancellations are ultimately implied by the identity:
\begin{equation}\label{eq:ward}
g_{\omega}(\underline{k}) g_{\omega}(\underline{k} + \underline{p}) = \frac{1}{D_{\omega}(\ul p)} \Big(g_{\omega}(\underline{k}) - g_{\omega}(\underline{k} + \underline{p})\Big) + F^{\omega}_{\delta}(\ul k;\ul p)
\end{equation}
where:
\begin{equation}\label{eq:defF}
F^{\omega}_{\delta}(\ul k;\ul p):= \frac{\chi_{\delta}^{\omega}(\ul k)[\chi_{\delta}^{\omega}(\ul k+\ul p)-1]}{D_{\omega}(\ul k)D_{\omega}(\ul k+\ul p)}+\frac{\chi_{\delta}^{\omega}(\ul k+\ul p)-\chi_{\delta}^{\omega}(\ul k)}{D_{\omega}(\ul p)D_{\omega}(\ul k+\ul p)}\;.
\end{equation}
This identity is a tree-level Ward identity, for the vertex function. It can be viewed as generated by the invariance of the relativistic $1+1$ dimensional model under chiral gauge transformations. The term $F^{\omega}_{\delta}(\ul k;\ul p)$ in the right-hand side introduces an anomaly in the usual Ward identity, due to the presence of the momentum cut-off. The starting point is the following proposition. A similar cancellation, for chiral fermions in the absence of UV cutoff, has been discussed in \cite{FGM2}.
\begin{proposition}[Loop cancellation.]\label{prp:loopcanc}
For any $\omega=1,\dots, N$, the following holds:
\begin{equation}\label{eq:sumloops}
\begin{split}
&\sum_{\pi\in S_{n}} \mathfrak B_{n+1}^{\omega}(\ul p_{\pi(1)},\dots,\ul p_{\pi(n)}) \\
&\qquad = -\frac{1}{n+1}\sum_{\pi\in S_{n+1}}\int F_{\delta}^{\omega}(\ul k;\ul p_{\pi(1)})\prod_{i=3}^{n+1} g_{\omega}\Big(\ul k+\sum_{j<i} \ul p_{\pi(l)}\Big)\frac{d\ul k}{(2\pi)^2}\;.
\end{split}
\end{equation}
\end{proposition}
\begin{proof} To begin, we define:
\begin{equation}
\mathfrak B_{n+1}^{\omega}(\ul p_{1},\dots,\ul p_{n}, \ul p_{n+1}) := \mathfrak B_{n+1}^{\omega}(\ul p_{1},\dots,\ul p_{n})\;.
\end{equation}
%
Observe that $\mathfrak B_{n+1}^{\omega}(\ul p_{1},\dots,\ul p_{n}, \ul p_{n+1})$ is invariant under the abelian subgroup of maximal cycles $\sigma$ in $\{1,2,\ldots, n+1\}$, 
\begin{equation}\label{eq:cycles}
\mathfrak B_{n+1}^{\omega}(\ul p_{1},\dots,\ul p_{n+1})=\mathfrak B_{n+1}^{\omega}(\ul p_{\sigma(1)},\dots,\ul p_{\sigma(n)},\ul p_{\sigma(n+1)})\;.
\end{equation}
In fact, let us consider the integrand in (\ref{eq:genbubble}):
\begin{equation}
g_{\omega}(\underline{k}) g_{\omega}(\underline{k} + \underline{p}_{1}) g_{\omega}(\underline{k} + \underline{p}_{1} + \underline{p}_{2}) \cdots g_{\omega}(\underline{k} + \underline{p}_{1} + \ldots + \underline{p}_{n})\;.  
\end{equation}
Using that $\underline{p}_{n+1} = - \underline{p}_{1} - \ldots - \underline{p}_{n}$ we can rewrite it as:
\begin{equation}\label{eq:ggg}
g_{\omega}(\underline{k} + \underline{p}_{1} + \ldots + \underline{p}_{n+1}) g_{\omega}(\underline{k} + \underline{p}_{1}) g_{\omega}(\underline{k} + \underline{p}_{1} + \underline{p}_{2}) \cdots g_{\omega}(\underline{k} + \underline{p}_{1} + \ldots + \underline{p}_{n})\;.  
\end{equation}
Now, perform the change of variables $\underline{k} + \underline{p}_{1} \to \underline{k}$ in the integral in the definition (\ref{eq:genbubble}). We get:
\begin{equation}
g_{\omega}(\underline{k} + \underline{p}_{2} + \ldots + \underline{p}_{n+1}) g_{\omega}(\underline{k}) g_{\omega}(\underline{k} + \underline{p}_{2}) \cdots g_{\omega}(\underline{k} + \underline{p}_{2} + \ldots + \underline{p}_{n})\;,
\end{equation}
which we can also rewrite as, after a rearrangement of the factors:
\begin{equation}
g_{\omega}(\underline{k} + \underline{p}_{1} + \ldots + \underline{p}_{n+1})   g_{\omega}(\underline{k} + \underline{p}_{2}) \cdots g_{\omega}(\underline{k} + \underline{p}_{2} + \ldots + \underline{p}_{n}) g_{\omega}(\underline{k} + \underline{p}_{2} + \ldots + \underline{p}_{n+1})
\end{equation}
which is precisely what one obtains from (\ref{eq:ggg}) replacing the ordered sequence of external momenta $\{\underline{p}_{1}, \underline{p}_{2}, \ldots, \underline{p}_{n}, \underline{p}_{n+1}\}$ with $\{\underline{p}_{2}, \underline{p}_{3}, \ldots, \underline{p}_{n+1}, \underline{p}_{1}\}$. This proves (\ref{eq:cycles}).

Next, we rewrite the left-hand side of (\ref{eq:sumloops}) as
\begin{equation}
\sum_{\pi\in S_{n}} \mathfrak B_{n+1}^{\omega}(\ul p_{\pi(1)},\dots,\ul p_{\pi(n)},\ul p_{n+1})= \sum_{\substack{\pi\in S_{n}\\ \pi(n+1)=n+1}} \mathfrak B_{n+1}^{\omega}(\ul p_{\pi(1)},\dots,\ul p_{\pi(n+1)})\;;
\end{equation}
using the aforementioned cyclicity property, we can also rewrite:
\begin{equation}\label{eq:sumpert}
\sum_{\substack{\pi\in S_{n}\\ \pi(n+1)=n+1}} \mathfrak B_{n+1}^{\omega}(\ul p_{\pi(1)},\dots,\ul p_{\pi(n+1)})=\frac{1}{n+1}\sum_{i=1}^{n+1}\sum_{\substack{\pi\in S_{n+1}\\ \pi(i)=n+1}} \mathfrak B_{n+1}^{\omega}(\ul p_{\pi(1)},\dots,\ul p_{\pi(n+1)})\;,
\end{equation}
where we used the invariance under maximal cyclic permutations.
%
%
Thus,
\begin{equation}
\sum_{\pi\in S_{n}} \mathfrak B_{n+1}^{\omega}(\ul p_{\pi(1)},\dots,\ul p_{\pi(n)},\ul p_{n+1})= \frac{1}{n+1}\sum_{\pi\in S_{n+1}} \mathfrak B_{n+1}^{\omega}(\ul p_{\pi(1)},\dots,\ul p_{\pi(n+1)})\;.
\end{equation}
We further rewrite:
\begin{equation}\label{eq:615}
\sum_{\pi\in S_{n}} \mathfrak B_{n+1}^{\omega}(\ul p_{\pi(1)},\dots,\ul p_{\pi(n)},\ul p_{n+1})= \frac{1}{n+1} \sum_{i=1}^{n+1} \sum_{\substack{\pi\in S_{n+1}\\ \pi(1)=i}} \mathfrak B_{n+1}^{\omega}(\ul p_{i},\dots,\ul p_{\pi(n+1)})\;;
\end{equation}
applying the Ward identity (\ref{eq:ward}) to $g_{\omega}(\underline{k} + \underline{p}_{i}) g_{\omega}(\underline{k})$ we get:
\begin{equation}\label{eq:perm}
\begin{split}
\mathfrak B_{n+1}^{\omega}(\ul p_{i},\dots,\ul p_{\pi(n+1)}) &= \frac{\mathfrak B_{n}^{\omega}(\ul p_{\pi(2)}+\ul p_{i},\dots,\ul p_{\pi(n+1)})-\mathfrak B_{n}^{\omega}(\ul p_{\pi(2)},\dots,\ul p_{\pi(n+1)}+\ul p_{i})}{D_{\omega}(\ul p_{i})}\\
&\quad-\int \frac{d\ul k}{(2\pi)^2}\, F^{\omega}_{\delta}(\ul k;\ul p_i)\prod_{m=3}^{n+1} g_{\omega}\Big(\ul k+\sum_{j<m} \ul p_{\pi(j)}\Big)\;.
\end{split}
\end{equation}
Let us plug the first term in the right-hand side of (\ref{eq:perm}) in the rightmost sum in (\ref{eq:615}). We get:
\begin{equation}
\begin{split}
&\sum_{\substack{\pi\in S_{n+1}\\ \pi(1)=i}} \mathfrak B_{n+1}^{\omega}(\ul p_{i},\dots,\ul p_{\pi(n+1)}) \\
&= \frac{1}{D_{\omega}(\ul p_{i})} \sum_{\substack{\pi\in S_{n+1}\\ \pi(1)=i}} \Big(\mathfrak B_{n}^{\omega}(\ul p_{\pi(2)}+\ul p_{i},\dots,\ul p_{\pi(n+1)})-\mathfrak B_{n}^{\omega}(\ul p_{\pi(2)},\dots,\ul p_{\pi(n+1)}+\ul p_{i})\Big) \\
& = \frac{1}{D_{\omega}(\ul p_{i})} \sum_{\substack{\pi\in S_{n+1}\\ \pi(1)=i}} \Big(\mathfrak B_{n}^{\omega}(\ul p_{\pi(2)}+\ul p_{i}, p_{\pi(3)},\dots,\ul p_{\pi(n+1)})-\mathfrak B_{n}^{\omega}(\ul p_{\pi(n+1)}+\ul p_{i}, \ul p_{\pi(2)}, \ldots, \underline{p}_{\pi(n)})\Big) \\
&= 0\;,
\end{split}
\end{equation}
where the last identity follows from the fact that given a permutation $\pi$ such that $\pi(1) = i$, we can define the new permutation $\tilde \pi$ such that $\tilde \pi(1)=i$ and $\tilde \pi(2) = \pi(n+1)$, $\tilde \pi(3) = \pi(2)$, \ldots, $\tilde\pi(n+1) = \pi(n)$, a relation that defines a one-to-one map between permutations. Hence, (\ref{eq:sumloops}) follows.
\end{proof}
The cancellation highlighted in the previous proposition implies an improved estimate for the $n$-th order contribution to the full response.
%
%
%
%
\begin{proposition}[Loop estimates: relativistic contributions]\label{prp:loopest} Let $p_{1}, \ldots, p_{n}$ such that $\hat\mu_{\alpha,\theta}(-p_{i}) \neq 0$. Suppose that $2\leq n< (1/4) \theta^{\alpha-1}$. Then, for $\beta, L$ large enough:
\begin{equation}\label{eq:bubblesest}
\frac{\theta^{n-1}}{n!}|S^{\beta,L}_{n;\nu}(\ul p_{1}, \ldots, \underline{p}_{n})|\leq C^{n} \theta^{n-1}.
\end{equation}
Suppose instead that $n\geq (1/4)\theta^{\alpha-1}$, then:
\begin{equation}\label{eq:estnlarge}
\frac{\theta^{n-1}}{n!}|S^{\beta,L}_{n;\nu}(\ul p_{1}, \ldots, \underline{p}_{n})|\leq C^{n}\frac{ a^{n-1}}{n!} \left|\log\eta\right|.
\end{equation}
\end{proposition}
\begin{proof} Let us prove (\ref{eq:bubblesest}). Let $2\leq n< (1/4) \theta^{\alpha-1}$. It will be enough to prove the bound (\ref{eq:bubblesest}) for the $\beta, L \to \infty$ limit. We start from the following rewriting, consequence of (\ref{prp:loopcanc}) and of (\ref{eq:bubbles}):
\begin{equation}\label{eq:621}
\begin{split}
S_{n;\nu}(\ul p_{1}, \ldots, \underline{p}_{n}) &= \sum_{\omega=1}^{N} v^{\nu}_{\omega}\sum_{\pi\in S_{n}} \mathfrak B_{n+1}^{\omega}(\ul p_{\pi(1)},\dots,\ul p_{\pi(n)},\ul p_{n+1}) \\
&= -\frac{1}{n+1}\sum_{\omega=1}^{N} v^{\nu}_{\omega}\sum_{\pi\in S_{n+1}}\int \frac{d\ul k}{(2\pi)^2}\, F_{\delta}^{\omega}(\ul k;\ul p_{\pi(1)})\prod_{i=3}^{n+1} g_{\omega}\Big(\ul k+\sum_{j<i} \ul p_{\pi(l)}\Big)\;.
\end{split}
\end{equation}
The proof of (\ref{eq:bubblesest}) is based on the support properties of the function $F_{\delta}^{\omega}$. For all $\|\ul p\|_\omega< (1/4) \delta$, we write:
\begin{equation}
\chi_{\delta}^\omega(\ul k+\ul p)=\chi_{\delta}^\omega(\ul k)+\frac{1}{\delta}\int_0^1 \chi'(\|\ul k - \ul k_{F}^{\omega} + t\ul p\|_\omega/\delta)\frac{\langle\ul k - \ul k_{F}^{\omega} + t\ul p, \ul p\rangle_\omega}{\|\ul k - \ul k_{F}^{\omega} +t\ul p\|_\omega}\,dt
\end{equation}
where $\langle \ul a,\ul b\rangle_\omega:=a_0b_0+v_\omega^2 ab$ is the scalar product associated to the norm $\left\|\cdot\right\|_\omega$. Thus, from the definition (\ref{eq:defF}), we have:
\begin{equation}
\begin{split}
&|F^\omega_{\delta}(\ul k;\ul p)| \\&\leq \frac{\chi^{\omega}_{\delta}(\ul k)[1-\chi^{\omega}_{\delta}(k)]}{\|\ul k - \ul k_{F}^{\omega}\|_{\omega}\|\ul k - \ul k_{F}^{\omega}+\ul p\|_{\omega}}\\&\quad + \int_0^1 \frac{|\chi'(\|\ul k - \ul k_{F}^{\omega}+t\ul p\|_\omega/\delta)|}{\delta\|\ul k - \ul k_{F}^{\omega} +\ul p\|_\omega}\left[\frac{\|\ul p\|_{\omega}}{\|\ul k - \ul k_{F}^{\omega}\|_{\omega}}\chi^{\omega}_{\delta}(\ul k)+1\right]\,dt.
\end{split}
\end{equation}
Sinces $\mathrm{supp}(\chi')\subseteq[1,2]$, we can estimate, for some $C_{\lambda} > 0$:
\begin{equation}\label{eq:Fbd}
|F^\omega_{\delta}(\ul k;\ul p)|\leq
\begin{cases}
 C_{\lambda}/\delta^2 & \|\ul k - \ul k_{F}^{\omega}\|_\omega\in[(3/4) \delta, (9/4)\delta]\\
 0 & \mathrm{otherwise.}
 \end{cases}
\end{equation}
Let us use this bound to estimate (\ref{eq:621}). Observe that, since $\underline{p}_{i}$ is such that $\hat\mu_{\alpha,\theta}(p_{i}) \neq 0$, we have $| p_{i} | \leq \theta^{1 - \alpha}$; hence we can assume that $\| \underline{p}_{i} \| \leq \theta^{1 - \alpha} + \theta / a \leq (1/4) \delta$ for $\theta$ small enough, and we can use the bound (\ref{eq:Fbd}). In particular, this estimate allows to show that, if the order $n$ is not too large, all the propagators are bounded uniformly in their arguments. In fact, from (\ref{eq:Fbd}), and for $n< (1/4) \theta^{\alpha-1}$:
\begin{equation}\label{eq:lowerk}
\begin{split}
\|{\ul k - \ul k_{F}^{\omega}+\textstyle\sum_{j\leq i}\ul p_{\pi(j)}}\|_\omega &\geq\|{\ul k} - \ul k_{F}^{\omega}\|_\omega - \|\textstyle\sum_{j\leq i} \ul p_{\pi(j)}\|_\omega\\&\geq \|\ul k - \ul k_{F}^{\omega}\|_{\omega} - n(\theta^{1-\alpha} + a\theta)\geq \delta/2\;.
\end{split}
\end{equation}
Thus, integration over $\ul k$ gives:
\begin{equation}
\int \frac{d\ul k}{(2\pi)^2}\, \big|F^\omega_{\delta}(\ul k;\ul p_{\pi(1)})\big|\prod_{i=3}^{n+1}\left|\frac{\chi^\omega_{\delta}(\ul k+\sum_{j<i}\ul p_{\pi(j)})}{D_{\omega}(\ul k+\sum_{j<i}\ul p_{\pi(j)})}\right| \leq (K/\delta)^{n-1};
\end{equation}
this bound, combined with (\ref{eq:621}), implies:
\begin{equation}\label{eq:Cnn}
|S_{n;\nu}(\ul p_{1}, \ldots, \underline{p}_{n})| \leq C^{n} n!\;,
\end{equation}
which proves Eq. (\ref{eq:bubblesest}). Since $S^{\beta,L}_{n;\nu}$ at nonzero external momenta converges to a limit for $\beta, L\to \infty$, an analogous estimate (\ref{eq:Cnn}) holds for $\beta, L$ large enough and $n\leq (1/4) \theta^{\alpha-1}$.

Next, let us prove (\ref{eq:estnlarge}). Suppose that $n\geq (1/4)\theta^{\alpha-1}$. For these large values of $n$, we cannot use the previous argument, because the condition (\ref{eq:lowerk}) might fail. Here we shall use that, thanks to the fact that the temporal part of the external momenta is $p_{i,0} = \eta$ for all $i=1,\ldots,n$, there is no accumulation of singularities in the loop integral. Starting from (\ref{eq:singular}), we estimate:
\begin{equation}\label{eq:Slarge}
\big|S^{\beta,L}_{n;\nu}(\underline{p}_{1}, \ldots, \underline{p}_{n})\big| \leq C\sum_{\omega = 1}^{N}  \int_{\beta,L} \frac{d\underline{k}}{(2\pi)^{2}} \sum_{\pi\in S_{n}} \prod_{i=1}^{n+1} \Big| g_{\omega}\Big(\ul k+\sum_{j< i} \ul p_{\pi(j)}\Big)\Big|
\end{equation}
and:
\begin{equation}\label{eq:5119}
\int_{\beta,L} \frac{d\underline{k}}{(2\pi)^{2}}\, \prod_{i=1}^{n+1} \Big| g_{\omega}\Big(\ul k+\sum_{j< i} \ul p_{j}\Big)\Big| \leq  \int_{\beta, L} \frac{d \underline{k}}{(2\pi)^{2}}\, \prod_{i=1}^{n+1} \frac{\chi_{\delta}^\omega(\ul k + \sum_{j< i}\ul p_j)}{|D_\omega(\ul k + \sum_{j< i}\ul p_j)|}\;.
\end{equation}
By the support properties of the cutoff function, the domain of summation over $\underline{k}$ is contained in $B^\omega_{2\delta} = \{\underline{k} \mid \|\underline{k} - \underline{k}_{F}^{\omega}\|_{\omega} \leq 2\delta\}$; in particular, $|k_{0}| \leq 2\delta$. We further split the Matsubara frequencies as:
\begin{equation}
\{ k_{0} \in \mathbb{M}_{\beta} \mid |k_{0}| \leq 2\delta \} = \mathcal{A} \cup \mathcal{A}^{c}\;,
\end{equation}
with, for $\bar n=\min\{n,\lfloor{2\delta/\eta}\rfloor\}$:
\begin{equation}
\begin{split}
\mathcal{A} &:= \bigcup_{j=1}^{\bar n} \mathcal{A}_{j} \\
\mathcal A_j &:= \left\{k_{0} \in \mathbb{M}_{\beta}  \mid k_{0}\in\eta\left [-j-\frac{1}{2},-j+\frac{1}{2}\right ]\right\}.
\end{split}
\end{equation}
Observe that, for $k_{0} \in \mathcal{A}_{j}$ and $i\neq j$:
\begin{equation}
|D_\omega(\ul k+\textstyle\sum_{l\leq i}\ul p_l)|\geq|k_0+\eta i|\geq \eta[|i-j|-1/2]\geq\eta|i-j|/2\;.
\end{equation}
Therefore, for $k_{0} \in \mathcal{A}_{j}$:
\begin{equation}
\begin{split}
\prod_{i=0}^{n}\frac{1}{|D_\omega(\ul k+\textstyle\sum_{l\leq i} \ul p_l)|} &\leq \frac{(2/\eta)^{n}}{\|\ul k - \ul k_{F}^{\omega} +\textstyle\sum_{l\leq j} \ul p_l\|_{\omega}}\displaystyle\prod_{\substack {i=0\\ i\neq j}}^{n} \frac{1}{|i-j|} \\
&=\frac{1}{n!}\binom{n}{j}\frac{(2/\eta)^{n}}{\|\ul k - \ul k_{F}^{\omega} +\textstyle\sum_{l\leq j} \ul p_l\|_{\omega}}\;.
\end{split}
\end{equation}
Hence:
\begin{equation}\label{eq:bdA}
\begin{split}
&\int_{\beta,L} \mathbbm{1}(k_{0} \in \mathcal{A}_{j}) \prod_{i=0}^{n} \frac{\chi_{\delta}^\omega(\ul k + \sum_{j\leq i}\ul p_j)}{|D_\omega(\ul k + \sum_{j\leq i}\ul p_j)|}\, \frac{d \underline{k}}{(2\pi)^{2}} \\
&\qquad \leq \frac{(2 / \eta)^{n}}{n!}\sum_{j=0}^{\bar n} \binom{n}{j} \int_{\beta,L} \mathbbm{1}(k_{0} \in \mathcal{A}_{j}) \frac{\chi^\omega_{\delta}(\ul k - \ul k_{F}^{\omega} +\textstyle\sum_{l\leq j} \ul p_l)}{\|\ul k - \ul k_{F}^{\omega} +\textstyle\sum_{l\leq j} \ul p_l\|_{\omega}}\,d\ul k \\
&\qquad \leq \frac{C^{n} \eta^{1-n}}{n!} \left| \log \eta\right| \sum_{j=1}^{\bar n} \binom{n}{j} \\
&\qquad \leq \frac{(2C)^{n} \eta^{1-n}}{n!} \left| \log \eta\right|\;,
\end{split}
\end{equation}
since $\bar n \leq n$. Consider now the contribution of the region $\mathcal A^c$. Here, we shall use that:
\begin{equation}
|D_{\omega}(\ul k+\textstyle\sum_{l\leq i}\ul q_l)|\geq |k_{0}+\eta i|\geq 
\begin{cases}
 (i+1)\eta/2 & \text{if $k_0\geq 1/2$}\\
(n-i+1)\eta/2 & \text{if $k_0\leq-n-1/2$,}
\end{cases}
\end{equation}
where the second case is needed only if $\bar n=n$, i.e. $n<\lfloor2\delta/\eta\rfloor$. In both cases, we can extract the inverse of a factorial, as before. We have:
\begin{equation}\label{eq:bdAc}
\begin{split}
&\int_{\beta,L} \mathbbm{1}(k_{0}\in \mathcal{A}^{c}) \prod_{i=0}^{n}\frac{\chi^\omega_{\delta}(\ul k+\sum_{j\leq i}\ul p_{j})}{|D_\omega(\ul k+\sum_{j\leq i}\ul p_{j})|}\,d\ul k \\
&\qquad \leq \frac{(C/\eta)^{n-1}}{(n-1)!} \int_{\beta,L} \mathbbm{1}(k_{0}\in \mathcal{A}^{c}) \frac{\chi_{\delta}^\omega(\ul k)\chi^\omega_{\delta}(\ul k+\ul p_{1})}{|D_\omega(\ul k)D_{\omega}(\ul k+\ul p_{1})|}\,d\ul k \\&\qquad \leq \frac{(K / \eta)^{n-1}}{(n-1)!}\left|\log\eta\right|.
\end{split}
\end{equation}
Putting together the above estimates we obtained:
\begin{equation}
\int_{\beta,L} \frac{d\underline{k}}{(2\pi)^{2}}\, \prod_{i=1}^{n+1} \Big| g_{\omega}\Big(\ul k+\sum_{j< i} \ul p_{j}\Big)\Big| \leq \frac{C^{n}}{n!} \frac{\left| \log \eta\right|}{\eta^{n-1}}\;;
\end{equation}
this implies, from (\ref{eq:Slarge}):
\begin{equation}
|S^{\beta,L}_{n;\nu}(\ul p_{1}, \ldots, \underline{p}_{n})| \leq C^{n} \frac{\left| \log \eta\right|}{\eta^{n-1}}\;,
\end{equation}
which proves (\ref{eq:estnlarge}). This concludes the proof of Proposition \ref{prp:loopest}.
\end{proof}
We shall now estimate the contribution of the terms $R^{\beta,L}_{n;\nu}(\underline{p}_{1}, \ldots, \underline{p}_{n})$ in (\ref{eq:fullresp}). To this end, let us recall the splitting:
\begin{equation}\label{eq:RRT}
R^{\beta,L}_{n;\nu}(\underline{p}_{1}, \ldots, \underline{p}_{n}) = \widetilde R^{\beta,L}_{n;\nu}(\underline{p}_{1}, \ldots, \underline{p}_{n}) + T^{\beta,L}_{n;\nu}(\underline{p}_{1}, \ldots, \underline{p}_{n})\;,
\end{equation}
with:
\begin{equation}\label{eq:rem}
\widetilde R^{\beta,L}_{n;\nu}(\underline{p}_{1}, \ldots, \underline{p}_{n}) = - \sum_{\pi\in S_{n}} \sum_{\substack{f\in\{r,s\}^{n+1}\\ f\not\equiv s}} \int_{\beta,L} \Tr\Big[\hat J_{\nu}(k, p_{n+1})\prod_{i=1}^{n+1} g_{f(i)}\Big(\ul k+\sum_{j< i} \ul p_{\pi(j)}\Big)\Big]\frac{d\ul k}{(2\pi)^2}\;,
\end{equation}
and:
\begin{equation}\label{eq:remT}
T^{\beta,L}_{n;\nu}(\underline{p}_{1}, \ldots, \underline{p}_{n}) = - \sum_{\omega=1}^{N}\sum_{\pi\in S_{n}}\int_{\beta,L} f_{n,\omega}^{\nu}(k;p_{\pi(1)},\dots,p_{\pi(n)}) \prod_{i=1}^{n+1} g_{\omega}\Big(\ul k+\sum_{j< i} \ul p_{\pi(j)}\Big) \frac{d\ul k}{(2\pi)^2}\;,
\end{equation}
where:
\begin{equation}\label{eq:fest}
|f_{n,\omega}^{\nu}(k;p_{1},\dots,p_{n})|\leq C_{\delta}\left(|k-k_{F}^{\omega}|_{\mathbb{T}}\delta_{\nu,1}+ n\sum_{i=1}^{n} |p_{i}|_{\mathbb{T}}\right)\;.
\end{equation}
%
%
%
%
\begin{proposition}[Loop estimates: remainder terms]\label{lemma:j1}
There exists $C>0$ such that for all $n\geq 2$:
\begin{equation}\label{eq:estrem}
\frac{\theta^{n-1}}{n!}|R^{\beta,L}_{n;\nu}(\underline{p}_{1}, \ldots, \underline{p}_{n})|\leq C^{n} a^{n-1} \left| \log \eta\right| \sum_{j=1}^{n-1} \frac{\eta^{j} + \delta_{j,1}  \sum_{i=1}^{n} |p_{i}|_{\mathbb{T}}}{(n-j)!}
\end{equation}
for all $\underline{p}_{i} = (\eta_{\beta}, p_{i})$, $i=1,\ldots, n$, such that $\hat\mu_{\alpha,\theta}(-p_{i}) \neq 0$.
\end{proposition}
\begin{proof} With reference to (\ref{eq:RRT}), let us start by proving an estimate for $\widetilde R^{\beta,L}_{n;\nu}(\underline{p}_{1}, \ldots, \underline{p}_{n})$. To this end, it is convenient to introduce:
\begin{equation}\label{eq:remj}
\widetilde R^{\beta,L}_{\nu;n,j}(\underline{p}_{1}, \ldots, \underline{p}_{n}):= - \sum_{\pi\in S_{n}} \sum_{\substack{f\in\{r,s\}^{n+1}\\ |f^{-1}(r)|=j}} \int_{\beta,L}\Tr\Big[\hat J_{\nu}(k, p_{n+1})\prod_{i=1}^{n+1} g_{f(i)}\Big(\ul k+\sum_{l< i} \ul p_{\pi(l)}\Big)\Big]\frac{d\ul k}{(2\pi)^2};
\end{equation}
that is, $\widetilde R^{\beta,L}_{\nu;n,j}$ collects the contribution coming from the loops with exactly $j$ bounded propagators $g_{\mathrm{r}}$. Then, we can estimate:
\begin{equation}\label{eq:remJ}
|\widetilde R^{\beta,L}_{\nu;n,j}(\underline{p}_{1}, \ldots, \underline{p}_{n})| \leq C^{j} \sum_{\pi\in S_{n}}\sum_{\substack{I\subseteq\{1,\dots,n+1\}\\ |I|=j}}\sum_{\Omega\in\{1,\dots,N\}^{I^{c}}} \int_{\beta,L} \prod_{i\in I^{c}}\Big| g_{\Omega(i)}\Big(\ul k+\sum_{l< i}\ul p_{\pi(l)}\Big)\Big|\, d\ul k,
\end{equation}
where the last sum runs over chiralities $\Omega(i)$ with $i\in I^{\text{c}}$, and the constant $C^{j}$ takes into account the estimate for the bounded propagators. For the moment, we suppose that $j\leq n-1$, so that in the integral we have at least two relativistic propagators.  Let us label the elements of $I^{c}$ as $I^{c} = \{m_{0}, m_{1}, \ldots, m_{n-j}\}$ in increasing order. Performing the change of variables $\underline{k} + \sum_{i\leq m_{0}} \underline{p}_{i} \to \underline{k}$, we can estimate:
\begin{equation}\label{eq:giest}
\int_{\beta,L} \prod_{i\in I^{c}}\Big| g_{\Omega(i)}\Big(\ul k+\sum_{l< i}\ul p_{\pi(l)}\Big)\Big|\, d\ul k \leq \int_{\beta,L} \prod_{i=0}^{n-j}\frac{\chi^{\Omega(m_i)}_{\delta}(\ul k + \textstyle \sum_{l >m_{0}}^{m_{i}} \ul p_l)}{\big|D_{\Omega(m_i)}(\ul k+\textstyle\sum_{l > m_{0}}^{m_{i}} \ul p_l)\big|}\,d\ul k\;.
\end{equation} 
Next, we proceed as after (\ref{eq:5119}). We write:
\begin{equation}
\begin{split}
&\int_{\beta,L} \prod_{i=0}^{n-j}\frac{\chi^{\Omega(m_i)}_{\delta}(\ul k+\textstyle \sum_{l >m_{0}}^{m_{i}} \ul p_l)}{\big|D_{\Omega(m_i)}(\ul k+\textstyle \sum_{l >m_{0}}^{m_{i}} \ul p_l)\big|}\,d\ul k \\
&\qquad = \int_{\beta,L} \mathbbm{1}(k_{0} \in \mathcal{A}^{c}) \prod_{i=0}^{n-j}\frac{\chi^{\Omega(m_i)}_{\delta}(\ul k+\textstyle \sum_{l >m_{0}}^{m_{i}} \ul p_l)}{\big|D_{\Omega(m_i)}(\ul k+\textstyle \sum_{l >m_{0}}^{m_{i}} \ul p_l)\big|}\,d\ul k\\
&\qquad\quad + \sum_{ \ell = 1}^{\bar n} \int_{\beta,L} \mathbbm{1}(k_{0} \in \mathcal{A}_{\ell}) \prod_{i=0}^{n-j}\frac{\chi^{\Omega(m_i)}_{\delta}(\ul k+\textstyle \sum_{l >m_{0}}^{m_{i}} \ul p_l)}{\big|D_{\Omega(m_i)}(\ul k+\textstyle \sum_{l >m_{0}}^{m_{i}} \ul p_l)\big|}\,d\ul k\;.
\end{split}
\end{equation}
Consider the argument of the sum with $\ell =m_{s} - m_{0}$ for $m_{s}\in I^{c}$. Proceeding as in (\ref{eq:bdA}):
\begin{equation}
\begin{split}
&\int_{\beta,L} \mathbbm{1}(k_{0} \in \mathcal{A}_{\ell}) \prod_{i=0}^{n-j}\frac{\chi^{\Omega(m_i)}_{\delta}(\ul k + \textstyle\sum_{l > m_0}^{m_i} \ul p_l)}{\big|D_{\Omega(m_i)}(\ul k+\textstyle\sum_{l > m_0}^{m_i} \ul p_l)\big|}\,d\ul k \\
&\qquad \leq (C / \eta)^{n-j} (\eta \left| \log \eta\right|)\prod_{\substack{i=0 \\ |m_{i} - m_{s}|>0}}^{n-j} \frac{1}{|m_{i} - m_{s}|}\;.
\end{split}
\end{equation}
Using that $|m_{i} - m_{s}| \geq |i-s|$, we have:
\begin{equation}\label{eq:aa1}
\begin{split}
&\sum_{s = 0}^{n-j} \int_{\beta,L} \mathbbm{1}(k_{0} \in \mathcal{A}_{m_{s} - m_{0}}) \prod_{i=0}^{n-j}\frac{\chi^{\Omega(m_i)}_{\delta}(\ul k + \textstyle\sum_{l > m_0}^{m_i} \ul p_l)}{\big|D_{\Omega(m_i)}(\ul k+\textstyle\sum_{l > m_0}^{m_i} \ul p_l)\big|}\,d\ul k \\
&\qquad  \leq (K / \eta)^{n-j} (\eta \left| \log \eta\right|) \frac{1}{(n-j)!} \sum_{s = 0}^{n-j} \binom{n-j}{s} \\
&\qquad \leq (C / \eta)^{n-j} (\eta \left| \log \eta\right|) \frac{1}{(n-j)!}\;.
\end{split}
\end{equation}
Instead, the contribution of all $\ell$ such that $\ell + m_{0} \notin I^{c}$ is:
\begin{equation}\label{eq:aa2}
\sum_{\ell: \ell + m_{0} \notin I^{c}} \int_{\beta,L} \mathbbm{1}(k_{0} \in \mathcal{A}_{\ell}) \prod_{i=0}^{n-j}\frac{\chi^{\Omega(m_i)}_{\delta}(\ul k+\textstyle\sum_{l > m_0}^{m_i} \ul p_l)}{\big|D_{\Omega(m_i)}(\ul k+\textstyle\sum_{l > m_0}^{m_i} \ul p_l)\big|}\,d\ul k \leq (C / \eta)^{n-j}  \frac{\eta}{(n-j)!}\;.
\end{equation} 
Consider now the contribution due to the region $\mathcal{A}^{c}$. Also here, none of the propagators in the product is singular. Therefore, proceeding as in (\ref{eq:bdAc}):
\begin{equation}\label{eq:aa3}
\int_{\beta,L} \mathbbm{1}(k_{0} \in \mathcal{A}^{c}) \prod_{i=0}^{n-j}\frac{\chi^{\Omega(m_i)}_{\delta}(\ul k+\textstyle\sum_{l > m_0}^{m_i} \ul p_l)}{\big|D_{\Omega(m_i)}(\ul k+\textstyle\sum_{l > m_0}^{m_i} \ul p_l)\big|}\,d\ul k \leq (C / \eta)^{n-j+1}  \frac{1}{(n-j-1)!} \left|\log \eta\right |\;.
\end{equation}
All in all, from (\ref{eq:giest}), (\ref{eq:aa1}), (\ref{eq:aa2}), (\ref{eq:aa3}) we obtained (with a different $C$):
\begin{equation}
\int_{\beta,L} \prod_{i=0}^{n-j}\frac{\chi^{\Omega(m_i)}_{\delta}(\ul k+\textstyle\sum_{l > m_0}^{m_i} \ul p_l)}{\big|D_{\Omega(m_i)}(\ul k+\textstyle\sum_{l > m_0}^{m_i} \ul p_l)\big|}\,d\ul k \leq (C / \eta)^{n-j} (\eta \left|\log \eta\right|) \frac{1}{(n-j)!}\;.
\end{equation}
Plugging this bound in (\ref{eq:remJ}), we get:
\begin{equation}\label{eq:Rdim}
\begin{split}
|\widetilde R^{\beta,L}_{\nu;n,j}(\underline{p}_{1}, \ldots, \underline{p}_{n})| \leq C^{n} n! \eta^{j-n} (\eta \left|\log \eta\right|) \frac{1}{(n-j)!}\;.
\end{split}
\end{equation}
Suppose that that $j=n$ or $j=n+1$. If $j=n+1$, we use two propagators $g_{\mathrm{r}}$ to sum over $\underline{k}$. The final estimate is:
\begin{equation}
|\widetilde R^{\beta,L}_{\nu;n,n+1}(\underline{p}_{1}, \ldots, \underline{p}_{n})| \leq C^{n} n!\;.
\end{equation}
If $j=n$, we sum over $\underline{k}$ with the relativistic propagator and with one propagator $g_{\mathrm{r}}$. Also in this case, the final estimate is:
\begin{equation}
|\widetilde R^{\beta,L}_{\nu;n,n}(\underline{p}_{1}, \ldots, \underline{p}_{n})| \leq C^{n} n!\;.
\end{equation}
Therefore, we obtained:
\begin{equation}\label{eq:tildeRest}
\frac{\theta^{n-1}}{n!} \sum_{j} |\widetilde R^{\beta,L}_{\nu;n,j}(\underline{p}_{1}, \ldots, \underline{p}_{n})| \leq C^{n} \theta^{n-1} + C^{n} a^{n-1} \eta \left| \log \eta\right| \sum_{j=1}^{n-1} \frac{\eta^{j-1}}{(n-j)!}\;.
\end{equation}
Finally, consider the error term (\ref{eq:remT}). The estimate for this term is performed as in the $j=0$ version of the argument above, keeping into account a factor $\left(|k-k_{F}^{\omega}|_{\mathbb{T}} + n\sum_{i=1}^{n} |p_{i}|_{\mathbb{T}}\right)$ arising from the estimate (\ref{eq:fest}). Since the combination $|k-k_{F}^{\omega}| g_{\omega}(\underline{k})$ is bounded, we have:
\begin{equation}\label{eq:Tdim}
\frac{\theta^{n-1}}{n!} |T^{\beta,L}_{n;\nu}(\underline{p}_{1}, \ldots, \underline{p}_{n})| \leq \frac{C^{n} a^{n-1}}{n!} \left| \log \eta\right| (\eta + \sum_{i} |p_{i}|_{\mathbb{T}})\;,
\end{equation}
we omit the details. Putting together (\ref{eq:tildeRest}) and (\ref{eq:Tdim}), Eq. (\ref{eq:estrem}) follows.
\end{proof}
We are now ready to prove Theorem \ref{thm:main1d}.
\begin{proof}[Proof of Theorem \ref{thm:main1d}] Recall the expression (\ref{eq:fullresp}) for the full response of the system. We are interested in proving an estimate for the sum of all contributions with $n\geq 2$, which is subleading with respect to the linear response as $\eta \to 0^{+}$. Let $n(\theta) = (1/4) \theta^{\alpha-1}$. We have, from Proposition \ref{prp:loopest} and Proposition \ref{lemma:j1}, for $\beta, L$ large enough:
\begin{equation}\label{eq:n>1}
\begin{split}
&\sum_{n=2}^\infty \frac{\theta^{n-1}}{n!} \frac{1}{L^{n}} \sum_{\{p_{i}\} \in B_{L}^{n}} \Big[\prod_{j=1}^{n} |\hat\mu_{\alpha,\theta}(-p_{j})|\Big] \Big(|S^{\beta,L}_{n;\nu}(\underline{p}_{1},\dots, \underline{p}_{n})|+|R^{\beta,L}_{n;\nu}(\underline{p}_{1},\dots, \underline{p}_{n})|\Big) \\
&\quad \leq \sum_{n = 2}^{n(\theta)} C^{n} \| \hat \mu_{\alpha,\theta} \|_{1}^{n} \theta^{n-1}  + \sum_{n>n(\theta)} C^{n} \| \hat \mu_{\alpha,\theta} \|_{1}^{n} \frac{ a^{n-1}}{n!} \left|\log\eta\right|\\
&\qquad  + \sum_{n\geq 2} C^{n} a^{n-1}(\eta + \theta) \left| \log \eta\right| \sum_{j=1}^{n-1} \frac{\eta^{j-1}}{(n-j)!}\;.
\end{split}
\end{equation}
The first term in the right-hand side comes from (\ref{eq:bubblesest}); the second term from (\ref{eq:estnlarge}); the last term from (\ref{eq:estrem}). In particular, the factor $\theta$ in the third term comes from $\| \hat \mu_{\alpha,\theta}(p_{i}) |p_{i}|_{\mathbb{T}} \|_{1} \leq C\theta$. We estimate the last sum as:
\begin{equation}\label{eq:sums}
\begin{split}
& (\eta + \theta) \left| \log \eta\right|\sum_{n\geq 2} C^{n} a^{n-1}  \sum_{j=1}^{n-1} \frac{\eta^{j-1}}{(n-j)!} \\
&\quad \leq \widetilde{C} (\eta + \theta) \left| \log \eta\right|\sum_{n\geq 0} (Ka)^{n}  \sum_{j=0}^{n} \frac{\eta^{j}}{(n-j)!} \\
&\quad \leq \widetilde{C}(\eta + \theta) \left| \log \eta\right| \sum_{j\geq 0} (K a\eta)^{j}  \sum_{n\geq j} \frac{(K a)^{n-j}}{(n-j)!} \\
&\quad= \widetilde{C}(\eta + \theta) \left| \log \eta\right| \frac{e^{Ka}}{1 - Ka\eta}\;.
\end{split}
\end{equation}
We choose $a \leq a(\eta)$, so that right-hand side of (\ref{eq:sums}) vanishes as $\eta \to 0$. This holds true if $a(\eta) = w\left| \log \eta\right|$ with $w>0$ small enough. Therefore, putting together (\ref{eq:n>1}), (\ref{eq:sums}), (\ref{eq:fullresp}) we get, for $\gamma>0$, and $\beta, L$ large enough:
\begin{equation}
\Big|\chi^{\beta,L}_{\nu}(x;\eta,\theta) - \chi^{\text{lin}}_{\nu}(x;\eta,\theta)\Big| \leq C\eta^{\gamma}
\end{equation}
with $\chi^{\text{lin}}_{\nu}(x;\eta,\theta)$ given by (\ref{eq:linear}). This concludes the proof of Theorem \ref{thm:main1d}.
\end{proof}
\section{Large scale response of edge modes of $2d$ systems}\label{sec:2d}
In this section we will adapt the previous analysis to study the response of edge currents and edge densities at the boundary of $2d$ topological insulators. 
\subsection{Lattice fermions on the cylinder}

\paragraph{Hamiltonian and Gibbs state.} Give $L\in 2\mathbb N +1$, we consider fermions on the two-dimensional lattice
\begin{equation}
\Gamma_{L} = \left\{ x\in \mathbb{Z}^{2}\, \Big| -\left\lfloor\frac{L}{2}\right\rfloor\leq x_{1} \leq  \left\lfloor\frac{L}{2}\right\rfloor,\ 0\leq x_{2}\leq L-1\right\}\;,
\end{equation}
endowed with periodic boundary conditions in the $x_{1}$ coordinate and Dirichlet boundary condition in the $x_{2}$ coordinate:
\begin{equation}
f(x_{1}, x_{2}) = f(x_{1} + L, x_{2})\qquad \text{and}\qquad f(x_{1}, 0) = f(x_{1}, L-1) = 0\qquad \text{for all $x\in \Gamma_{L}$}\;.
\end{equation}
We shall introduce the following distance on $\Gamma_{L}$:
\begin{equation}
| x - y |_{L}^{2} = \min_{n\in \mathbb{Z}} | x_{1} - y_{1} + n L |^{2} + |x_{2} - y_{2}|^{2}\;.
\end{equation}
As for the one-dimensional case, we will allow for the presence of internal degrees of freedom; the resulting decorated lattice is $\Lambda_{L} = \Gamma_{L} \times S_{M}$, and we denote by ${\bf x} = (x, \sigma)$ with $x\in \Gamma_{L}$ and $\sigma \in S_{M}$ the points on $\Lambda_{L}$.

We shall suppose that the Hamiltonian $H$ on $\ell^{2}(\Lambda_{L})$ is the periodization of a Hamiltonian $H^{\infty}$ on $\ell^{2}((\mathbb{Z} \times [0, L-1] \cap \mathbb{Z}) \times S_{M})$:
\begin{equation}
H((x,\rho); (y,\rho')) = \sum_{a \in \mathbb{N}} H^{\infty}( (x + aL, \rho); (y,\rho') )\;.
\end{equation}
We assume that the Hamiltonian $H^{\infty}$ is finite-ranged and translation-invariant.
%
%
 Without loss of generality we can assume that the range is $\sqrt{2}$, up to increasing the number of internal degrees of freedom. As for the one-dimensional case, translation-invariance implies that the Hamiltonian $H$ can be fibered in momentum space. Let $k\in B_{L}$, with $B_{L}$ as in (\ref{eq:BLdef}). We define the Bloch Hamiltonian as:
\begin{equation}
\hat H_{\rho\rho'}(k; x_{2}, y_{2}) = \sum_{x} e^{-ikx} H_{\rho\rho'}((x, x_{2}); (0, y_{2}))\;.
\end{equation}
The operator $\hat H(k)$ is self-adjoint on $\ell^{2}\left(([0, L-1] \cap \mathbb{Z}\right) \times S_{M})$. By (\ref{eq:periodiz}), it is given by the restriction to $B_{L}$ of the Bloch Hamiltonian $\hat H^{\infty}(k)$ associated with $H^{\infty}$, which is defined for $k\in \mathbb{T}$. We shall make the following assumptions on the spectrum of the Bloch Hamiltonian and on the chemical potential $\mu$.
\begin{assumption}[Low-energy spectrum]\label{ass:B} There exists $\Delta > 0$ such that the following is true.
\begin{itemize}
\item[(i)] There exists $N \in \mathbb{N}$, disjoint sets $I_{\omega} \subset \mathbb T$ labelled by $\omega = 1, \ldots, N$, and strictly monotone, smooth functions $e_{\omega}: I_{\omega} \to \mathbb{R}$, such that
\begin{equation}
\sigma(H^{\infty}) \cap (\mu - \Delta,\, \mu + \Delta) = \bigcup_{\omega = 1}^{N} \mathrm{Ran}(e_{\omega})\;.
\label{eq:spectrum2}
\end{equation}
\item[(ii)] We introduce the $\omega$-Fermi point $k_{F}^{\omega} \in I_{\omega}$ and the $\omega$-Fermi velocity $v_{\omega}$ as
\begin{equation}
e_{\omega}(k_{F}^{\omega}) = \mu\;,\qquad v_{\omega} = \partial_{k} e_{\omega}(k_{F}^{\omega})\;. 
\end{equation}
Notice that $v_{\omega} \neq 0$, by the strict monotonicity of $e_{\omega}$.
\item[(iii)] For any $\omega=1,\dots, N$, and $k\in I_{\omega}$, $e_{\omega}(k)$ is a non-degenerate eigenvalue of $\hat H^{\infty}(k)$. Let $\xi^{\omega}(k)$ be the corresponding eigenfunction,
\begin{equation}
\hat H^{\infty}(k)\xi_{\omega}(k) = e_{\omega}(k)\xi_{\omega}(k)\;.
\end{equation}
We assume that $\xi_{\omega}(k)$ is exponentially localised near one of the boundaries of the cylinder:
\begin{equation}\label{eq:edgeest}
| \partial^{n}_{k}\xi_{\omega;\rho}(k;x_{2}) | \leq C_{n} e^{-c x_{2}}\qquad or \qquad | \partial^{n}_{k}\xi_{\omega;\rho}(k;x_{2})| \leq C_{n} e^{-c(L-x_{2})}\;,
\end{equation}
for all $n\in \N$.
\end{itemize}
\end{assumption}
\begin{remark} For short, we shall write $| \partial^{n}_{k}\xi_{\omega;\rho}(k;x_{2}) | \leq C_{n} e^{-c |x_{2}|_{\omega}}$, with the notation:
\begin{equation}
|x_{2}|_{\omega} = x_{2}\qquad \text{or}\qquad |x_{2}|_{\omega} = L - x_{2}\;,\qquad \text{for all $x_{2}$.}
\end{equation}
\end{remark}
\begin{figure}
    \centering
    \includegraphics[scale=0.5]{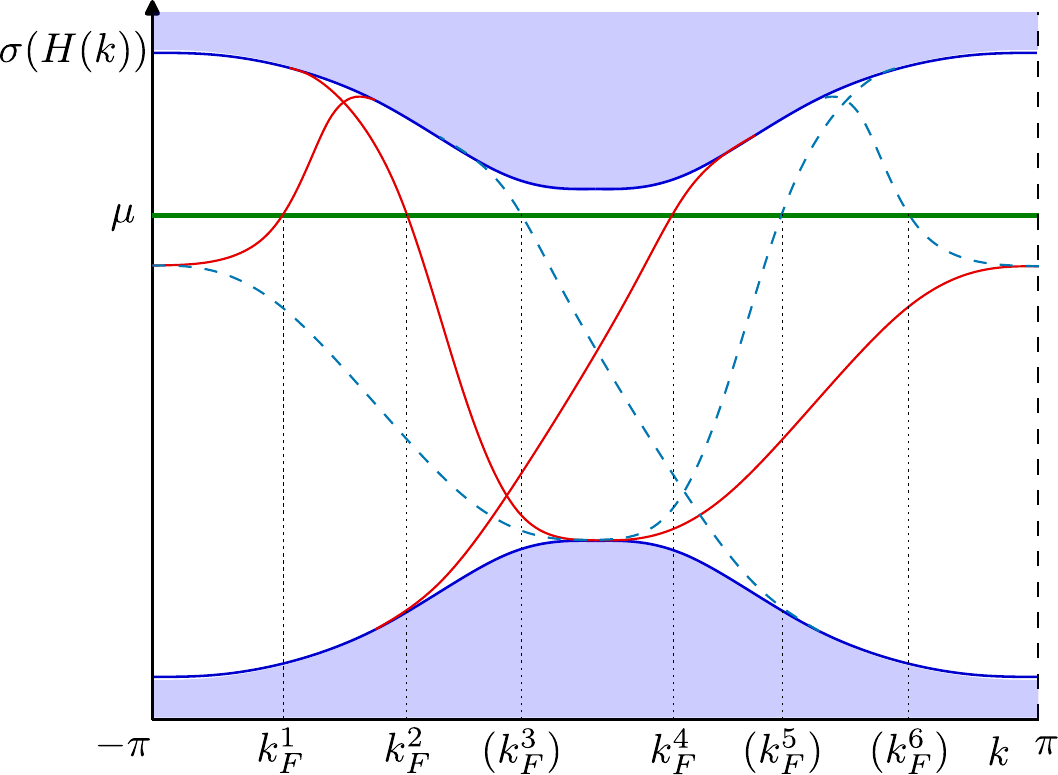}
    \caption{Typical form for the spectrum of $H$. The purple region corresponds to the ``bulk spectrum''. The red curves are the edge modes localized at $x_{2} = 0$, while the dotted curves are the edge modes localized at $x_{2} = L-1$.}
    \label{fig:edgespec}
\end{figure}
The second-quantized grand-canonical Hamiltonian is
\begin{equation}\label{eq:H2d}
\mathcal{H} = \sum_{{\bf x}, {\bf y} \in \Lambda_{L}} a^{*}_{{\bf x}} H({\bf x}; {\bf y}) a_{{\bf y}}
\end{equation}
with $a_{{\bf x}}$ and $a^{*}_{{\bf x}}$ the usual fermionic creation/annihilation operators, compatible with the boundary conditions on $\Gamma_{L}$. Similarly to (\ref{eq:foua}), (\ref{eq:foua2}) we define, for all $k\in B_{L}$:
\begin{equation}
\hat a_{(k, x_{2},\rho)} = \sum_{x_{1}\in \Gamma_{L}} a_{(x, \rho)} e^{-ikx_{1}}\;,\qquad \hat a^{*}_{(k, x_{2}, \rho)} = \sum_{x_{1}\in \Gamma_{L}} a^{*}_{(x,\rho)} e^{ikx_{1}}\;,
\end{equation}
which can be inverted as:
\begin{equation}
a_{(x,\rho)} = \frac{1}{L} \sum_{k \in B_{L}} e^{ikx_{1}} \hat a_{(k,x_{2}, \rho)},\qquad a^{*}_{(x,\rho)} = \frac{1}{L} \sum_{k \in B_{L}} e^{-ikx_{1}} \hat a^{*}_{(k, x_{2}, \rho)}\;.
\end{equation}
In terms of the momentum-space operators, the Hamiltonian reads:
\begin{equation}
\mathcal{H} = \frac{1}{L}\sum_{k \in B_{L}} \sum_{x_{2},y_{2} = 1}^{L}\sum_{\rho,\rho'\in S_{M}} \hat a^{*}_{(k, x_{2}, \rho)} \hat H_{\rho\rho'}(k; x_{2},y_{2}) \hat a_{(k, y_{2},\rho')}\;.
\end{equation}
The grand-canonical Gibbs state of the system at inverse temperature $\beta>0$ is:
\begin{equation}
\langle \mathcal{O} \rangle_{\beta,L} = \Tr \mathcal{O} \rho_{\beta, \mu, L}\;,\qquad \rho_{\beta, \mu, L} = \frac{e^{-\beta (\mathcal{H} - \mu \mathcal{N})}}{\Tr e^{-\beta (\mathcal{H} - \mu \mathcal{N})}}\;,
\end{equation}
with $\mu \in \mathbb{R}$ the chemical potential and $\mathcal{N}$ the number operator. We will suppose that $\mu$ is chosen so that Assumption \ref{ass:B} holds.

\paragraph{Perturbing the system.} Similarly to the one-dimensional case, we expose our system to a time-dependent and slowly varying perturbation, by introducing the Hamiltonian:
\begin{equation}\label{eq:H2dtime}
\mathcal{H}(\eta t) = \mathcal{H} + e^{\eta t} \mathcal{P}\qquad\quad \eta>0,\, t\leq0\;,
\end{equation}
where the perturbation $\mathcal P$ is given by
\begin{equation}\label{eq:Pdef1}
\mathcal{P} =  \theta \sum_{{\bf x} \in \Lambda_{L}} \mu(\theta x) a^{*}_{{\bf x}} a_{{\bf x}}\;;
\end{equation}
given $\mu_{\infty}\in\mathcal C^{\infty}_{c}(\R\times\R^{+})$ supported near the boundary $x_{2}=0$, the function $\mu(\theta x)$ is its periodization:
\begin{equation}\label{eq:period1}
\mu(\theta x) = \sum_{n\in \mathbb{Z}} \mu_{\infty}(\theta (x + n Le_{1}))\;.
\end{equation}
We can further write:
\begin{equation}\label{eq:muFou1}
\begin{split}
\mu(\theta x) &= \frac{1}{L} \sum_{p \in B_{L}} e^{ipx_{1}} \hat \mu_{\theta}(p,x_{2})\;,\\
\hat \mu_{\theta}(p,x_{2}) &= \sum_{m\in \mathbb{Z}} \frac{1}{\theta} \hat \mu_{\infty}((p + 2\pi m)/ \theta,\theta x_{2})\qquad \text{for all $p\in B_{L}$.}
\end{split}
\end{equation}
Let $\mathcal{U}(t;s)$ be the two-parameter unitary group generated by $\mathcal{H}(\eta t)$:
\begin{equation}
i\partial_{t} \mathcal{U}(t;s) = \mathcal{H}(\eta t) \mathcal{U}(t;s)\;,\qquad \mathcal{U}(s;s) = \mathbbm{1}\;.
\end{equation}
The evolution of the Gibbs state $\rho_{\beta, \mu, L}$ is given as:
\begin{equation}
\rho(t) = \lim_{T\to +\infty} \mathcal{U}(t;-T) \rho_{\beta, \mu, L} \mathcal{U}(t;-T)^{*}\;.
\end{equation}
As in the one-dimensional case, we will be interested in the variation of physical observables,
\begin{equation}\label{eq:diffO3}
\Tr \mathcal{O} \rho(t) - \Tr \mathcal{O} \rho_{\beta, \mu, L}\;,
\end{equation}
for $\mathcal{O}$ given by the density and by the current operator, defined below.
\paragraph{Density and current operators.} Similarly to Section \ref{sec:1d}, we set:
\begin{equation}
n_{{\bf x}} = a^{*}_{{\bf x}} a_{{\bf x}}\;,\qquad n_{x}=\sum_{\rho\in S_{M}} n_{(x,\rho)}\;.
\end{equation}
Proceeding as after Eq. (\ref{eq:conteq}), we obtain the lattice continuity equation:
\begin{equation}\label{eq:cons2d}
\partial_{t} \tau_{t}( n_{x} ) = -\mathrm{div}_{x} \tau_{t}(\vec \jmath_{x})\equiv - \mathrm d_{1} j_{1,x} - \mathrm d_{2} j_{2,x}\;,
\end{equation}
where $\mathrm d_{i}f_{x}=f(x)-f(x-e_{i})$ and with the current densities:
\begin{equation}
\begin{split}
j_{1,x} &= j_{(x,x+e_{1})} + \frac{1}{2}\left[ j_{(x,x+e_{1}-e_{2})} + j_{(x,x+e_{1}+e_{2})} +j_{(x-e_{2},x+e_{1})} + j_{(x+e_{2},x+e_{1})}\right]\\
j_{2,x} &= j_{(x,x+e_{2})} + \frac{1}{2}\left[ j_{(x,x+e_{2}-e_{1})} + j_{(x,x+e_{2}+e_{1})} +j_{(x-e_{1},x+e_{2})} + j_{(x+e_{1},x+e_{2})}\right]\;.
\end{split}
\end{equation}
%
%
%
Setting $j_{0,x} = n_{x}$, we define the smeared current and density operators as:
\begin{equation}\label{eq:jmuphi2d}
j_{\nu}(\mu_{\theta}) = \sum_{x\in \Gamma_{L}} j_{\nu,x} \theta \mu(\theta x)\;.
\end{equation}
Also, we introduce the current and the density operators associated with a strip of length $\ell$ in proximity of the $x_{2} = 0$ boundary as:
\begin{equation}\label{eq:jell}
j_{\nu,x_{1}}^{\ell}:=\sum_{x_{2}=0}^{\ell-1} j_{\nu,(x_{1},x_{2})}\;.
\end{equation}
\begin{figure}
    \centering
    \includegraphics[scale=0.4]{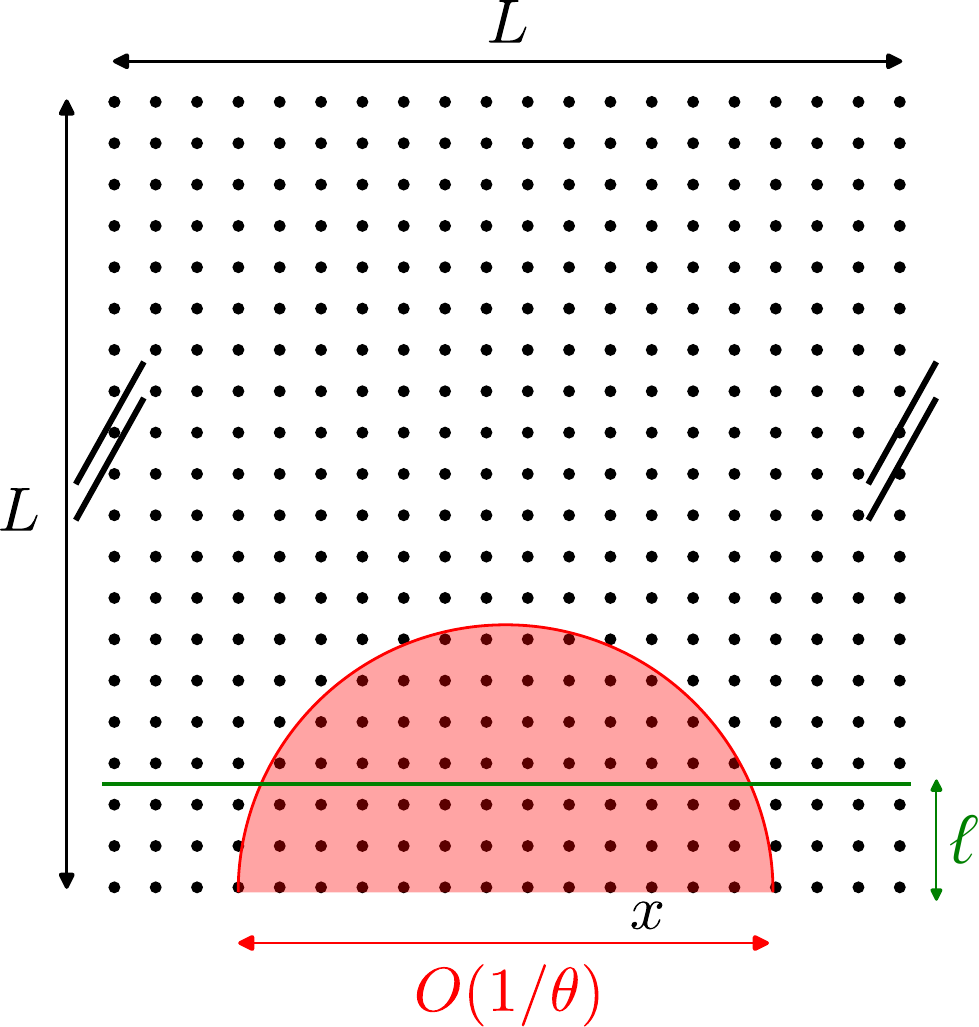}
    \caption{Representation of the lattice $\Gamma_{L}$: the red semi-circle is the typical support of the perturbation $\theta \mu(\theta x)$, whereas $1\ll \ell \ll 1/\theta$ is the length of the fiducial line entering the definition of the edge 2-current $j_{\nu,x}^{\ell}$.}
    \label{fig:edgeperturb}
\end{figure}
We shall call $j_{0,x_{1}}^{\ell}$ the edge density and $j_{1,x_{1}}^{\ell}$ the edge current, corresponding to the point $x_{1}$ at the edge $x_{2}=0$. Finally, in what follows it will also be convenient to consider the total current across the fiducial line at $x_{1}$, $j_{\nu,x_{1}}^{L} \equiv j_{\nu,x_{1}}$. In Fourier space, for $\nu = 0,1$:
\begin{equation}\label{eq:current2d}
\begin{split}
\hat j_{\nu, p} &= \sum_{x_{1}} e^{-ipx_{1}} j_{\nu,x_{1}} \\
&= \frac{1}{L} \sum_{k\in B_{L}} \sum_{y,y'=0}^{L-1} \big(a^{*}_{(k-p,y)}, \hat J_{\nu}(k,p;y,y') a_{(k,y')}\big)
\end{split}
\end{equation}
where the brackets denote summation over the internal degrees of freedom, as in (\ref{eq:currFou}), and:
\begin{equation}
\hat J_{0}(k,p)=1\qquad \hat J_{1}(k,p)=i\frac{\hat H(k)-\hat H(k-p)}{1-e^{-ip}}\;.
\end{equation}
In terms of these operators, the current (\ref{eq:jell}) can be expressed in momentum space as:
\begin{equation}
\hat j^{\ell}_{\nu, p} = \frac{1}{L} \sum_{k\in B_{L}} \sum_{y,y'=0}^{L-1} \big(a^{*}_{(k-p,y)}, \hat J^{\ell}_{\nu}(k,p;y,y') a_{(k,y')}\big)
\end{equation}
where $\hat J^{\ell}_{\nu}(k,p;y,y')$ is a bounded and smooth kernel such that:
\begin{equation}
\begin{split}
\hat J^{\ell}_{\nu}(k,p;y,y') &= 0\quad\qquad\qquad\qquad \text{if $y > \ell$ or $y' > \ell$,} \\
\hat J^{\ell}_{\nu}(k,p;y,y') &= \hat J_{\nu}(k,p;y,y')\qquad \text{if $y < \ell-1$ and $y' < \ell-1$.}
\end{split}
\end{equation}

\subsection{Validity of edge linear response}
In analogy with the one-dimensional case, let us define:
\begin{equation}\label{eq:resp2d}
\chi^{\beta, L,\ell}_{\nu}(x;\eta,\theta) := \frac{1}{\theta} \Big( \Tr j^{\ell}_{\nu,x} \rho(0) - \Tr j^{\ell}_{\nu,x} \rho_{\beta, \mu, L}\Big)\qquad \nu = 0, 1\;.
\end{equation}
The quantity $\chi^{\beta, L,\ell}_{\nu}$ describes the edge response of the system, after introducing a time-dependent perturbation in proximity of the edge, as in (\ref{eq:Pdef1}).  We are interested in studying this quantity in the following order of limits: first $L \to \infty$; then $\beta \to \infty$; then $\eta, \theta \to 0$; and then $\ell \to \infty$. 
%
%
%
\begin{theorem} [Edge response of two-dimensional systems] \label{thm:main2d} Let $\theta = a\eta$. There exist constants $w,\gamma,\eta_{0} > 0$ such that, for all $0\leq\eta<\eta_{0}$ and $a \leq w\left|\log\eta\right|$, for $\beta, L$ large enough:
\begin{equation}\label{eq:mainclaim2d}
\begin{split}
&\chi^{\beta, L,\ell}_{\nu}(x;\eta,\theta) \\
&\quad =-\sum_{\omega}^{*}\frac{v_{\omega}^{\nu}}{2\pi|v_{\omega}|}\int_\R \hat\mu_{\infty}(q,0)e^{iq\theta x}\frac{v_\omega q}{-i/a+v_\omega q}\,\frac{dq}{2\pi} + O(\ell\eta^{\gamma})+O( e^{-c\ell}) + O\Big(\frac{\ell}{\eta^{4} \beta}\Big)
\end{split}
\end{equation}
where the asterisk denotes summation over the edge modes localised near the lower edge, recall (\ref{eq:edgeest}), and where $v_{\omega}^{\nu}=\delta_{\nu,0}+ v_{\omega}\delta_{\nu,1}$.
\end{theorem}
\begin{corollary}\label{cor:2d}
In particular, we can isolate three regimes.
\begin{enumerate}
\item Suppose that $a \to 0$ as $\eta\to0^{+}$. Then, 
\begin{equation}
\chi^{\beta,L,\ell}_{\nu}(x,\eta,\theta) = O(a) + O(\ell\eta^{\gamma})+O( e^{-c\ell}) + O\Big(\frac{\ell}{\eta^{4} \beta}\Big)\;.
\end{equation}
\item Suppose that $a$ is constant in $\eta$. Then, calling $\Theta$ the Heaviside step function:
\begin{equation}\label{eq:explicit2d}
\begin{split}
\chi^{\beta,L,\ell}_{\nu}(x;\eta,\theta) &= -\sum_{\omega}^{*}\frac{v_{\omega}^{\nu}}{2\pi|v_{\omega}|}\int_\R \mu_{\infty}(\theta x-y,0)\Big[\delta(y)-\frac{1}{a|v_{\omega}|} e^{-y/av_{\omega}}\Theta(y/v_{\omega})\Big]\,dy\\&\quad + O(\ell\eta^{\gamma})+O( e^{-c\ell}) + O\Big(\frac{\ell}{\eta^{4} \beta}\Big)\;.
\end{split}
\end{equation}
\item Finally, suppose that $a\to \infty$ as $\eta \to 0^{+}$. Then:
\begin{equation}
\chi^{\beta,L,\ell}_{\nu}(x;\eta,\theta) = -\mu_{\infty}(\theta x,0)\sum_{\omega}^{*}\frac{v_{\omega}^{\nu}}{2\pi|v_{\omega}|} + O(1/a) + O(\ell\eta^{\gamma})+O( e^{-c\ell}) + O\Big(\frac{\ell}{\eta^{4} \beta}\Big)\;;
\end{equation}
in particular, 
\begin{equation}\label{eq:edgecond}
\lim_{\ell\to\infty}\lim_{\eta\to0^{+}} \lim_{\beta \to \infty} \lim_{L\to \infty} \chi_{1}^{\beta, L, \ell}(x;\eta,\theta)=-\mu_{\infty}(0,0)\sum_{\omega}^{*}\frac{\mathrm{sgn}(v_{\omega})}{2\pi}\;.
\end{equation}
\end{enumerate}
\end{corollary}
\begin{remark} 
\begin{itemize}
\item[(i)] Theorem \ref{thm:main2d} and Corollary \ref{cor:2d} extend Theorem \ref{thm:main1d} and Corollary \ref{cor:1d} to the case of edge currents of $2d$ quantum systems. In contrast to the $1d$ case, the edge conductance (\ref{eq:edgecond}) might be non-zero. Whenever that is the case, by the bulk-edge duality the system is in a non-trivial topological phase.
\item[(ii)] The proof of Theorem \ref{thm:main2d} is an adaptation of the proof of Theorem \ref{thm:main1d}, discussed below. It relies on the emergent one-dimensional nature of the edge modes, and on the loop cancellation for chiral relativistic fermions, Proposition \ref{prp:loopcanc}.
\end{itemize}
\end{remark}
\subsection{Proof of Theorem \ref{thm:main2d}}
\paragraph{Auxiliary dynamics.} Let:
\begin{equation}
\mathcal{H}_{\beta,\eta}(t) = \mathcal{H} + \theta e^{\eta_{\beta} t} \sum_{{\bf x} \in \Lambda_{L}} \mu_{\alpha}(\theta x) a^{*}_{{\bf x}} a_{{\bf x}}\;,
\end{equation}
with:
\begin{equation}
\mu_{\alpha}(\theta x) := \frac{1}{L} \sum_{p \in B_{L}} e^{ipx} \hat\mu_{\alpha,\theta}(p,x_{2})\;,\qquad \hat\mu_{\alpha,\theta}(p,x_{2}) := \hat \mu_{\theta}(p,x_{2}) \chi( \theta^{\alpha-1} |p|_{\mathbb{T}})\;,
\end{equation}
where $\alpha\in(0,1)$. As in the $1d$ case, we are cutting off all momenta $p$ of norm greater than $2\theta^{1-\alpha}$. Let $\widetilde{\mathcal{U}}(t;s)$ the two-parameter unitary group generated by $\mathcal{H}_{\beta,\eta}(t)$. Then, the following holds.
\begin{proposition}[Approximation by the auxiliary dynamics.]\label{prp:LR2d} Under the same assumptions of Theorem \ref{thm:main2d}, it follows that, for any $m\in \mathbb{N}$:
\begin{equation}\label{eq:LRprop2d}
\Big\| \widetilde{\mathcal{U}}(t;-\infty)^{*} j_{\nu,x}^{\ell} \widetilde{\mathcal{U}}(t;-\infty) - \mathcal{U}(t;-\infty)^{*} j_{\nu,x}^{\ell} \mathcal{U}(t;-\infty)   \Big\| \leq C_{m} \frac{\theta^{1 + \alpha (m -1)}\ell}{\eta^{3}} + \frac{C\theta\ell}{\eta^{4} \beta}.
\end{equation}
\end{proposition}
\begin{proof}
The proof is obtained by straightforwardly adapting the Proof of Proposition \ref{prp:LR} to this two-dimensional setting, with the following differences.  The first is that the Lieb-Robinson ball has volume proportional to $|t-s|^{2}$ (instead of $|t-s|$), and this produces errors proportional to $\eta^{-4}$ bound (instead of $\eta^{-3}$); the second is the presence of the factor $\ell$ in the estimate, due to $\| j_{\nu}^{\ell} \| \leq C\ell$.
\end{proof}
Thanks to Proposition \ref{prp:LR2d} we can write
\begin{equation}\label{eq:approxdyn2d}
\Tr j^{\ell}_{\nu,x} \rho(t) = \Tr j_{\nu,x}^{\ell} \widetilde\rho(t) + \mathcal E_{\nu}^{\beta,L,\ell}(x,t;\eta,\theta)\;,
\end{equation}
where $\tilde \rho(t) = \widetilde{\mathcal{U}}(t;-\infty) \rho_{\beta, \mu, L} \widetilde{\mathcal{U}}(t;-\infty)^{*}$ and the error term $\mathcal E_{\nu}^{\beta,L,\ell}(x,t;\eta,\theta)$ takes into account the approximation by the auxiliary dynamics:

\begin{equation}\label{eq:erraux}
| \mathcal E_{\nu}^{\beta,L,\ell}(x,t;\eta,\theta) | \leq C_{m} \frac{\theta^{1 + \alpha (m -1)}\ell}{\eta^{3}} + \frac{C\theta\ell}{\eta^{4} \beta}\;.
\end{equation}
We require that this error term, once divided by $\theta = a\eta$ (recall the definition (\ref{eq:resp2d})), vanishes as $\eta\to0^{+}$. This forces to choose the parameter $a$ such that:
\begin{equation}\label{eq:errcond2d}
a=o\left(\eta^{\frac{3}{\alpha(m-1)}-1}\right)\;.
\end{equation}
To allow for a divergent $a$ as $\eta \to 0^{+}$, we choose $m\in\N$ such that $\alpha(m-1)>3$. Later, we will have to introduce stronger constraints on $a$.
\paragraph{Analysis of the Duhamel series.} By Proposition \ref{prop:wick}, we can represent the Duhamel series as:
\begin{equation}
\begin{split}
& \Tr j^{\ell}_{\nu,x} \widetilde\rho(0) -  \Tr j^{\ell}_{\nu,x} \rho_{\beta, \mu, L} \\&= \sum_{n= 1}^{\infty}\frac{(-1)^{n}}{n!} \int_{[0,\beta)^{n}} d\underline{s}\, e^{-i\eta_{\beta}(s_{1} + \ldots +s_{n})}\langle {\bf T} \gamma_{s_{1}}(j_{0}(\mu_{\alpha,\theta}));  \cdots; \gamma_{s_{n}}(j_{0}(\mu_{\alpha,\theta})); j^{\ell}_{\nu,x}\rangle_{\beta,L}\;.
 \end{split}
\end{equation}
Proceeding as in (\ref{eq:currfourier})-(\ref{eq:densitytimefourier}), one obtains
\begin{equation}\label{eq:response2d}
\begin{split}
\Tr j^{\ell}_{\nu,x} \widetilde\rho(0) -  \Tr j^{\ell}_{\nu,x} \rho_{\beta, \mu, L} &= \sum_{n= 1}^{\infty}\frac{(-\theta)^{n}}{n!L^{n}} \sum_{\{p_{i}\}\in B_{L}^{n}} \sum_{\{y_{i}\} \in \{0, \ldots, L-1\}^{n}} \Big[\prod_{i=1}^{n} \hat\mu_{\alpha,\theta}(-p_{i}, y_{i})\Big] e^{-ip_{n+1}x} \\&\quad \cdot \frac{1}{\beta L}\langle {\bf T}\, \hat n_{\underline{p}_{1},y_{1}}\;; \cdots \;; \hat n_{\underline{p}_{n},y_{n}}\;; \hat\jmath^{\ell}_{\nu,\underline{p}_{n+1}} \rangle_{\beta, L}\;,
 \end{split}
\end{equation}
with the usual conventions $\ul p_{i} = (\eta_{\beta},p_{i})$, $\ul p_{n+1} =-\sum_{i=1}^{n} \ul p_{i}$. The proof of Theorem \ref{thm:main2d} is based on a detailed analysis of the cumulant expansion in (\ref{eq:response2d}). By Wick's rule:
\begin{equation}\label{eq:corrwick2d}
\begin{split}
&\frac{1}{\beta L}\langle {\bf T}\, \hat n_{\underline{p}_{1},y_{1}}\;; \cdots \;; \hat n_{\underline{p}_{n},y_{n}}\;; \hat\jmath^{\ell}_{\nu,\underline{p}_{n+1}} \rangle_{\beta, L} \\
&= - \sum_{\pi\in S_{n}} \int_{\beta,L} \frac{d\underline{k}}{(2\pi)^{2}} \sum_{z, z'} \sum_{\{y_{i}\}}  \Tr \Big[ \hat J^{\ell}_{\nu}(k, p_{n+1}; z,z')\prod_{i = 1}^{n+1} g\Big(\underline{k} + \sum_{j <i} \underline{p}_{\pi(j)}; y_{\pi(i-1)},y_{\pi(i)}\Big) \Big]
\end{split}
\end{equation}
where the trace is over the internal degrees of freedom, we set $y_{\pi(0)}\equiv z'$ and $y_{\pi(n+1)}\equiv z$, and where the propagator $g$ is given by
\begin{equation}
g(\underline{k}) = \frac{1}{ik_{0} + \hat H(k) - \mu}\;.
\end{equation}
As done after (\ref{eq:gfou}), we decompose the propagator in a singular plus a more regular part:
\begin{equation}\label{eq:g2d}
g(\underline{k}) = g_{\text{a}}(\underline{k}) + g_{\text{b}}(\underline{k})\;,
\end{equation}
where: 
\begin{equation}
 g_{\text{b}}(\underline{k}) = g(\underline{k}) [1-\chi(|\hat H(k) - \mu|/\Delta)]
\end{equation}
with $\mu,\Delta$ as in Assumption \ref{ass:B}. The propagator $ g_{\text{b}}(\underline{k})$ satisfies, by a Combes-Thomas argument (see {\it e.g.} \cite[Appendix A]{CMP} for more details):
\begin{equation}
\Big| \text{d}_{k_{0}}^{n_{0}} \text{d}_{k_{1}}^{n_{1}} g_{\text{b};\rho,\rho'}(\underline{k}; y,z) \Big| \leq \frac{C_{n_{0}, n_{1}}}{1 + |k_{0}|} e^{-c|x_{2} - y_{2}|}\;.
\end{equation}
Instead, the propagator $g_{\text{a}}(\underline{k})$ is given by:
\begin{equation}
g_{\text{a}}(\underline{k}) = \sum_{\omega=1}^{N} \frac{\chi(|e_{\omega}(k) - \mu|/\Delta)}{ik_{0}+e_{\omega}(k)-\mu} P_{\omega}(k)
\end{equation}
where $e_{\omega}$ is the dispersion relation of the $\omega$ edge mode, and $P_{\omega}$ is the rank-$1$ projector onto the edge mode $\xi_{\omega}$. Next, proceeding as after (\ref{eq:omegadelta}), we extract from $g_{\text{a}}(\underline{k})$ its relativistic part, up to remainder term which is more regular. All in all, we have:
\begin{equation}\label{eq:g2dsplit}
\begin{split}
&g(\underline{k}) = g_{\text{s}}(\underline{k}) + g_{\text{r}}(\underline{k}) \\
g_{\text{s}}(\underline{k}) = \sum_{\omega=1}^{N} \frac{\chi^{\omega}_{\delta}(\underline k)}{D_{\omega}(\underline k)} P_{\omega}(k)&\;,\qquad \Big\| \text{d}^{n}_{k_{\alpha}} g_{\mathrm{r}}(\underline{k}; y,z) \Big\| \leq \sum_{\omega = 1}^{N} \frac{C_{n}}{1+|k_{0}|} \frac{e^{-c|y-z|}}{\|\underline{k} - \underline{k}_{F}^{\omega}\|^{n}}\;,
\end{split}
\end{equation}
with $\chi^{\omega}_{\delta}(\underline k)$ and $D_{\omega}(\underline{k})$ as in (\ref{eq:omegadelta}). We omit the details. 

Proceeding as in (\ref{eq:split1d}), we rewrite the correlation function in (\ref{eq:corrwick2d}) as:
\begin{equation}
\begin{split}
&\frac{1}{\beta L}\langle {\bf T}\, \hat n_{\underline{p}_{1},y_{1}}\;; \cdots \;; \hat n_{\underline{p}_{n},y_{n}}\;; \hat\jmath^{\ell}_{\nu,\underline{p}_{n+1}} \rangle_{\beta, L} \\&\qquad = \widetilde S^{\beta,L,\ell}_{n;\nu}(\underline{p}_{1}, y_{1}, \ldots, \underline{p}_{n}, y_{n}) + \widetilde R^{\beta,L,\ell}_{n;\nu}(\underline{p}_{1}, y_{1}, \ldots, \underline{p}_{n}, y_{n})\;,
\end{split}
\end{equation}
where:
\begin{equation}\label{eq:550}
\begin{split}
&\widetilde S^{\beta,L,\ell}_{n;\nu}(\underline{p}_{1}, y_{1}, \ldots, \underline{p}_{n}, y_{n}) \\
&= - \sum_{\pi\in S_{n}} \int_{\beta,L} \frac{d\underline{k}}{(2\pi)^2} \sum_{z, z'}  \Tr \Big[ \hat J^{\ell}_{\nu}(k, p_{n+1}; z,z')\prod_{i = 1}^{n+1} g_{\text{s}}\Big(\underline{k} + \sum_{j <i} \underline{p}_{\pi(j)}; y_{\pi(i-1)},y_{\pi(i)}\Big) \Big]\;,
\end{split}
\end{equation}
and where the error term $\widetilde R^{\beta,L,\ell}_{n;\nu}$ contains at least one propagator $g_{\text{r}}$. The next lemma is the analogue of Lemma \ref{lemma:singular}.
\begin{lemma}[Structure of the singular part] For any $p_{1},\dots,p_{n}$ in the support of $\hat\mu_{\alpha,\theta}$,
\begin{equation}\label{eq:splitting2}
\widetilde S^{\beta,L,\ell}_{n;\nu}(\underline{p}_{1}, y_{1}, \ldots, \underline{p}_{n}, y_{n}) = S^{\beta,L,\ell}_{n;\nu}(\underline{p}_{1},y_{1}, \ldots, \underline{p}_{n},y_{n}) + T^{\beta,L,\ell}_{n;\nu}(\underline{p}_{1},y_{1}, \ldots, \underline{p}_{n},y_{n})
\end{equation}
where: for $g_{\omega}(\underline{k}) = \chi_{\delta}^{\omega}(\underline{k}) / D_{\omega}(\underline{k})$,
\begin{equation}\label{eq:singular2d}
\begin{split}
&S^{\beta,L,\ell}_{n;\nu}(\underline{p}_{1},y_{1}, \ldots, \underline{p}_{n},y_{n}) = - \sum_{\omega = 1}^{N} v^{\nu,\ell}_{\omega} \prod_{i=1}^{n} \| \xi_{\omega}(k_{F}^{\omega}; y_{i}) \|^{2} \int_{\beta,L} \frac{d\underline{k}}{(2\pi)^{2}} \sum_{\pi\in S_{n}} \prod_{i=1}^{n+1} g_{\omega}\Big(\ul k+\sum_{j< i} \ul p_{\pi(j)}\Big)\\
&T^{\beta,L,\ell}_{n;\nu}(\underline{p}_{1},y_{1}, \ldots, \underline{p}_{n},y_{n}) \\
&\quad = - \sum_{\omega = 1}^{N} \int_{\beta,L} \frac{d\underline{k}}{(2\pi)^{2}} \sum_{\pi\in S_{n}}f_{n,\omega}^{\nu,\ell}(k;p_{\pi(1)},y_{\pi(1)}, \dots,p_{\pi(n)}, y_{\pi(n)}) \prod_{i=1}^{n+1} g_{\omega}\Big(\ul k+\sum_{j< i} \ul p_{\pi(j)}\Big)\;,
\end{split}
\end{equation}
with $v^{\nu,\ell}_{\omega}$ such that:
\begin{equation}
\begin{split}
\big| v^{\nu,\ell}_{\omega} - \delta_{\nu,0} - \delta_{\nu,1} v_{\omega} \big| \leq Ce^{-c\ell} \qquad &\text{if the edge mode $\omega$ is localized at $x_{2} = 0$} \\
\big| v^{\nu,\ell}_{\omega}  \big| \leq Ce^{-c(L - \ell)} \qquad &\text{if the edge mode $\omega$ is localized at $x_{2} = L-1$}\;, 
\end{split}
\end{equation}
and $\| \xi_{\omega}(k_{F}^{\omega}; y_{i}) \|^{2} = \sum_{\rho=1}^{M} | \xi_{\omega;\rho}(k_{F}^{\omega}; y_{i}) |^{2}$. Furthermore, 
%
%
\begin{equation}\label{eq:rembound2}
|f_{n,\omega}^{\nu,\ell}(k;p_{1},y_{1}, \dots,p_{n}, y_{n})| \leq C \Big(|k-k_{F}^{\omega}|_{\mathbb{T}}\delta_{\nu,1}+ n\sum_{i=1}^{n} |p_{i}|_{\mathbb{T}}\Big)e^{- c\sum_{i=1}^{n} |y_{i}|_{\omega}}\;.
\end{equation}
\end{lemma}
\begin{proof} Let us write $g_{\text{s}}(\underline{k}) = \sum_{\omega} g_{\omega}(\underline{k}) P_{\omega}(k)$ in (\ref{eq:550}). As in the one-dimensional case, for $\eta, |p_{i}|$ small enough, in the loop integral no jump between different chiralities is allowed. Thus, we have: 
\begin{equation}\label{eq:trace}
\begin{split}
&\widetilde S^{\beta,L,\ell}_{n;\nu}(\underline{p}_{1}, y_{1}, \ldots, \underline{p}_{n}, y_{n}) \\
&= - \sum_{\pi\in S_{n}}\sum_{\omega = 1}^{N} \int_{\beta,L} \frac{d\underline{k}}{(2\pi)^2} \sum_{z, z'}  \Tr \Big[ \hat J^{\ell}_{\nu}(k, p_{n+1}; z,z')\prod_{i = 1}^{n+1} P_{\omega} \Big(k+ \sum_{j <i} p_{\pi(j)}; y_{\pi(i-1)},y_{\pi(i)}\Big)\Big] \\
&\qquad \cdot \prod_{i=1}^{n+1} g_{\omega}\Big(\underline{k} + \sum_{j <i} \underline{p}_{\pi(j)}\Big)\;.
\end{split}
\end{equation}
Consider the trace in (\ref{eq:trace}). We rewrite it as:
\begin{equation}\label{eq:555}
\begin{split}
&\sum_{z, z'}  \Tr \Big[ \hat J^{\ell}_{\nu}(k, p_{n+1}; z,z')\prod_{i = 1}^{n+1} P_{\omega} \Big(k + \sum_{j <i} p_{\pi(j)}; y_{\pi(i-1)},y_{\pi(i)}\Big)\Big] \\
& = \sum_{z, z'}  \Tr \Big[ \hat J^{\ell}_{\nu}(k, 0; z,z')\prod_{i = 1}^{n+1} P_{\omega} \Big(k; y_{\pi(i-1)},y_{\pi(i)}\Big)\Big] + \mathcal{E}_{n;\omega;1}^{\nu}(k,p_{1}, \ldots, p_{n}; y_{1}, \ldots, y_{n})\;;
\end{split}
\end{equation}
using the decay estimates of the edge modes (\ref{eq:edgeest}) and the smoothness of the kernel of $\hat J_{\nu}(k, p)$, the error term can be estimated as:
\begin{equation}
| \mathcal{E}_{n;\omega;1}^{\nu}(k,p_{1}, \ldots, p_{n}; y_{1}, \ldots, y_{n}) | \leq Cne^{-c\sum_{i = 1}^{n} |y_{i}|_{\omega}} \sum_{i} |p_{i}|_{\mathbb{T}}\;.
\end{equation}
Consider now the main term in (\ref{eq:555}). 
%
%
%
%
%
%
%
Let:
\begin{equation}
\sum_{z,z'} \sum_{\rho,\rho'}\overline{\xi_{\omega,\rho}(k; z)} \hat J^{\ell}_{\nu;\rho,\rho'}(k, 0; z,z') \xi_{\omega,\rho'}(k; z') =: v^{\nu,\ell}_{\omega}(k)\;;
\end{equation}
if $\omega$ labels an edge mode localized at $x_{2} = L-1$, then $|v^{\nu,\ell}_{\omega}(k)|\leq Ce^{-c(L - \ell)}$ by the exponential decay of the edge modes. Instead, if $\omega$ labels an edge mode localized at $x_{2} = 0$,
\begin{equation}
v^{0,L}_{\omega}(k) = 1\;,\qquad v^{1,L}_{\omega}(k) = \partial_{k} e_{\omega}(k)\;,
\end{equation}
where the last identity follows from the Feynman-Hellmann argument. Also, by the decay and the regularity of the edge modes,
\begin{equation}
\Big| v^{\nu,\ell}_{\omega}(k) - v^{\nu,L}_{\omega}(k)\Big| \leq C e^{-c|\ell|_{\omega}}\;,\qquad \Big| v^{\nu,\ell}_{\omega}(k) - v^{\nu,\ell}_{\omega}(k_{F}^{\omega}) \Big| \leq C|k - k_{F}^{\omega}|_{\mathbb{T}}\;.
\end{equation}
We then write:
\begin{equation}\label{eq:FH}
\sum_{z, z'}  \Tr \Big[ \hat J^{\ell}_{\nu}(k, 0; z,z')\prod_{i = 1}^{n+1} P_{\omega} \Big(k; y_{\pi(i)},y_{\pi(i-1)}\Big)\Big] = v_{\omega}^{\nu,\ell}(k) \prod_{i=1}^{n} \| \xi_{\omega}(k; y_{i}) \|^{2}\;.
\end{equation}
Combining (\ref{eq:trace})-(\ref{eq:FH}), we obtain the leading term in (\ref{eq:splitting2}), using once more the regularity properties (\ref{eq:edgeest}) of the edge modes. The bound for the error term (\ref{eq:rembound2}) follows putting together all the estimates for the error terms accumulated. We omit the details.
\end{proof}
Let:
\begin{equation}\label{eq:errR2d}
R^{\beta,L,\ell}_{n;\nu}(\underline{p}_{1},y_{1}, \ldots, \underline{p}_{n},y_{n}) = T^{\beta,L,\ell}_{n;\nu}(\underline{p}_{1},y_{1}, \ldots, \underline{p}_{n},y_{n}) + \widetilde{R}^{\beta,L,\ell}_{n;\nu}(\underline{p}_{1},y_{1}, \ldots, \underline{p}_{n},y_{n})\;.
\end{equation}
Recalling the definition (\ref{eq:resp2d}) of the response $\chi^{\beta, L,\ell}_{\nu}(x;\eta,\theta)$, and the expansion (\ref{eq:response2d}), we have:
\begin{equation}\label{eq:respexp2d}
\begin{split}
\chi^{\beta, L,\ell}_{\nu}(x;\eta,\theta) &= -\sum_{n= 1}^{\infty}\frac{(-\theta)^{n}}{n!L^{n}} \sum_{\{p_{i}\}\in B_{L}^{n}} \sum_{\{y_{i}\} \in \{0, \ldots, L-1\}^{n}} \Big[\prod_{i=1}^{n} \hat\mu_{\alpha,\theta}(-p_{i}, y_{i})\Big] e^{ip_{n+1}x} \\&\quad \cdot [S_{n;\nu}^{\beta, L,\ell}(\underline{p}_{1},y_{1}, \ldots, \underline{p}_{n},y_{n})+ R_{n;\nu}^{\beta, L, \ell}(\underline{p}_{1},y_{1}, \ldots, \underline{p}_{n},y_{n})] + \frac{1}{\theta}\mathcal E_{\nu}^{\beta,L,\ell}(x;\eta,\theta)\;,
\end{split}
\end{equation}
with the error $\mathcal E_{\nu}^{\beta,L,\ell}$ satisfying (\ref{eq:erraux}).
\paragraph{Evaluation of the linear response.} Let us consider the $\beta, L \to \infty$ contribution to the linear response,
\begin{equation}\label{eq:linlin2d}
\chi^{\ell, \text{lin}}_{\nu}(x;\eta,\theta)= - \int_{\mathbb{T}} \frac{dp}{2\pi} \sum_{y = 0}^{\infty} \hat\mu_{\alpha,\theta}(-p,y) e^{-ip x} [S_{1;\nu}^{\ell}(\ul p,y)+ R_{1;\nu}^{\ell}(\ul p,y)]\;.
\end{equation}
The next proposition is the analogue of Proposition \ref{prop:linear}.
\begin{proposition}[Evaluation of the linear response]\label{prop:linear2d} We have:
\begin{equation}\label{eq:lin2d}
\chi_{\nu}^{\ell,\mathrm{lin}}(x,\eta,\theta)=-\sum_{\omega}^{*}\chi_{\nu,\omega}\int_\R \frac{dq}{2\pi}\, \hat\mu_{\infty}(q,0)e^{iq\theta x}\frac{v_\omega q}{-i/a+v_\omega q} + O(\ell\theta^{\alpha}) + O(e^{-c\ell})\;,
\end{equation}
with $\chi_{\nu,\omega}=v_{\omega}^{\nu}/(2\pi |v_\omega|)$, and where the sum runs only over the edge modes localised  at $x_{2} = 0$.
\end{proposition}
\begin{proof} By the smoothness of the test function, recall (\ref{eq:muFou1}):
\begin{equation}
\big| \hat\mu_{\alpha,\theta}(p,y) - \hat\mu_{\alpha,\theta}(p,0) \big|  \leq \frac{C_{r} y}{1 + |p / \theta|_{\mathbb{T}_{{\theta}^{-1}}}^{r}}\qquad \text{for all $r\in \mathbb{N}$.}
\end{equation}
Using the exponential decay of the edge modes and the estimates on the propagator, see discussion after (\ref{eq:g2d}), we have:
\begin{equation}
\begin{split}
\chi^{\ell, \text{lin}}_{\nu}(x;\eta,\theta) &= - \int_{\mathbb{T}} \frac{d p}{2\pi} \sum_{y = 1}^{\infty} \hat\mu_{\alpha,\theta}(-p,0) e^{-ip x} [S_{1;\nu}^{\ell}(\ul p,y)+ R_{1;\nu}^{\ell}(\ul p,y)] + O(\ell \theta)\;.
\end{split}
\end{equation}
The same estimates also allow to show, for $\alpha>0$:
\begin{equation}\label{eq:Rdiff2d}
\sum_{y = 0}^{\infty} |R_{1;\nu}^{\ell}(\ul p,y) - R_{1;\nu}^{\ell}(\ul 0,y)|\leq C \ell \|\underline{p}\|^{\alpha} + C e^{-c\ell}\;.
\end{equation}
Consider the contribution of $S_{1;\nu}^{\ell}$ to the linear response. Recalling the equation (\ref{eq:bubble}) defining the relativistic bubble diagram $\mathfrak B^{\omega}_{\delta}$, we have:
\begin{equation}\label{eq:s1bubble2d}
\begin{split}
\sum_{y = 0}^{\infty} S_{1;\nu}^{\ell}(\ul p,y) &=  - \sum_{y = 0}^{\infty} \sum_{\omega}^{*} v_{\omega}^{\nu,\ell} \| \xi_{\omega}(k_{F}^{\omega}; y) \|^{2} \int \frac{d\underline{k}}{(2\pi)^{2}}\, g_{\omega}(\underline{k}) g_{\omega}(\underline{k} + \underline{p}) \\
&=\sum_{\omega}^{*} v_{\omega}^{\nu,\ell} \mathfrak B_{\delta}^{\omega}(\ul p)\;,
\end{split}
\end{equation}
where the sum is restricted to the edge modes localized at $y = 0$. Therefore:
\begin{equation}\label{eq:resp2d0}
\begin{split}
\chi^{\ell, \text{lin}}_{\nu}(x;\eta,\theta) &= - \int_{\mathbb{T}} \frac{dp}{2\pi} \sum_{y = 0}^{\infty} \hat\mu_{\alpha,\theta}(-p,0) e^{-ip x} [S_{1;\nu}^{\ell}(\ul p,y)+ R_{1;\nu}^{\ell}(\ul p,y)] + O(\ell \theta) \\
&= -\int_{\mathbb{T}} \frac{dp}{2\pi} \hat\mu_{\theta}(-p,0) e^{-ip x} \Big[\sum_{\omega}^{*} v_{\omega}^{\nu,\ell} \mathfrak B_{\delta}^{\omega}(\ul p) + \sum_{y = 0}^{\infty} R_{1;\nu}^{\ell}(\underline{0},y)\Big] + O( \ell \theta^{\alpha}) + O(e^{-c\ell})\;,
\end{split}
\end{equation}
where the error terms take into account (\ref{eq:Rdiff2d}) and the replacement of $\hat\mu_{\alpha,\theta}(p,0)$ with $\hat\mu_{\theta}(p,0)$. Next, we fix the second term in the integral using the lattice continuity equation. Let us introduce the shorthand $\mathcal O_{\ul x, x_{2}}=\gamma_{x_{0}}(\mathcal O_{(x_{1}, x_{2})})$; then, Eq. (\ref{eq:cons2d}) implies:
\begin{equation}
\begin{split}
i\partial_{x_{0}}\langle\mathbf T n_{\ul x,x_{2}};j^{\ell}_{\nu,\ul x'}\rangle_{\beta,L} + \mathrm d_{x_{1}} \langle\mathbf T j_{1,\ul x,x_{2}};j^{\ell}_{\nu,\ul x'}\rangle_{\beta,L} &+ \mathrm d_{x_{2}} \langle\mathbf T j_{2,\ul x,x_{2}};j^{\ell}_{\nu,\ul x'}\rangle_{\beta,L}\\&= i\delta(x_{0}-x'_{0}) \langle [n_{(x,x_{2})},j^{\ell}_{\nu, x'}]\rangle_{\beta,L}\;.
\end{split}
\end{equation}
Taking the Fourier transform, summing over $x_{2}$, and using the Dirichlet boundary conditions:
\begin{equation}
p_{0} \frac{1}{\beta L}\langle\mathbf T \hat n^{L}_{\ul p};\hat\jmath^{\ell}_{\nu,-\ul p}\rangle_{\beta,L} + (1-e^{-ip}) \frac{1}{\beta L} \langle\mathbf T \hat\jmath_{1,\ul p}^{L};\hat\jmath^{\ell}_{\nu,-\ul p}\rangle_{\beta,L} = i\sum_{x\in \Gamma_{L}} e^{ipx_{1}}\langle[n_{x},j^{\ell}_{\nu,0}]\rangle_{L}\;.
\end{equation}
Choosing $\ul p = (\eta_{\beta},0)\neq\ul 0$, the right-hand side vanishes because $\sum_{x\in\Gamma_{L}} n_{x}=\mathcal N$ is the number operator, which commutes with $j^{\ell}_{\nu,0}$. Thus:
\begin{equation}
\sum_{x_{2} = 0}^{L-1} \langle\mathbf T \hat n_{(\eta,0),x_{2}};\hat\jmath^{\ell}_{\nu,(-\eta,0)}\rangle_{\beta, L}=0\quad\Longrightarrow\quad \sum_{y = 0}^{L-1} R_{\nu;1}^{\beta,L,\ell}((\eta_{\beta},0),y) = -\sum_{y = 0}^{L-1} S_{\nu;1}^{\beta,L,\ell}((\eta_{\beta},0),y)\;.
\end{equation}
Taking the limit $\beta, L \to \infty$, and using the continuity in $\underline{p}$ of $R_{\nu;1}^{\ell}(\underline{p}, y)$, we can determine the second term in the integral in (\ref{eq:resp2d0}) in terms of (\ref{eq:s1bubble2d}). We ultimately get:
\begin{equation}
\chi^{\ell,\mathrm{lin}}_{\nu}(x,\eta,\theta)=-
\sum_{\omega}^{*}\frac{v^{\nu,\ell}_{\omega}}{2\pi|v_{\omega}|}\int_{\mathbb{T}} \frac{dp}{(2\pi)}\, \hat\mu_{\theta}(p,0) e^{ipx}\frac{v_{\omega}p}{-i\eta+v_{\omega} p}+ O( \ell \theta^{\alpha}) + O(e^{-c\ell})\;.
\end{equation}
The final claim (\ref{eq:lin2d}) follows recalling the definition of $\hat\mu_{\theta}(p,0)$ in terms of $\hat \mu_{\infty}(p/\theta, 0)$, Eq. (\ref{eq:muFou1}), using the fast decay of $\hat \mu_{\infty}(p)$, performing the change of variables $p/\theta = q$, and recalling that $|v^{\nu,\ell}_{\omega} - v^{\nu}_{\omega}|\leq Ce^{-c\ell}$.
\end{proof}
We are now ready to prove the main result of this section, Theorem \ref{thm:main2d}.
\begin{proof}[Proof of Theorem \ref{thm:main2d}] The starting point is the identity (\ref{eq:respexp2d}). The explicit form of the linear response term, Eq. (\ref{eq:explicit2d}), is obtained proceeding as in Proposition \ref{prop:linear}. Let us now consider the higher order terms,
\begin{equation}\label{eq:hot2d}
\begin{split}
&-\sum_{n= 2}^{\infty}\frac{(-\theta)^{n-1}}{n!L^{n}} \sum_{\{p_{i}\} \in B_{L}^{n}} \sum_{\{y_{i}\} \in \{0, \ldots, L-1\}^{n}} \Big[\prod_{i=1}^{n} \hat\mu_{\alpha,\theta}(-p_{i}, y_{i})\Big] e^{ip_{n+1}x} \\&\qquad \quad \cdot [S_{n;\nu}^{\beta, L,\ell}(\underline{p}_{1},y_{1}, \ldots, \underline{p}_{n},y_{n})+ R_{n;\nu}^{\beta, L, \ell}(\underline{p}_{1},y_{1}, \ldots, \underline{p}_{n},y_{n})]\;.
\end{split}
\end{equation}
We claim that, for some $\gamma>0$:
\begin{equation}\label{eq:hot2dest}
|(\ref{eq:hot2d})| \leq C\ell \eta^{\gamma}\;.
\end{equation}
This estimate, combined with the evaluation of the linear response, Proposition \ref{prop:linear2d}, and with the bound (\ref{eq:erraux}) for the error terms produced by the comparison with the auxiliary dynamics, implies the claim of Theorem \ref{thm:main2d}, Eq. (\ref{eq:erraux}).

Let us prove (\ref{eq:hot2dest}). Consider the contribution associated with $S_{n;\nu}^{\beta, L,\ell}$. By (\ref{eq:singular2d}), 
\begin{equation}\label{eq:Sncontribution}
\begin{split}
&-\sum_{n= 2}^{\infty}\frac{(-\theta)^{n-1}}{n!L^{n}} \sum_{\{p_{i}\} \in B_{L}^{n}} \sum_{\{y_{i}\} \in \{0, \ldots, L-1\}^{n}} \Big[\prod_{i=1}^{n} \hat\mu_{\alpha,\theta}(-p_{i}, y_{i})\Big] e^{ip_{n+1}x} S_{n;\nu}^{\beta, L,\ell}(\underline{p}_{1},y_{1}, \ldots, \underline{p}_{n},y_{n}) \\
&\quad = -\sum_{n= 2}^{\infty}\frac{(-\theta)^{n-1}}{n!L^{n}} \sum_{\{p_{i}\} \in B_{L}^{n}} \sum_{\omega =1 }^{N} \Big[\prod_{i=1}^{n} \tilde \mu_{\alpha,\theta;\omega}(-p_{i})\Big] e^{ip_{n+1}x} S_{n;\nu;\omega}^{\beta, L,\ell}(\underline{p}_{1} \ldots, \underline{p}_{n})
\end{split}
\end{equation}
where:
\begin{equation}
\tilde \mu_{\alpha,\theta;\omega}(p_{i}) = \sum_{y_{i}=0}^{L-1}\hat\mu_{\alpha,\theta}(p_{i}, y_{i}) \| \xi_{\omega}(k_{F}^{\omega}; y_{i}) \|^{2}\;,
\end{equation}
and:
\begin{equation}
S_{n;\nu;\omega}^{\beta, L,\ell}(\underline{p}_{1} \ldots, \underline{p}_{n}) =  - v^{\nu,\ell}_{\omega} \int_{\beta,L} \frac{d\underline{k}}{(2\pi)^{2}} \sum_{\pi\in S_{n}} \prod_{i=1}^{n+1} g_{\omega}\Big(\ul k+\sum_{j< i} \ul p_{\pi(j)}\Big)\;.
\end{equation}
The right-hand side of eq. (\ref{eq:Sncontribution}) is the same type of contribution that has been studied for one-dimensional systems. By Proposition \ref{prp:loopest}, and proceeding as in the proof of Theorem \ref{thm:main1d}, we have, for $\gamma > 0$:
\begin{equation}
\sum_{n= 2}^{\infty}\frac{\theta^{n-1}}{n!L^{n}} \sum_{\{p_{i}\} \in B_{L}^{n}} \sum_{\omega = 1}^{N} \Big[\prod_{i=1}^{n} \big|\tilde \mu_{\alpha,\theta,\omega}(-p_{i}) \big|\Big] \Big| S_{n;\nu;\omega}^{\beta, L,\ell}(\underline{p}_{1} \ldots, \underline{p}_{n}) \Big| \leq C\eta^{\gamma}\;.
\end{equation}
Let us now discuss the error terms $R_{n;\nu}^{\beta, L, \ell}$ in (\ref{eq:respexp2d}). Consider the term $T^{\beta,L,\ell}_{n;\nu}$ in (\ref{eq:errR2d}). We have:
\begin{equation}\label{eq:cpp}
\begin{split}
&\sum_{n= 2}^{\infty}\frac{\theta^{n-1}}{n!L^{n}} \sum_{\{p_{i}\} \in B_{L}^{n}} \sum_{\{y_{i}\} \in \{0, \ldots, L-1\}^{n}} \Big[\prod_{i=1}^{n} \big| \hat\mu_{\alpha,\theta}(-p_{i}, y_{i}) \big| \Big] |T^{\beta,L,\ell}_{n;\nu}(\underline{p}_{1}, \ldots, \underline{p}_{n})| \\
& \leq \sum_{n= 2}^{\infty}\frac{\theta^{n-1}}{n!L^{n}} \sum_{\{p_{i}\} \in B_{L}^{n}} \sum_{\{y_{i}\} \in \{0, \ldots, L-1\}^{n}} \Big[\prod_{i=1}^{n} \big| \hat\mu_{\alpha,\theta}(-p_{i}, y_{i}) \big| \Big] \\
&\quad \cdot \sum_{\omega = 1}^{N} \int_{\beta,L} \frac{d\underline{k}}{(2\pi)^{2}} \sum_{\pi\in S_{n}} \big|f_{n,\omega}^{\nu,\ell}(k;p_{\pi(1)},y_{\pi(1)}, \dots,p_{\pi(n)}, y_{\pi(n)})\big| \prod_{i=1}^{n+1} \Big| g_{\omega}\Big(\ul k+\sum_{j< i} \ul p_{\pi(j)}\Big)\Big|\;;
\end{split}
\end{equation}
%
using (\ref{eq:rembound2}):
\begin{equation}\label{eq:cp}
\begin{split}
&\sum_{\{y_{i}\} \in \{0, \ldots, L-1\}^{n}} \Big[\prod_{i=1}^{n} \big| \hat\mu_{\alpha,\theta}(-p_{i}, y_{i}) \big| \Big] \big|f_{n,\omega}^{\nu,\ell}(k;p_{1},y_{1}, \dots,p_{n}, y_{n})\big| \\
&\quad \leq \sum_{\{y_{i}\} \in \{0, \ldots, L-1\}^{n}} \Big[\prod_{i=1}^{n} \big| \hat\mu_{\alpha,\theta}(-p_{i}, y_{i}) \big| \Big] C^{n} \Big(|k-k_{F}^{\omega}|_{\mathbb{T}} + \sum_{i=1}^{n} |p_{i}|_{\mathbb{T}} \Big) e^{- c\sum_{i=1}^{n} |y_{i}|_{\omega}} \\
&\quad \equiv C^{n} \Big(|k-k_{F}^{\omega}|_{\mathbb{T}} + \sum_{i=1}^{n} |p_{i}|_{\mathbb{T}} \Big) \Big[\prod_{i=1}^{n} \big| \bar{\mu}_{\alpha,\theta,\omega}(-p_{i}) \big| \Big] \;.
\end{split}
\end{equation}
Plugging (\ref{eq:cp}) into (\ref{eq:cpp}), we obtain an expression completely analogous to the one-dimensional case. Proceeding as for (\ref{eq:Tdim}):
\begin{equation}\label{eq:Test2d}
\begin{split}
(\ref{eq:cpp}) &\leq \sum_{n\geq 2}\frac{C^{n} a^{n-1}}{n!} \left| \log \eta\right| (\eta + n \theta) \\
&\leq (1 + n a) e^{Ca} \left| \log \eta\right| \eta\;.
\end{split}
\end{equation}
To conclude, let us estimate the terms $\widetilde{R}^{\beta,L,\ell}_{n;\nu}$ in (\ref{eq:errR2d}). Recall the representation:
\begin{equation}\label{eq:tildeR2d}
\begin{split}
&\widetilde{R}^{\beta,L,\ell}_{n;\nu}(\underline{p}_{1},y_{1}, \ldots, \underline{p}_{n},y_{n}) \\
&= \sum_{\pi\in S_{n}} \sum_{\substack{f\in\{r,s\}^{n+1}\\ f\not\equiv s}} \sum_{z, z'} \int_{\beta,L}\frac{d\underline k}{(2\pi)^{2}}  \Tr\Big[\hat J^{\ell}_{\nu}(k,p_{n+1};z,z')\prod_{i=1}^{n+1} g_{f(i)}\Big(\ul k+\sum_{l< i} \ul p_{\pi(l)}; y_{\pi(i-1)},y_{\pi(i)}\Big)\Big]\;.
\end{split}
\end{equation}
Let $\widetilde R^{\beta,L,\ell}_{n;\nu,j}$ be the contribution to $\widetilde R^{\beta,L,\ell}_{n;\nu}$ with $j$ propagators labelled by $\text{r}$. Let 
\begin{equation}
\tilde g_{\text{r}}(\underline{k}; y_{i}, y_{j}) = (1 + |k_{0}|) g_{\text{r}}(\underline{k}; y_{i}, y_{j})
\end{equation}
and recall:
\begin{equation}
g_{\text{s}}(\underline{k}; y_{i}, y_{j}) = \sum_{\omega} g_{\omega}(\underline{k}) P_{\omega}(k; y_{i}, y_{j})\;.
\end{equation}
Let $m_{0}, m_{1}, \ldots, m_{n-j}$ be the labels such that $f(m_{k}) = \text{s}$, in increasing order. Denoting by $I$ the set of labels associated with the $g_{\text{r}}$ propagators, we can rewrite the trace in (\ref{eq:tildeR2d}) as:
\begin{equation}\label{eq:trace2daa}
\begin{split}
&\Tr\Big[\hat J^{\ell}_{\nu}(k,p_{n+1};z,z')\prod_{i=1}^{n+1} g_{f(i)}\Big(\ul k+\sum_{l< i} \ul p_{\pi(l)}; y_{\pi(i-1)},y_{\pi(i)}\Big)\Big] \\
& = \sum_{\{\omega_{m_{i}}\}} g_{\omega_{m_{0}}}\Big( \ul k+\sum_{l< m_{0}} \ul p_{\pi(l)} \Big) \cdots g_{\omega_{m_{n-j}}}\Big( \ul k+\sum_{l< m_{n-j}} \ul p_{\pi(l)} \Big) \prod_{q \in I} \frac{1}{|k_{0} + q \eta| + 1} \\
&\quad \cdot \Tr\Big[\hat J^{\ell}_{\nu}(k,p_{n+1};z,z')\prod_{i=1}^{n+1} A_{f(i)}\Big(\ul k+\sum_{l< i} \ul p_{\pi(l)}; y_{\pi(i-1)},y_{\pi(i)}\Big)\Big]\;,
\end{split}
\end{equation}
where the sum is over the labels $\omega_{m_{0}},\ldots, \omega_{m_{n-j}}$, and:
\begin{equation}
A_{f(i)}(\ul k) = \left\{ \begin{array}{cc} P_{\omega_{i}}(k) & \text{if $i\in I^{c}$} \\ \tilde g_{\text{r}}(\underline{k}) & \text{if $i \in I$.}  \end{array} \right.
\end{equation}
Using the exponential decay of the edge modes and the exponential decay in $|y_{i} - y_{j}|$ of $\tilde g_{\text{r}}(\underline{k}; y_{i}, y_{j})$, we easily get:
\begin{equation}
\begin{split}
&\sum_{z, z'} \sum_{\{y_{i}\} \in \{0, \ldots, L-1\}^{n}} \Big[\prod_{i=1}^{n} \big| \hat\mu_{\alpha,\theta}(-p_{i}, y_{i}) \big| \Big] \\&\quad \cdot  \Big| \Tr\Big[\hat J^{\ell}_{\nu}(k,p_{n+1};z,z')\prod_{i=1}^{n+1} A_{f(i)}\Big(\ul k+\sum_{l< i} \ul p_{\pi(l)}; y_{\pi(i-1)},y_{\pi(i)}\Big)\Big]\Big| \\
&\quad \leq C^{n} \ell \Big[\prod_{i=1}^{n} \big| \nu_{\alpha,\theta}(p_{i}) \big| \Big]
\end{split}
\end{equation}
for:
\begin{equation}
\nu_{\alpha,\theta}(p) := \frac{1}{\theta} \frac{\chi( \theta^{\alpha-1} |p|)}{1 + | p/\theta |_{\mathbb{T}_{\theta^{-1}}}^{2}}\;.
\end{equation}
Thus, the contribution associated with (\ref{eq:trace2daa}) to (\ref{eq:tildeR2d}) summed over the $\{y_{i}\}$ variables is bounded as:
\begin{equation}\label{eq:591}
\begin{split}
&C^{n} \ell \sum_{\pi\in S_{n}} \int_{\beta, L}\frac{d\underline k}{(2\pi)^{2}} \sum_{\{\omega_{m_{i}}\}} \Big| g_{\omega_{m_{0}}}\Big( \ul k+\sum_{l< m_{0}} \ul p_{\pi(l)} \Big)\cdots g_{\omega_{m_{n-j}}}\Big( \ul k+\sum_{l< m_{n-j}} \ul p_{\pi(l)} \Big) \Big| \prod_{q \in I} \frac{1}{|k_{0} + q \eta| + 1} \\
&\qquad \cdot \Big[\prod_{i=1}^{n} \big| \nu_{\alpha,\theta}(p_{i}) \big| \Big]\;;
\end{split}
\end{equation}
from now on, the discussion proceeds exactly as in the one-dimensional case, see the discussion after (\ref{eq:remj}). In fact, the functions in the product over $q\in I$ in Eq. (\ref{eq:591}) play the role of the bounded propagators in the one-dimensional case. Therefore, recalling (\ref{eq:n>1}), choosing $a\leq w\left|\log\eta\right|$ for $w$ small enough, and for some $\gamma > 0$, the final result is:
\begin{equation}\label{eq:tildeR2d2}
\sum_{n= 2}^{\infty}\frac{\theta^{n-1}}{n!L^{n}} \sum_{\{p_{i}\} \in B_{L}^{n}} \sum_{\{y_{i}\} \in \{0, \ldots, L-1\}^{n}} \Big[\prod_{i=1}^{n} \big| \hat\mu_{\alpha,\theta}(p_{i}, y_{i}) \big| \Big] \big|\widetilde{R}^{\beta,L,\ell}_{\nu;n}(\underline{p}_{1},y_{1}, \ldots, \underline{p}_{n},y_{n})\big| \leq C \ell \eta^{\gamma}\;.
\end{equation}
Eqs. (\ref{eq:Test2d}), (\ref{eq:tildeR2d2}) prove (\ref{eq:hot2dest}), and conclude the proof of Theorem \ref{thm:main2d}. 
\end{proof}

\end{document}